\documentclass[11pt]{book} % document type
\pdfoutput=1
\usepackage{mypackages}
\begin{document}

    % First page  
    \newdateformat{monthyeardate}{\monthname[\THEMONTH], \THEYEAR} % date format

    \begin{center}
        \textbf{Multiple Landau level filling for a large magnetic field limit of 2D fermions} \\[0.5cm]
        %\monthyeardate\today \\[0.5cm]
        Denis \textsc{Périce} \\
        Unité de mathématiques pures et appliquées \\ \'Ecole normale supérieure de lyon \\ Lyon - France \\[1cm]
    \end{center}
    \textbf{Abstract:} \\[0.2cm]
        Motivated by the quantum hall effect, we study $N$ two dimensional interacting fermions in a large magnetic field limit. We work in a bounded domain, ensuring finite degeneracy of the Landau levels. In our regime, several levels are fully filled and inert: the density in these levels is constant. We derive a limiting mean-field and semi classical description of the physics in the last, partially filled Landau level.

    %\tableofcontents

    \section{Context and result}

\subsection{Model}

    We consider a system of $N$ interacting fermionic particles in two dimensions. They are placed under a homogeneous magnetic field perpendicular to the domain. In this context the kinetic energy of the particles is quantized into discrete energy levels called Landau levels, separated by a finite energy gap. This problem has initially been studied by Lieb Solovej and Yngvason in \cite{LSY95}, \cite{LS94}, \cite{LSY94}, \cite{LSY94II}, \cite{LSY92} and more recently by Fournais, Lewin and Madsen in \cite{FM}, \cite{FLS}.

    Our goal is to study the mean field and semi-classical limit under high magnetic field so the Landau level quantization plays an important role. This setup is related to that of \cite{LSY95} where three regimes are studied. In the first one, the energy gap is small with respect to the potential contributions in the energy so particles occupy all Landau levels and a standard Thomas–Fermi model is obtained in the limit. In the second one, the energy gap is comparable to the potential energy terms, particles optimise both their Landau level and their position in the potentials and the limit is a magnetic Thomas-Fermi model. For the last scaling, the gap is large compared to the potential energies so all particles occupy the lowest Landau level and the limit is described with a classical continuum electrostatic theory in this level. We want to deal with the intermediate situation where only a finite number of Landau levels are completely filled. Precisely, our result is a limit where the $q$ first Landau Level are fully filled, the next Landau level is partially filled with filling ratio $r < 1$ and all higher Landau levels are empty. We also provide a model for the physics in the partially filled Landau level. This setup is physically motivated by the quantum Hall effect which mostly takes place in a partially filled Landau level while lower Landau levels are filled and inert, and higher levels are empty, see \cite{Jain}.

   In this perspective we want to fix the limit ratio of the number of particles to the degeneracy of Landau levels. On the whole space $\mathbb{R}^2$ this degeneracy is infinite. To ensure finiteness of the degeneracy of the Landau levels (see \cref{subsec:explicit diag}), we work on a bounded domain. For simplicity, we would like to consider a torus with periodic boundary conditions. But, in the presence of a magnetic field the periodic boundary conditions must be modified. This is a well known issue, for example see \cite[Section 3.9]{Jain}. As explained in  \cref{subsec:tau}, we define magnetic translation operators to ensure commutation with the magnetic momentum. These magnetic translations operators define the so called magnetic periodic boundary conditions.

    \begin{notation}[title=Model, label=model]
        We work on the domain $\Omega \coloneq [0,L]^2$ of fixed size $L > 0$. The one body kinetic energy operator, also called magnetic Laplacian, is
        \begin{align*}
            \mathscr{L}_{\hbar, b} \coloneq \prth{i\hbar \nabla + bA}^2 \yesnumber{L hbar bA}
        \end{align*}
        We work in the Coulomb gauge:
        \begin{align*}
            \nabla \cdot A = 0
        \end{align*}
        where $A\in C^\infty\prth{\mathbb{R}^2, \mathbb{R}^2}$ is the vector potential. Identifying $\mathbb{R}^2$ with $\mathbb{R}^2\times \{0\} \subset \mathbb{R}^3$, we assume
        \begin{align*}
            \nabla \wedge A = (0,0,1) \yesnumber{B}
        \end{align*}
        $b$ is the magnetic field amplitude with associated magnetic length 
        \begin{align*}
            l_b \coloneq\sqrt{\dfrac{\hbar}{b}}
        \end{align*}
        We identify $\mathbb{R}^2$ with $\mathbb{C}$ and use complex notation for the variables $(x, y) \in \mathbb{R}^2$ namely
        \begin{align*}
            (x, y) = x+iy \in \mathbb{C}    
        \end{align*}
        Let $z_0\in \mathbb{C}$, by \cref{eq:B}
        \begin{align*}
            \nabla \wedge(A - A(\bullet  - z_0)) = 0
        \end{align*}
        so we can choose $\varphi_{z_0} \in C^\infty\prth{\mathbb{R}^2, \mathbb{R}}$ such that 
        \begin{align*}
            A - A(\bullet - z_0) \eqcolon l_b^2 \nabla \varphi_{z_0} \yesnumber{magnetic phase}
        \end{align*}
        For some usual expressions see \cref{eq:sym gauge} and \cref{eq:Lan gauge functions}. As detailed in \cref{subsec:tau}, for $\psi\in L^2(\Omega)$ the magnetic periodic boundary conditions are
        \begin{align*}
            \forall t \in [0, L], \, \syst{\psi(L + it) = e^{i\varphi_L(L + it)}\psi (it)\\
            \psi(t + iL) = e^{i\varphi_{iL}(t + iL)}\psi(t)} \yesnumber{magn per}
        \end{align*}
        and the domain of the magnetic Laplacian is
        \begin{align*}
            \text{Dom}\prth{\mathscr{L}_{\hbar,bA}} \coloneq \sett{ \psi \in H^2(\Omega) \text{ such that \cref{eq:magn per} holds}} \yesnumber{Dom L}
        \end{align*}
        Now, the $N$-body Hamiltonian is
        \begin{align*}
            \mathscr{H}_{N} \coloneq \sum_{j=1}^N \left(\prth{i\hbar \nabla_{x_j} + bA(x_j) } ^2 + V(x_j) \right) + \dfrac{2}{N-1}\sum_{1\leq j < k \leq N} w(x_j - x_k) \yesnumber{H_N}
        \end{align*}
        acting on the space of $N$-body fermionic states
        \begin{align*}
            L^2_- \prth{\Omega^{N}} \coloneq \bigwedge ^N L^2 (\Omega).
        \end{align*}
        We denote $\mathbb{T} \coloneq \mathbb{R}^2/L\mathbb{Z}^2$. $V \in L^2(\mathbb{T})$ is the external potential and $w \in L^2(\mathbb{T})$ the interaction potential assumed to be radial for the metric on the torus:
        \begin{align*}
            w(x - y) \eqcolon \widetilde{w}(d(x, y)) \text{ with } d(x, y) \coloneq \minm\displaylimits_{r \in L\mathbb{Z}^2}\abs{x - y + r}
        \end{align*}
        The domain of the $N$-body Hamiltonian \cref{eq:H_N} is
        \begin{align*}
            \text{Dom}\prth{\mathscr{H}_N} \coloneq \bigwedge^N \text{Dom}\prth{\mathscr{L}_{\hbar,b}}
        \end{align*}
        We define the $N$-body ground state energy
        \begin{align*}
            E_N^0 \coloneq \text{inf}\sett{\braket{\psi_N | \mathscr{H}_N \psi_N}, \psi_N \in \text{Dom}(\mathscr{H}_N) \text{ such that } \norm{\psi_N}_{L^2}=1} \yesnumber{E_N^0}
        \end{align*}
    \end{notation}
    
    There are $N(N-1)/2$ interacting pairs of fermions. Thus, we divide the interactions term by $(N-1)/2$ so that the order of the contribution coming from interactions is $\mathcal{O}(N)$ and comparable to the contribution coming from the external potential.
    
    As we will see in \cref{subsec:LL quant}, the self adjointness of the magnetic Laplacian and the existence of its eigenvectors require the magnetic field flux $bL^2$ going through the domain to be quantized in multiples of $2\pi\hbar$:
    \begin{align*}
        \exists d \subset \mathbb{N} \text{ such that } 2\pi d = \dfrac{b}{\hbar}L^2  = \dfrac{L^2}{l_b^2}
    \end{align*}
    In \cref{subsec:explicit diag} we will explain that $d$ is the degeneracy of Landau levels. Now, we can fix the number of filled Landau levels by choosing a scaling for which the ratio $N/d$ is fixed.
    
    \begin{notation}[title=Scaling, label=scaling]
        We take Planck's constant to be a sequence $\hbar \coloneq \prth{\hbar_N}_{N\in\mathbb{N}}$ such that
        \begin{align*}
            N^{-\frac{1}{2}} \ll \hbar \ll N^{-\frac{1}{4}} \yesnumber{hbar}
        \end{align*}
        Let $q \in \mathbb{N}, r \in [0,1), b \coloneq (b_N)_{N \in \mathbb{N}}$ be such that
        \begin{align*}
            d \coloneq \dfrac{L^2}{2\pi l_b^2} \subset \mathbb{N^*} \yesnumber{flux quant}
        \end{align*}
        and
        \begin{align*}
            \frac{N}{d} \eq_{N\to\infty} q + r + o\prth{\frac{1}{\hbar b}} \yesnumber{speed of LL conv}
        \end{align*}
        where $E^* \coloneq E\backslash{\sett{0}}$ for $E \subset \mathbb{R}$.
    \end{notation}
    
    $q$ will give the number of fully filled Landau levels and $r$ the filling ratio of the $(q+1)^{th}$ Landau level. Note that the lowest Landau level index is $0$ in our convention. With this notation,
    \begin{align*}
        \dfrac{N}{d} = \dfrac{2\pi l_b^2 N}{L^2} \limit\displaylimits_{N\to\infty} q+r \text{~ and ~} \dfrac{1}{l_b^2} = \dfrac{b}{\hbar} \eqvl\displaylimits_{N\to\infty} \dfrac{2\pi N}{(q+r)L^2} \yesnumber{l_b scaling}
    \end{align*}
    With this scaling, we find that the order of the magnetic field is $b = \mathcal{O}(\hbar N)$. It is known \cref{eq:H to N}, that the order or the kinetic energy is
    \begin{align*}
        \hbar b = \mathcal{O}(\hbar^2 N) \yesnumber{O hbar b} \gg 1
    \end{align*}
    The kinetic energy contribution needs to be of leading order compared to the potential terms if we want to impose the number of filled Landau level and this is true if and only if
    \begin{align*}
        \hbar^2 N \gg 1
    \end{align*}
    hence the lower bound in \cref{eq:hbar}. The upper bound $\hbar \ll N^{-\frac{1}{4}}$ is necessary in our approach to control some error terms coming from the kinetic energy. This is also the reason why we impose the convergence rate in \cref{eq:speed of LL conv}. This scaling is a semi-classical limit because Planck's constant goes to $0$.
    
\subsection{Semi-classical limit model}
        
    In the limit, we obtain a semi-classical model where the energy no longer depends on the wave-function but on the density in phase space. This comes with a non linearity in the interaction term. The phase space is $\mathbb{N}\times\Omega$. This means that particles have two degrees of freedom: the first one is $n\in \mathbb{N}$ the quantum number representing the Landau Level index and $x\in \Omega$ representing the position of particles in space. In classical mechanics, one can think of $x$ as the center of the cyclotron orbit of the particles and $n$ as the index of the quantized angular velocity of the cyclotron orbit. This model is semi-classical in the sense that the Pauli principle still holds as a bound on the density.
    
    \begin{notation}[title=Semi-classical functional, label=sc functional]
        We consider the measure on phase space
        \begin{align*}
            \eta \coloneq \prth{ \sum_{n\in\mathbb{N}}\delta_n }\otimes \lambda_\Omega 
        \end{align*}
        where $\lambda_\Omega $ is the Lebesgue measure on $\Omega$. For $m \in L^1\prth{\mathbb{N}\times  \Omega, \mathbb{R}_+}$, define
        \begin{align*}
            \mathcal{E}_{sc, \hbar b}\sbra{m} \coloneq \intr_{\mathbb{N}\times\Omega} E_n m(n, R)d\eta(n, R) + \intr_{\mathbb{N}\times\Omega}  V m d\eta + \intr_{(\mathbb{N}\times\Omega)^2}  w m^{\otimes 2} d\eta^{\otimes 2} \yesnumber{SC energy}
        \end{align*}
        where, as we will see in \cref{sec:II quantization},
        \begin{align*}
            E_n \coloneq 2\hbar b\left(n + \dfrac{1}{2}\right) \yesnumber{E_n}    
        \end{align*}
        is the energy of the $n^{th}$ Landau level. Define the semi-classical domain
        \begin{align*}
            \mathcal{D}_{sc} \coloneq \sett{m \in L^1\prth{\mathbb{N}\times  \Omega} \text{ such that } \intr_{\mathbb{N}\times \Omega}m d\eta =1 \text{ and } 0\le m\le \dfrac{1}{(q+r)L^2}} \yesnumber{sc domain}
        \end{align*}
        and the semi-classical ground state energy
        \begin{align*}
            E_{sc, \hbar b}^0 \coloneq \inf_{m\in\mathcal{D}_{sc}} \mathcal{E}_{sc, \hbar b}\sbra{m}
        \end{align*}
    \end{notation}

    We then define the model for the partially filled Landau level that only depends on the density.
    
    \begin{notation}[title=Electrostatic model for the partially filled level, label=sc density]
        Define
        \begin{align*}
            \mathcal{E}_{qLL}[\rho] \coloneq  \intr_{\Omega} V\rho  + \iint\displaylimits_{\Omega^2} w(x-y)\rho(x)\rho(y)dxdy \yesnumber{qLL}
        \end{align*}
        with domain
        \begin{align*}
            \mathcal{D}_{qLL} \coloneq \sett{\rho \in L^1(\Omega) \text{ such that } \intr_\Omega \rho = \dfrac{r}{q+r} \text{ and }0 \le \rho \le \dfrac{1}{(q+r)L^2}} \yesnumber{rho}
        \end{align*} 
        The associated ground state energy is
        \begin{align*}
            E_{qLL}^0 \coloneq \inf_{ \mathcal{D}_{qLL}}\mathcal{E}_{qLL}
        \end{align*}
        We define the following energies:
        \begin{align*}
            E^{q,r} \coloneq \dfrac{q^2+2qr+r}{q+r}  \qquad
            E_{V}^{q,r} \coloneq \dfrac{q}{(q+r)L^2}\int\displaylimits_\Omega  V \qquad
            E_{w}^{q,r} \coloneq \dfrac{q^2 + 2qr}{(q+r)^2 L^4}\int\displaylimits_{\Omega^2} w
        \end{align*}
            Let $\rho\in\mathcal{D}_{qLL}$, define
        \begin{align*}
            m_\rho(n, x) \coloneq \mathbb{1}_{n < q}\dfrac{1}{L^2(q+r)} + \mathbb{1}_{n=q}\rho(x) \yesnumber{saturated low LL}
        \end{align*}
    \end{notation}
    
    $m_\rho$ is a phase space density constructed with the $qLL$ lowest Landau levels saturating the Pauli principle in \cref{eq:sc domain} and \cref{eq:rho} and with the density $\rho$ in the partially filled Landau level. The ratio of particles in the partially filled Landau level is 
    \begin{align*}
        \dfrac{r}{q+r}
    \end{align*}
    This corresponds to the normalization constraint in \cref{eq:rho}. With this we see that the Pauli principle is indeed the correct bound on the densities to have
    \begin{align*}
        \intr_{\mathbb{N}\times\Omega}m_\rho d\eta = 1
    \end{align*}
        A direct computation shows that
    \begin{align*}
        \mathcal{E}_{sc, \hbar b}\sbra{m_\rho} =  \hbar b E^{q,r} + E_{V}^{q,r} + E_{w}^{q,r} + \mathcal{E}_{qLL} \sbra{\rho} \yesnumber{sat LLL}
    \end{align*}
    where
    \begin{itemize}
        \item $\hbar b E^{q,r}$ is the kinetic energy contribution from the $q+1$ lowest Landau levels
        \item $E_{V}^{q,r}$ is the external potential energy contribution from the $q$ lowest Landau levels
        \item $E_{w}^{q,r}$ is the energy contribution from interactions between the $q$ lowest Landau levels and the interactions between the $q$ lowest Landau levels and the $\prth{q+1}^{th}$ Landau level. In other words, it contains all the interactions except those inside the partially filled level.
    \end{itemize}

    If $w$ has a non-negative fourier transform, this functional is convex and thus has a unique minimizer. For example, with $V=0$, the minimizer is
    \begin{align*}
        \rho_0 \coloneq \dfrac{r}{(q+r)L^2}
    \end{align*}

\subsection{Main results}

    We can now state our main result:
    
    \begin{theorem}[title=Mean field limit with magnetic periodic conditions, label=main]
        With the definitions introduced in \cref{not:model}, \cref{not:scaling}, \cref{not:sc functional}, \cref{not:sc density},
        \begin{align*}
            \dfrac{E_N^0}{N} \eq_{N\to\infty} \hbar b E^{q,r} + E_{V}^{q,r} + E_{w}^{q,r} + E_{qLL}^0 + o(1)
        \end{align*}
    \end{theorem}
    
    This means that in the limit, the first order in the quantum many body energy per particle is the trivial energy $\hbar b E^{q,r}$. Then for terms of order 1, the only non trivial contribution to the energy are the external potential term and the interaction term inside the partially filled Landau level. The lower Landau levels are totally filled and therefore their contribution to the energy is constant. The interaction of the partially filled level with all other level will also be a constant. For higher Landau levels, their contribution to the energy is null because they are totally empty.
    
    The regularity assumptions on the potentials are not minimal, we expect this result to hold true if potentials have an $L^1$ positive part and an $L^2$ negative part. Under these assumptions, one needs to prove that the particles will not concentrate in the $L^1$ positive singularities of the potentials. This has been done in \cite{LSY95} for the repulsive $1/\abs{x}$ Coulomb potential. We will not deal with this issue in this paper.
    
    The number of variables of the densities is going to infinity in our limit, thus we look at reduced densities.

    \begin{notation}[title=Reduced densities, label=reduced densities]
        We denote $\mathcal{L}^p$ the set of $p$-Schatten class operators along with $\norm{\bullet}_{\mathcal{L}^p}$ the $p$-Schatten norm. Let $\gamma_N \in \mathcal{L}^1\prth{L^2_- \prth{\Omega^N}}$ a positive operator of trace 1. We call such an operator an $N$-body density matrix. We will denote in the same way operators and their integral kernel. We introduce compact notation for lists:
        \begin{align*}
            1:n \coloneq (1, \dots, n) \qquad
            x_{1:n} \coloneq (x_1, \dots, x_n) 
        \end{align*}
        The density associated to $\gamma_N$ is
        \begin{align*}
            \rho_{\gamma_N}(x_{1:N}) \coloneq \gamma_N(x_{1:N}, x_{1:N})
        \end{align*}
        Let $\text{Tr}_I$ be the partial trace that traces out coordinates in $I \subset \intint{1, N}$ of $L^2(\Omega)^{\otimes N}$, it is defined as the linear operator on $\mathcal{L}^1\prth{L^2\prth{\Omega^N}}$ satisfying
        \begin{align*}
            \forall A_{1:N} \in \mathcal{L}^1\prth{L^2(\Omega)}, \text{Tr}_I\sbra{\bigotimes_{i=1}^N A_i} \coloneq \Tr{\bigotimes_{i\in I} A_i} \bigotimes_{i \notin I}A_i
        \end{align*}
        Let $1\le k < N$, we define the $k^{th}$ reduced density matrix associated to $\gamma_N$ by
        \begin{align*}
            \gamma_N^{(k)} = \text{Tr}_{k+1:N}\sbra{\gamma_N} \yesnumber{red density}
        \end{align*}
        with the convention that $\gamma_N^{(N)} \coloneq \gamma_N$. For an $N$ variables symmetric density $\rho_N$ we denote $\rho_N^{(k)}$ its $k^{th}$ marginal. 
        If one starts from $\psi_N \in L^2_-\prth{\Omega^N}$ we use the notation
        \begin{align*}
            &\gamma_{\psi_N} \coloneq \ket{\psi_N}\bra{\psi_N} \\
            &\rho_{\psi_N} \coloneq \rho_{\gamma_{\psi_N}} = \abs{\psi_N}^2 \yesnumber{lists}
        \end{align*}
    \end{notation}
    
    Note that with this notation
    \begin{align*}
        \rho_{\gamma_N^{(k)}} = \rho_{\gamma_N}^{(k)} \yesnumber{consistent not}
    \end{align*}
        
    The domain $\Omega$ is a locally compact metric space, the set of Radon measures on it is the dual of continuous and compactly supported functions
    \begin{align*}
        \mathcal{M}(\Omega) = C_c^0(\Omega)^*
    \end{align*}
    In our case, since $\Omega$ is compact this is just the dual of continuous functions. We denote $\mathcal{M}_+(\Omega)$ the set of positive Radon measures. Let $\mathcal{P}(\Omega)$ be the set of probabilities on $\Omega$:
    \begin{align*}
        \mathcal{P}(\Omega) \coloneq \sett{\mu \in \mathcal{M}_+(\Omega) \text{ such that } \mu(\Omega) = 1}
    \end{align*}
    On this space the weak star topology is metrizable using a Wasserstein metric \cite[Section 7.12]{Villani}. Moreover $\mathcal{P}(\Omega)$ is also locally compact \cite[section 17.E]{Kechris}, thus it is possible to iterate and define the space of probability measures on $\mathcal{P}(\Omega)$ namely $\mathcal{P}\prth{\mathcal{P}(\Omega)}$.
    
    Now, we have the following theorem for the convergence of reduced densities:
    
    \begin{theorem}[title=Densities convergence with magnetic periodic conditions, label=density conv]
        With the definitions introduced in \cref{not:model}, \cref{not:scaling}, \cref{not:sc functional}, \cref{not:sc density}, \cref{not:reduced densities}, if $(\psi_N)$ is a sequence of minimizers of \cref{eq:E_N^0}, then $\exists \mu \in \mathcal{P}\prth{\mathcal{D}_{qLL}}\, \text{such that} \, $
        \begin{itemize}
            \item $\mu$ only charges minimizers of the limit energy functional \cref{eq:qLL}
            \item  $\forall k \in \mathbb{N}^*$, in the sense of Radon measures
                \begin{align*}
                    \rho_{\psi_N}^{(k)} &\wslim\displaylimits_{N\to\infty} \intr_{\mathcal{D}_{qLL}} \prth{\dfrac{q}{L^2(q+r)} + \rho}^{\otimes k}  d\mu (\rho) \yesnumber{th conv rho}
                \end{align*}
        \end{itemize}
    \end{theorem}
    
    The density of particles converge to a convex combination of densities of the form
    \begin{align*}
        \dfrac{q}{L^2(q+r)} + \rho    
    \end{align*}
    From the Pauli principle in \cref{eq:rho} we see that the constant term in this expression corresponds to particles in the $q$ lowest and fully filled Landau levels. Then the density of particles in the partially filled Landau level is given by a minimizer $\rho$ of the limit functional \cref{eq:qLL}.

\subsection{Scaling}

    Another way to obtain the scaling in \cref{not:scaling} is to observe that we have two characteristic length-scales:
    \begin{itemize}
        \item $\dfrac{L}{\sqrt{N}}$, measuring the mean distance between particles
        \item $l_b$, the magnetic length, which, in classical mechanics corresponds to the radius of a cyclotron orbit. Due to the Pauli principle, $l_b$ will be the order of the minimum distance between particles inside a Landau level. More precisely the Pauli principle takes the form of an upper bound on the density in phase space.
    \end{itemize}
    The square ratio of these length is
    \begin{align*}
        \dfrac{L^2}{Nl_b^2} =  \dfrac{b L^2}{\hbar N} \yesnumber{carac length ratio}
    \end{align*}
    If this ratio goes to zero, the mean distance between particles is very small compared to the minimal length-scale between two particles in a fixed Landau level. This implies that the particles must fill many Landau levels and this corresponds to the scaling in \cite{LSY95} where the energy gap between Landau level is small compared to the potential terms.
    
    If this ratio goes to infinity, the mean distance between particles is very large compared to the minimal length-scale between two particles in a fixed Landau Level. As a consequence, all particles can be placed in the lowest Landau level and this corresponds to the regime in \cite{LSY95} where particles only occupy the lowest Landau level and do not feel the Pauli principle.
    
    In the limit we study, we see from \cref{eq:l_b scaling} that the ratio \cref{eq:carac length ratio} has been fixed to be
    \begin{align*}
        \dfrac{L^2}{Nl_b^2} \limit\displaylimits_{N\to\infty} \dfrac{2\pi}{q+r} \yesnumber{length ratio}
    \end{align*}
    in order to fill a finite number of Landau levels. In our limit we fixed $L$ and took $l_b$ going to zero, but one can also ensure \cref{eq:length ratio} by fixing a magnetic length $\widetilde{l_b} > 0$ and taking a domain length $\widetilde{L}$ going to infinity as
    \begin{align*}
        \widetilde{L} \coloneq \dfrac{\widetilde{l_b}}{l_b} L \yesnumber{scaling L big}
    \end{align*}
    In this limit the density of particles in the domain $\Omega$ is fixed:
    \begin{align*}
        \dfrac{\widetilde{L}^2}{N} \limit\displaylimits_{N\to\infty} \widetilde{l_b}^2 \dfrac{2\pi}{q+r} \yesnumber{thermo limit}
    \end{align*}
    Those limits are equivalent in the sense that the $N$-body Hamiltonian \cref{eq:H_N} is unitarily equivalent to 
    \begin{align*}
        \dfrac{1}{\tau} \mathscr{H}_{N, \tau} \coloneq \dfrac{1}{\tau}\prth{\sum_{j=1}^N \prth{\prth{i\hbar_\tau \nabla_{j} +b_\tau A_\tau (x_j)} ^2 + V_\tau (x_j) } + \dfrac{2}{N-1}\sum_{1\leq j < k \leq N} w_\tau (x_j - x_k)} \yesnumber{HN tau}
    \end{align*}
    where
    \begin{align*}
        \hbar_\tau \coloneq \dfrac{\hbar}{\sqrt{\tau}} \qquad 
        A_\tau \coloneq \dfrac{1}{\tau} A\prth{\tau\bullet} \qquad 
        b_\tau  \coloneq \tau^{\frac{3}{2}} b \qquad
        V_\tau \coloneq \tau V\prth{\tau\bullet} \qquad
        w_\tau \coloneq \tau w\prth{\tau\bullet}
    \end{align*}
    Taking, $\tau \coloneq L/\widetilde{L}$ and using \cref{eq:HN tau}, we confirm that the new magnetic length is
    \begin{align*}
        \sqrt{\dfrac{\hbar_\tau}{b_\tau}} = \dfrac{\widetilde{L} l_b}{L} = \widetilde{l_b}
    \end{align*}
    Moreover if one chooses 
    \begin{align*}
        A = A_{Lan} \qquad V(x) = \varrho * \dfrac{1}{\abs{x}} \qquad w(x) = \dfrac{1}{\abs{x}}
    \end{align*}
    then the vector potential and the interaction potential are not rescaled :
    \begin{align*}
        A_\tau = A \qquad w_\tau = w
    \end{align*}
    If we assume that the external potential is generated by a background charge density $\varrho \in L^1(\Omega)$ it transforms as
    \begin{align*}
        & V_\tau (x) = \intr_\Omega \tau \varrho(\tau x - y)\dfrac{1}{\abs{y}}dy =  \intr_{\sbra{0, \widetilde{L}}^2} \tau^2  \varrho(\tau\prth{y - x})\dfrac{1}{\abs{y}}dy \eqcolon \varrho_\tau * \dfrac{1}{\abs{x}}
    \end{align*}
    The re-scaling preserves the total charge
    \begin{align*}
       \intr_{\sbra{0, \widetilde{L}}^2} \varrho_\tau dx = \intr_\Omega \varrho dx
    \end{align*}
    and
    \begin{align*}
        \mathscr{H}_{N, \tau} =  \sum_{j=1}^N \prth{\prth{i\hbar_\tau \nabla_{j} +b_\tau A (x_j)} ^2 + \rho_\tau * \dfrac{1}{\abs{x}}} + \dfrac{2}{N-1}\sum_{1\leq j < k \leq N} w (x_j - x_k) 
    \end{align*}
    We conclude that our initial scaling is equivalent to a thermodynamic limit.
    
\subsection{Organisation of the paper}
    
    The next two sections contain preparations and necessary tools. \cref{sec:II quantization} is about the diagonalisation of the magnetic Laplacian \cref{eq:L hbar bA}. In \cref{sec:Projector} we define the orthogonal projection on Landau levels and localise it in space, this will be the central object in the definition of the semi-classical densities. Then we prove a Lieb-Thirring inequality in \cref{sec:LT} to deal with $L^2$ potentials. The last two sections contain the proof of \cref{th:main} and \cref{th:density conv}. In \cref{sec:SC} we justify the semi-classical approximation and express the energy in terms of semi-classical densities. Finally, in \cref{sec:MF} we prove the mean-field approximation giving an upper and a lower energy bound. \cref{sec:appendix} contains technical lemmas whose proofs can safely be skipped by the reader.
    \section{Quantization} \label{sec:II quantization}

    In this Section, we recall the diagonalization of the magnetic Laplacian \cref{eq:L hbar bA}. We construct an orthonormal basis of $L^2(\Omega)$ adapted to the Landau levels in terms of magnetic periodic eigenstates of $\mathscr{L}_{\hbar, b}$ \cref{eq:psi_nl}. This fact is already well known in the literature, see \cite{AS}, \cite{Almog} or in \cite[section 3.9]{Jain}. Thus the reader may go directly to \cref{eq:psi_nl}. Then, we will prove another expression for the eigenfunctions in \cref{prop:poisson wave func} using the Poisson summation formula.

\subsection{Magnetic translation operators}\label{subsec:tau}

    In this subsection we explain the definition of the boundary conditions \cref{eq:magn per}.
    
    \begin{notation}
        Let $z_0 \in \mathbb{C}$, define the translation operator on $u \in L^2_{loc}\prth{\mathbb{R}^2}$ by
        \begin{align*}
            T_{z_0} u \coloneq u(\bullet - z_0)
        \end{align*}
        We define the magnetic translation operators as
        \begin{align*}
            \tau_{z_0} \coloneq e^{i\varphi_{z_0}}T_{z_0} \yesnumber{tau}
        \end{align*}
        They define the conditions \cref{eq:magn per} on $\partial\Omega$. Let $k \ge 1$, the magnetic periodic Sobolev space is
        \begin{align*}
            H^k_{mp}(\Omega) \coloneq \sett{ \psi \in H^k(\Omega) \text{ such that \cref{eq:magn per} hold}}
        \end{align*}
        We will use similar notation for other usual functional spaces where the subscript $mp$ stands for magnetic periodic and $p$ for periodic. The domain of the magnetic momentum
        \begin{align*}
            \mathscr{P}_{\hbar, b} \coloneq i\hbar \nabla + bA
        \end{align*}
        is
        \begin{align*}
            \text{Dom}(\mathscr{P}_{\hbar, b}) \coloneq H^1_{mp}(\Omega) \qquad
        \end{align*}
        In Coulomb gauge, there exists $\phi \in C^\infty\prth{\mathbb{R}^2, \mathbb{R}}$ such that the vector potential satisfies
        \begin{align*}
            A = \nabla^\perp \phi \coloneq \left(\begin{matrix} -\partial_y \phi \\ \partial_x \phi \end{matrix}\right) \yesnumber{A}
        \end{align*}
    \end{notation}
    
    For $k>1, H^k(\Omega) \hookrightarrow C^0(\Omega)$, so the conditions \cref{eq:magn per} are well defined. For $k=1$ they are defined with the trace operator $T$ and $\psi_{|\Omega} \coloneq T \psi$.
    
    For some examples of Coulomb gauges, one can take the symmetric gauge:
    \begin{align*}
        \phi_{sym} \coloneq \dfrac{|z|^2}{4} \qquad
        A_{Sym} \coloneq \dfrac{1}{2}(x, y)^\perp \coloneq \dfrac{1}{2}(-y, x) \qquad
        \varphi_{z_0, sym} \coloneq \dfrac{x_0 y - y_0 x}{2l_b^2} \yesnumber{sym gauge}
    \end{align*}
    or the Landau gauge:
    \begin{align*} 
        \phi_{Lan} \coloneq \dfrac{y^2}{2} \qquad
        A_{Lan} = (-y, 0) \qquad
        \varphi_{z_0, Lan} \coloneq -\dfrac{y_0 x}{l_b^2} \yesnumber{Lan gauge functions} 
    \end{align*}
    If we insert the Landau gauge \cref{eq:Lan gauge functions} in \cref{eq:magn per} we get the boundary conditions in Landau gauge:
    \begin{align*}
        \forall t \in [0, L], \, \syst{\psi(L + it) = \psi (it) \\ \psi(t + iL) = e^{-i\tfrac{Lt}{l_b^2}}\psi(t)} \yesnumber{tau Lan}
    \end{align*}
    We here emphasize the importance of the flux quantization \cref{eq:flux quant}. We are able to impose magnetic periodic boundary conditions in both directions if and only if the flux is quantized:
    \begin{align*}
        [\tau_L, \tau_{iL}] = 0 \iff \dfrac{L^2}{l_b^2} \in 2\pi \mathbb{Z}
    \end{align*}
    and in this case,
    \begin{align*}
        \forall k \ge 1, H^k_{mp}(\Omega) = \sett{\psi_{|\Omega}, \psi \in H^k_{loc}\prth{\mathbb{R}^2} \text{ such that } \tau_L \psi = \tau_{iL} \psi = \psi} \yesnumber{Hkmp identification}
    \end{align*}
    In view of \cref{eq:Hkmp identification} we will identify $\psi \in H^k_{mp}(\Omega)$ and its extension on $\mathbb{R}^2$ from now on. The magnetic Laplacian \cref{eq:L hbar bA} commutes with the magnetic translations defined in \cref{eq:magnetic phase}, \cref{eq:tau}:
    \begin{align*}
        [\mathscr{P}_{\hbar, b}, \tau_{z_0}] = 0 \text{ on } H^1_{mp}(\Omega) \text{ and } [\mathscr{L}_{\hbar, b}, \tau_{z_0}] = 0 \text{ on } H^2_{mp}(\Omega)
    \end{align*}

\subsection{Landau Level quantization} \label{subsec:LL quant}
    
    This subsection is devoted to the usual formalism for the description of the magnetic Laplacian in terms of annihilation and creation operators. More details about these operators and the properties of Landau levels can be found in \cite{RY}.
    
    \begin{notation}
    We denote by $\pi_x, \pi_y$ the coordinates of the magnetic momentum:
        \begin{align*}
            \mathscr{P}_{\hbar, b} \eqcolon \left( \begin{matrix} i\hbar\partial_x + bA_x \\ i \hbar \partial_y + bA_y \end{matrix} \right) \eqcolon \left( \begin{matrix} \pi_x \\ \pi_y \end{matrix} \right)
        \end{align*}
        and define the annihilation and creation operators respectively as
        \begin{align*}
            a \coloneq \dfrac{\pi_y - i \pi_x}{\sqrt{2 \hbar b}} \hspace{1cm} a^\dagger \coloneq \dfrac{\pi_y + i \pi_x}{\sqrt{2 \hbar b}} \yesnumber{creation op}
        \end{align*}
        and the number of excitation operator $\mathcal{N} \coloneq a^\dagger a$.
    \end{notation}
    
    The quantization of the magnetic Laplacian comes from the following commutation relations:
    \begin{align*}
        &\sbra{\pi_x, \pi_y} = i\hbar b  \yesnumber{com_pi}\\
        &\sbra{a, a^\dagger} = \text{Id} \text{ (canonical commutation relation)} \\
        &\sbra{\tau_{z_0}, a} = \sbra{\tau_{z_0}, a^\dagger} = 0 \yesnumber{comm tau a}
    \end{align*}
    and
    \begin{align*}
        \mathscr{L}_{\hbar, b} = 2\hbar b\left(\mathcal{N} + \dfrac{\text{Id}}{2}\right) \yesnumber{H to N}
    \end{align*} 
    $\mathscr{L}_{\hbar, b}$ defines the Sobolev space $\prth{\text{Dom}\prth{\mathscr{L}_{\hbar, b}}, \braket{\bullet}_\mathscr{L}}$ with
    \begin{align*}
        \braket{\chi | \psi}_\mathscr{L} \coloneq \braket{\mathscr{L}_{\hbar, b}\chi|\psi}
    \end{align*}
    which is equivalent to $(H^2_{mp}(\Omega), \braket{\bullet}_{H^1})$. The quadratic form defined by $\braket{\bullet}_\mathscr{L}$ is continuous, Hermitian and coercive on $\text{Dom}\prth{\mathscr{L}_{\hbar, b}}$. Thus $\mathscr{L}_{\hbar, b}$ is a closed positive self-adjoint operator and the embedding 
    \begin{align*}
        \text{Dom}\prth{\mathscr{L}_{\hbar, b}} \hookrightarrow L^2(\Omega)
    \end{align*}
    is continuous and compact.
    
    $H^2_{mp}(\Omega)$ contains the smooth and compactly supported functions, so it is dense in $L^2( \Omega)$. We can conclude using the Lax-Milgram theorem \cite[Corollary 4.26]{CR} that the resolvent of $\mathscr{L}_{\hbar, b}$ is well defined and compact. Applying the spectral theorem to the resolvent of $\mathscr{L}_{\hbar, b}$ proves that its spectrum is punctual and $L^2(\Omega)$ is a Hilbertian direct sum of eigenspaces of $\mathscr{L}_{\hbar, b}$. The same conclusions also holds for the $N$-body Hamiltonian \cref{eq:H_N} since the magnetic Laplacian is of dominant order in it. $\mathcal{N}$ inherits the properties of $\mathscr{L}_{\hbar, b}$ and it is well known that
    \begin{align*}
        \text{sp}(\mathcal{N})=\mathbb{N}
    \end{align*}
    
    \begin{notation}[title=Landau levels]
        We define the $n^{th}$ Landau level as the eigenspace associated to $n \in \mathbb{N}$: 
        \begin{align*}
            \text{nLL} \coloneq \sett{\psi \in \text{Dom}\prth{\mathscr{L}_{\hbar, b}} \text{ such that } \mathcal{N} \psi =  n\psi}
        \end{align*}
        The ground level, denoted $\text{LLL}$ for \textit{Lowest Landau Level} has energy $E_0 = \hbar b$.
    \end{notation}
    
    The Landau levels are isomorphic and the operator $a^\dagger/ \sqrt{n+1}$ is a unitary map from $\text{nLL}$ to $\text{(n+1)LL}$ of inverse $a/\sqrt{n+1}$ \cite{RY}. This means that one can generate a basis of any Landau level with successive applications of $a^\dagger$ on basis elements of $\text{LLL}$.

\subsection{Expression of the eigenfunctions}\label{subsec:explicit diag}

    It is well known \cite{RY} that
    \begin{align*}
        \text{LLL} \subset \text{ker}(a) \subset \mathcal{O}(\Omega) e^{-\frac{\phi}{l_b^2}}
    \end{align*}
    where $\mathcal{O}(\Omega)$ is the set of holomorphic functions and $\phi$ is defined in \cref{eq:A}. This proves that the zeros of $\psi \in \text{LLL}$ are given by the zeros of a holomorphic function. Since zeros of a holomorphic function must be isolated, the compactness of the domain implies that eigenfunctions have a finite number of zeros. Actually, this number of zeros is $d$ defined in \cref{eq:flux quant}, and therefore independent of the state, see \cite[section 1]{AS} as a reference. 
      
    \begin{notation}[title=Theta functions]
        Let $\tau$ be a complex parameter in the upper half plane, we define
        \begin{align*}
            \theta(z, \tau) \coloneq \sum_{k \in \mathbb{Z}} e^{i\pi \tau k^2 + 2i\pi k z}  
        \end{align*}
    \end{notation}
    
    Following the proof of \cite{Almog} or \cite[Chapter V Theorem 8]{Chandrasekharan}, the Landau levels have a finite degeneracy and
    \begin{align*}
       \forall n \in \mathbb{N}, \text{Dim(nLL)} = d
    \end{align*}
    With a direct computation we can express the eigenfunctions of the magnetic Laplacian \cref{eq:L hbar bA} in term of theta functions. The following family, indexed by $l \in [\![0,d-1]\!]$, is an orthonormal basis of the lowest Landau level in Landau gauge:
    \begin{align*}
        \psi_{0l}(z) &\coloneq \dfrac{\pi^{-\frac{1}{4}}}{\sqrt{Ll_b}} e^{2i\pi l\frac{x}{L}}\sum_{k \in \mathbb{Z}}e^{2i\pi kd\frac{x}{L} -\frac{1}{2l_b^2}\left(y+kL+l\frac{L}{d}\right)^2}  \yesnumber{phi_l} \\
        &= \dfrac{\pi^{-\frac{1}{4}}}{\sqrt{Ll_b}} e^{-\frac{\pi l^2}{d} -\frac{y^2}{2l_b^2} + 2i\pi l\frac{z}{L}} \theta\left(d\dfrac{z}{L}+ il, id\right)  \yesnumber{psi0l}
    \end{align*}
    One can check that the above eigenfunctions satisfy the boundary conditions \cref{eq:tau Lan}. Using \cref{eq:phi_l} we observe the $L$-periodicity along the real axis. Along the imaginary axis we increment the index $k$ by $1$:
    \begin{align*}
        \psi_{0l}(z + iL) = \dfrac{\pi^{-\frac{1}{4}}}{\sqrt{Ll_b}} e^{2i\pi l\frac{x}{L}}\sum_{k \in \mathbb{Z}}e^{2i\pi kd\frac{x}{L} -\frac{1}{2l_b^2}\left(y + (k+1)L+l\frac{L}{d}\right)^2} = e^{-2i\pi d\frac{x}{L}} \psi_{0l}(z)
    \end{align*}
    and obtain the magnetic periodic boundary conditions in Landau gauge \cref{eq:Lan gauge functions}. The lowest Landau level is then generated by successive magnetic translations. Indeed if $l \in [\![0,d-1]\!]$,
    \begin{align*}
        \psi_{0l}= \prth{\tau_{-i\frac{L}{d}}}^l \psi_{00} = \tau_{-il\frac{L}{d}}\psi_{00}   \yesnumber{gen mang}
    \end{align*}
        
    In order to obtain a full basis of $L^2$, we only need to apply successively $a^\dagger$ to generate the Landau levels and $\tau_{-i\frac{L}{d}}$ to generate the eigenfunctions inside a Landau level. The successive applications of $a^\dagger$ bring out Hermite polynomials.

    \begin{notation}[title=Hermite polynomials]
        For $n \in \mathbb{N}$, we define the $n^{th}$ Hermite polynomial by
        \begin{equation}
            H_n \coloneq (-1)^n e^{x^2} \left(\dfrac{d}{dx}\right)^n e^{-x^2}
        \end{equation}
    \end{notation}
      
    Using this, a direct computation shows that the following family indexed by $(n, l) \in \mathbb{N}\times  [\![0,d-1]\!]$ is a Hilbert basis of eigenfunctions of $\mathscr{L}_{\hbar, b}$ in Landau gauge:
    \begin{align*} 
        \psi_{nl} &\coloneq \dfrac{{a^\dagger}^n}{\sqrt{n!}} \prth{\tau_{-i\frac{L}{d}}}^l \psi_{00}\\ 
        &= \dfrac{c_n}{\sqrt{L l_b}} e^{2i\pi l\frac{x}{L}} \sum_{k\in\mathbb{Z}} H_n \left(\dfrac{1}{l_b}\left[y+kL+l\dfrac{L}{d}\right]\right) e^{2i\pi k d \frac{x}{L} - \frac{1}{2l_b^2}\left(y + kL + l\frac{L}{d}\right)^2}  \yesnumber{psi_nl}
    \end{align*}
    with the normalization factor
    \begin{align*}
        c_n \coloneq \dfrac{1}{\pi^{\frac{1}{4}}\sqrt{n!}} \left( \dfrac{-i}{\sqrt{2}}\right)^n
    \end{align*}
    
    As expected with our boundary conditions the modulus of the eigenfunctions:
    \begin{align*}
        |\psi_{nl}| =  \abs{\dfrac{c_n}{\sqrt{L l_b}}\sum_{k\in\mathbb{Z}} H_n \prth{\dfrac{1}{l_b}\sbra{y+kL+l\dfrac{L}{d}}} e^{2i\pi k d \frac{x}{L} - \frac{1}{2l_b^2}\left(y + kL + l\frac{L}{d}\right)^2}}  \yesnumber{psi_nl mod}
        \end{align*}
    is periodic on the lattice $L\mathbb{Z^2}$, but the periodicity along the real axis is even shorter. Indeed we see in \cref{eq:psi_nl mod} that $|\psi_{nl}|$ is $L/d$-periodic in $x$.
    
    We can write another useful form of equation \cref{eq:psi_nl} using the Poisson summation formula. The advantage of the expression in \cref{prop:poisson wave func} is the fact that the index $l$ is decoupled from the polynomials and the Gaussian factors which is not the case in \cref{eq:psi_nl}. This will simplify the computation of the Landau level's projector when we will sum over $l$ in \cref{eq:l simplification}.
    
    \begin{notation}[title=Fourier transform]
        We use the convention
        \begin{align*}
            \mathcal{F}g(\nu) \coloneq \hat{g}(\nu) \coloneq \dfrac{1}{\sqrt{2\pi}}\intr_{\mathbb{R}} g(x)e^{-i\nu x}dx
        \end{align*}
        for which $\mathcal{F}$ is unitary on $L^2(\mathbb{R})$. And denote the Hermite function
        \begin{align*}
            h_n(x) \coloneq H_n(x)e^{-\frac{x^2}{2}} \yesnumber{hermite func}
        \end{align*}
    \end{notation}
    
    In this convention the Poisson summation formula is
    \begin{align*} 
        \sum_{k\in \mathbb{Z}}g(k) = \sqrt{2\pi}\sum_{k\in \mathbb{Z}}\hat{g}\left(2\pi k\right)  \yesnumber{poisson formula}
    \end{align*}
    $h_n$ are the eigenfunctions of the one dimensional harmonic oscillator and of the Fourier transform:
    \begin{align*} 
        \mathcal{F}h_n = (-i)^n h_n  \yesnumber{HO wave func transf} 
    \end{align*}
    with the following normalization
    \begin{align*}
        \norm{h_n}_{L^2}^2 = \sqrt{\pi}2^n n! \yesnumber{hn norm}
    \end{align*}
    With this we are ready for the next computation (see \cref{sec:appendix}):
    
    \begin{proposition}[title=Poisson summation of eigenfunctions, label=poisson wave func]
        In Landau gauge,
        \begin{align*}
            \psi_{nl}(z)=\widetilde{c}_n\dfrac{\sqrt{l_b}}{L^{\frac{3}{2}}}e^{-i\frac{xy}{l_b^2}}  \sum_{k\in \mathbb{Z}} H_n \left(\dfrac{1}{l_b}\left[x + k\dfrac{L}{d}\right]\right) e^{-2i\pi k \left(\frac{y}{L}+\frac{l}{d}\right) -\frac{1}{2l_b^2}\left(x + k\frac{L}{d}\right)^2}
        \end{align*}
            with the normalization factor
        \begin{align*}
            \widetilde{c}_n \coloneq \dfrac{\pi^{\frac{1}{4}} (-1)^n2^{\frac{1 - n}{2}}}{\sqrt{n!}}
        \end{align*}
    \end{proposition}

    \section{Projectors on Landau levels}\label{sec:Projector}
    
    From the construction of an $L^2(\Omega)$ basis adapted to Landau levels, we define the projectors on Landau levels in \cref{not:proj}. Since the phase space is $\mathbb{N}\times \Omega$ we also want to localise the projectors in space. Then we prove some properties of the projector that will be needed for the semi-classical analysis. In \cref{prop:proj conv} we give an equivalent for the diagonal of the projector's integral kernel, and in \cref{cor:tr pi} an equivalent for its trace.

\subsection{\text{nLL} projectors}
    
    \begin{notation}[title=Projectors, label=proj]
        The orthogonal projector on $\text{nLL}$ is
        \begin{align*}
            \Pi_n \coloneq& \sum_{l = 0}^{d-1} \ket{\psi_{nl}} \bra{\psi_{nl}} 
        \end{align*}
        Let $g \in C^\infty(\mathbb{R}^2, \mathbb{R}_+)$ radial with support included in the ball $B(0, L/2)$ such that $\norm{g}_{L^2} = 1$. Let $\lambda \ge 1$, define the localizer $g_\lambda \in C^\infty(\mathbb{T})$ defined by 
        \begin{align*}
            g_\lambda(x) \coloneq \syst{\lambda g(\lambda x) &\text{if } x \in B\prth{0, \frac{L}{2\lambda}} \\ 0 &\text{else}}
        \end{align*}
        Note that
         \begin{align*}
             \norm{g_\lambda}_{L^2} = 1
         \end{align*}
        $\forall X \coloneq (n, R) \in \mathbb{N}\times\Omega$, define the localised projector
        \begin{align*}
            \Pi_X \coloneq \Pi_{n,R} \coloneq& g_\lambda(\bullet - R)\Pi_n g_\lambda(\bullet - R) \yesnumber{pinR}
        \end{align*}
        We assume the following scaling for $\lambda \coloneq (\lambda_N)_N$:
        \begin{align*}
            1 \ll \lambda \ll \dfrac{N^{-\frac{1}{2}}}{\hbar ^2} \yesnumber{lambda scaling}
        \end{align*}
    \end{notation}
    
    This localised projector was introduced by Lieb, Solovej and Yngvason in \cite{LS91} and \cite{Yngvason} where it has been called coherent operator. We take the bounds \cref{eq:lambda scaling} in order to have $g_\lambda^2 \wslim \delta$ so the projector is well localised and
    \begin{align*}
        \dfrac{\hbar^2}{l_b} \lambda = \hbar b \lambda l_b \ll 1
    \end{align*}
    This is necessary because $\hbar b \lambda l_b$ is the order of some error terms coming from the kinetic energy (for example in \cref{prop:HF to SC}). $\forall X \coloneq (n, R) \in \mathbb{N}\times\Omega, \Pi_n$ and $\Pi_X$ are positive satisfy the following resolution of identity:
    \begin{align*}
        \sum_{n\in\mathbb{N}} \Pi_n = \text{Id},
        \qquad \intr_{\mathbb{N}\times\Omega} \Pi_{X}d\eta(X)  = \text{Id}  \yesnumber{nR res id}
    \end{align*}

\subsection{Integral kernels of the projectors}

    From \cref{prop:poisson wave func}, using
    \begin{align*}
        \sum_{l=0}^{d-1} e^{2i\pi l\frac{p - k}{d} } = d \mathbb{1}_{p=k \,(\text{mod}\,d)} \yesnumber{l simplification}
    \end{align*}
    and the notations \cref{eq:hermite func} and $x\coloneq x_1 + i x_2$, $y = y_1 + i y_2$, the expression of the \text{nLL}-projector's kernel is \newpage
    \begin{align*}
        \Pi_n(x, y) &= \dfrac{1}{\norm{h_n}_{L^2}^2 Ll_b} e^{i\frac{y_1y_2 - x_1x_2}{l_b^2}}\sum_{k, q \in \mathbb{Z}} H_n \prth{\dfrac{1}{l_b}\sbra{x_1+k\dfrac{L}{d}}} H_n \prth{\dfrac{1}{l_b}\sbra{y_1+ qL + k\dfrac{L}{d}}} \\
        &\cdot e^{2i\pi k\frac{y_2 - x_2}{L} + 2i\pi d q \frac{y_2}{L} -\frac{1}{2l_b^2}\prth{x_1+k\frac{L}{d}}^2 -\frac{1}{2l_b^2}\prth{y_1+ qL + k\frac{L}{d}}^2} \yesnumber{Pi_n}
    \end{align*}
    
    The simplification for the sum in $l$ is the reason why we used the Poisson formula on eigenfunctions. The argument does not work on the expression in \cref{eq:psi_nl} since the Gaussian terms depend on $l$. If we consider the same setup on the whole space $\mathbb{R}^2$ instead of $\Omega$, the expression of the projector in Landau gauge becomes (see \cite[Section 3.2]{Jain}):
    \begin{align*}
        \Pi^\infty_0(x, y) \coloneq \dfrac{1}{2\pi l_b^2}e^{-\frac{\abs{x - y}}{2l_b^2} + i\frac{\text{Im}\sbra{x\overline{y}}}{4l_b^2} + i \frac{y_1y_2 - x_2x_1}{2l_b^2}}
    \end{align*}
    The next proposition states that the diagonal of the projector's kernel on $\Omega$ converges to that of the projector on the whole space. This is expected since the limit is equivalent to a scaling where the size of the domain goes to infinity.
    
    \begin{proposition}[title=Convergence of the integral kernel, label=proj conv]
        The kernel \cref{eq:Pi_n} satisfies
        \begin{align*}
            \norm{\Pi_n(z, z) - \dfrac{1}{2\pi l_b^2}}_{L^\infty}\le \dfrac{C(n)}{L l_b}  \yesnumber{kernel approx 1}
        \end{align*}
        Moreover with notation \cref{eq:hermite func},
        \begin{align*}
            \norm{(\mathscr{P}_{\hbar, b}\Pi_n)(z, z) - \dfrac{b}{l_b}\cdot \dfrac{1}{2\pi\norm{h_n}_{L^2}^2}\intr_\mathbb{R}\prth{\matrx{i h_n'(u)\\  u h_n(u)}}  h_n(u)e^{-u^2}du}_{L^\infty} \le C(n) b \yesnumber{kernel approx 2}
        \end{align*}
    \end{proposition}

    The proof of this result is purely about the convergence of some Riemann sums, see \cref{sec:appendix}. Finally, we compare the trace of $\Pi_{n,R}$ to the trace of the projector on the whole space.
    
    \begin{corollary}[title=Approximation of the projector's trace, label=tr pi]
        \begin{align*}
            \abs{\Tr{\Pi_{n,R}} - \dfrac{1}{2\pi l_b^2}} \le \frac{C(n)}{l_b}
        \end{align*}
    \end{corollary}
    
    \begin{proof}
        This is a direct consequence of \cref{prop:proj conv} after integrating on $z \in \Omega$:
        \begin{align*}
            \Tr{\Pi_{n,R}} = \intr_\Omega \Pi_{n,R}(z, z)dz = \dfrac{1}{2\pi l_b^2}\intr_\Omega g_\lambda(z-R)^2dz + \mathcal{O}\prth{\dfrac{1}{l_b}} = \dfrac{1}{2\pi  l_b^2} + \mathcal{O}\prth{\dfrac{1}{l_b}} 
        \end{align*}
    \end{proof}

    \section{A Lieb-Thirring inequality}\label{sec:LT}

    In this section we prove a Lieb-Thirring inequality for the magnetic Laplacian with magnetic periodic boundary conditions: 
    
    \begin{theorem}[title= Kinetic energy inequality, label=kin energy ineq]
        Let $\gamma \in \mathcal{L}^1(L^2(\Omega))$ a positive operator, then for large enough $b$, 
        \begin{align*}
            \intr_\Omega \rho_\gamma^2 \le \dfrac{C\norm{\gamma}_{\mathcal{L}^\infty}}{\hbar^2} \Tr{\mathscr{L}_{\hbar, b} \gamma} \yesnumber{kin energy}
        \end{align*}
        Moreover if $\psi_N \in L_-^2(\Omega^N)$ with $\norm{\psi_N}_{L^2} = 1$,
        \begin{align*}
            \norm{\rho_{\psi_N}^{(1)}}_{L^2}^2 \le \dfrac{C}{\hbar b} \Tr{\mathscr{L}_{\hbar, b} \gamma_{\psi_N}^{(1)}} 
            \text{ and } 
            \abs{\intr_\Omega V \rho_{\psi_N}^{(1)}} \le \dfrac{C}{\hbar b } \Tr{\mathscr{L}_{\hbar, b} \gamma_{\psi_N}^{(1)}} \norm{V}_{L^2} \yesnumber{Kin E ineq}
        \end{align*}
        \begin{align*}
            \intr_{\Omega^2} w\rho_{\psi_N}^{(2)} \le \dfrac{C}{\hbar b } \Tr{\mathscr{L}_{\hbar, b} \gamma_{\psi_N}^{(1)}} \norm{w}_{L^2} \yesnumber{Kin E ineq w}
        \end{align*}
    \end{theorem}

    This result is well known in the absence of magnetic field, see \cite[Theorem 5.2]{FLLS13}. We adapt the proof of \cite[Chapter 4]{LS08} to magnetic periodic boundary conditions. To achieve this we prove the following sequence of inequalities: a Kato inequality (\cref{lem:Kato}), a diamagnetic inequality for Green functions (\cref{prop:diam ineq}), a Lieb-Thirring inequality (\cref{th:LT}) from which we deduce the inequality on the kinetic energy (\cref{th:kin energy ineq}). The reader already familiar with Lieb-Thirring inequalities might jump to \cref{sec:SC}.
    
\subsection{Reduced densities} \label{ssec:red densities}

    We give here some usual properties of the reduced density matrices, see \cref{not:reduced densities}. Let $\gamma_N$ be an $N$-body density matrix,
    \begin{align*}
        \dfrac{\Tr{\mathscr{H}_N\gamma_N}}{N} =  \Tr{\prth{\mathscr{L}_{\hbar, b} + V}\gamma_N^{(1)}} + \Tr{w \gamma_N^{(2)}} \yesnumber{E reduced d}
    \end{align*}
    moreover, reduced densities inherit trace and Pauli principle from $\gamma_N$:
    \begin{align*}
        \Tr{\gamma_N^{(k)}} = 1 ,\quad 0\le \gamma_N^{(k)} \le \frac{k!\prth{N-k}!}{N!} \yesnumber{red d}
    \end{align*}
    We can also express the reduced density matrices in term of integral kernels:
    \begin{align*}
        \gamma_N^{(k)}(x_{1:k}, y_{1:k})  \coloneq \intr_{\Omega^{N - k}} \gamma_N(x_{1:k}, x_{k+1:N}; y_{1:k}, x_{k+1:N})dx_{k+1:N} \yesnumber{kernel red d}
    \end{align*}
    The reduced density matrices are symmetric under permutation of coordinates:
     \begin{align*}
        \forall \sigma \in  S_k,\gamma_N^{(k)}\prth{x_{\sigma(1:k)}, y_{\sigma(1:k)}} = \gamma_N^{(k)}\prth{x_{1:k}, y_{1:k}}
    \end{align*}
    and consistent:
    \begin{align*}
        \forall q\in\intint{1:k},\gamma_N^{(q)}\prth{x_{1:q}, y_{1:q}} = \intr_{\Omega^{k-q}} \gamma_N(x_{1:q}, x_{q+1:k}; y_{1:q}, x_{q+1:k})dx_{q+1:k}
    \end{align*}
    Note that the densities $\rho_{\gamma_N}^{(k)}$ inherit the symmetry and the consistency from the density matrices.

\subsection{A Kato inequality with periodic boundary conditions}

    One can look at \cite[Theorem X.27]{RS} for a proof of the Kato inequality in the non magnetic case.

    \begin{lemma}[title=Periodic Kato inequality, label=Kato]
        Define the complex sign
        \begin{align*}
            s(u) \coloneq \syst{\matrx{\dfrac{\overline{u}}{\abs{u}} &\text{ if } &u \neq 0 \\ 0 &\text{ if } &u = 0}}
        \end{align*}
        Let $u \in C^\infty(\Omega)$ then $\abs{u}\in H^1(\Omega)$ and
        \begin{align*}
            \abs{\hbar \nabla \abs{u}} \le \abs{\mathscr{P}_{\hbar, b}u} \yesnumber{kato H1}
        \end{align*}
        Moreover if $\abs{u}$ is periodic, then
        \begin{align*}
            - \hbar^2\Delta \abs{u} \le \text{Re}\sbra{s(u)\mathscr{L}_{\hbar, b} u} \yesnumber{kato H2}
        \end{align*}
        in the weak sense on $C^\infty_{p}(\Omega)$, or equivalently, $\forall \varphi \in C_p^\infty\prth{\Omega, \mathbb{R}_+}$,
        \begin{align*}
            -\hbar^2 \intr_\Omega \abs{u} \Delta \varphi \le \intr_\Omega \text{Re}\sbra{s(u)\mathscr{L}_{\hbar, b} u} \varphi
        \end{align*}
    \end{lemma}
    
    \begin{proof}
        \begin{align*}
            \dfrac{1}{2}\hbar \nabla \abs{u}^2 = \dfrac{1}{2}\hbar \nabla (\overline{u}u) = \text{Re}\sbra{\overline{u}\hbar \nabla u} = \text{Re}\sbra{\overline{u}\prth{\hbar \nabla - ibA}u}
        \end{align*}
        so taking absolute values
        \begin{align*}
            \abs{\hbar \dfrac{\nabla\abs{u}^2}{2}} \le \abs{u}\abs{\mathscr{P}_{\hbar, b}u} \yesnumber{kato1}
        \end{align*}
        Define
        \begin{align*}
            u_\epsilon = \sqrt{\abs{u}^2 + \epsilon^2} \in C^\infty_p(\Omega,  \mathbb{R}^*_+)
        \end{align*}
        Using \cref{eq:kato1},
        \begin{align*}
            \abs{\hbar\nabla u_\epsilon} = \dfrac{\abs{\hbar \nabla \abs{u}^2}}{2 u\epsilon} \le \dfrac{\abs{u}}{u_\epsilon} \abs{\mathscr{P}_{\hbar, b} u} \le \abs{\mathscr{P}_{\hbar, b} u} \yesnumber{approx kato1}
        \end{align*}
        So $\nabla u_\epsilon$ is bounded in $L^2(\Omega, \mathbb{R}^2)$ and converges weakly to $v \in L^2(\Omega, \mathbb{R}^2)$ up to sequence of $\epsilon$. Let $\varphi \in C_c^\infty\prth{\text{int}(\Omega), \mathbb{R}^2}$, since $\varphi \in L^2\prth{\Omega, \mathbb{R}^2}$ and $0\le u_\epsilon - \abs{u} \le \epsilon$
        \begin{align*}
            \intr_\Omega v\cdot \varphi = \lim \intr_\Omega \nabla u_\epsilon \cdot \varphi = - \lim \intr_\Omega u_\epsilon \nabla\cdot\varphi = - \intr_\Omega \abs{u} \nabla \cdot\varphi
        \end{align*}
        so $v = \nabla \abs{u}$ and the bound \cref{eq:approx kato1} passes to the limit and we obtain \cref{eq:kato H1}.
        
        To prove \cref{eq:kato H2} we use polar coordinates
        \begin{align*}
            u \eqcolon \abs{u}e^{i\theta}
        \end{align*}
        Let $x \in \Omega$, if $\abs{u}(x) \neq 0$, $\abs{u}$ is smooth on a neighbourhood $V_x$ of $x$ where $\abs{u} > 0$ and thus
        \begin{align*}
            \nabla u_\epsilon = \dfrac{\abs{u}}{u_\epsilon} \nabla\abs{u} \limit \nabla \abs{u} \text{ pointwhise on } V_x
        \end{align*}
        $e^{i\theta} = u/\abs{u}$ is also smooth and up to another restriction of $V_x$ we can invert the complex exponential so $\theta$ is smooth. Using Cauchy-Schwarz inequality
        \begin{align*}
            \text{Re}\sbra{s(u) \mathscr{L}_{\hbar, b} u} 
            =& \text{Re}\sbra{e^{-i\theta}\prth{-\hbar^2\Delta + 2i\hbar b A\cdot \nabla + i\hbar b (\nabla.A) + \abs{bA}^2}\abs{u}e^{i\theta}} \\
            =& -\hbar^2\Delta \abs{u} + \text{Re}\sbra{ -\abs{u}e^{-i\theta}\hbar^2\Delta e^{i\theta} - 2i\hbar^2\nabla \abs{u} \cdot \nabla \theta + 2i\hbar b A\cdot \nabla \abs{u} } \\
            &- 2\hbar b\abs{u} A\cdot \nabla \theta + \abs{u}\abs{bA}^2 \\
            =& -\hbar^2\Delta \abs{u} + \text{Re}\sbra{ - i\hbar^2\abs{u} \Delta \theta + \abs{u}\hbar^2\abs{\nabla \theta}^2} - 2\hbar b\abs{u} A\cdot \nabla \theta + \abs{u}\abs{bA}^2 \\
            =& -\hbar^2\Delta \abs{u} + \hbar^2\abs{u}\abs{\nabla \theta}^2 - 2\hbar b \abs{u} A\cdot \nabla \theta + \abs{u}\abs{bA}^2
            \ge -\hbar^2\Delta \abs{u}
        \end{align*}
        Note that if $u(x)=0$ then $x$ is a local minimum of $u_\epsilon$ so
        \begin{align*}
            \Delta u_\epsilon (x) \ge 0
        \end{align*}
        Let $\varphi \in C^\infty_p \prth{\Omega, \mathbb{R}_+}$, since $u_\epsilon$ and $\varphi$ are periodic
        \begin{align*}
            \intr_\Omega u_\epsilon \Delta \varphi = \intr_\Omega \varphi \Delta u_\epsilon \ge \intr_{\Omega\backslash u^{-1}\prth{\sett{0}}} \varphi \Delta u_\epsilon \yesnumber{u_eps}
        \end{align*}
        Now we take $\epsilon \to 0$, $u_\epsilon$ converges uniformly to $\abs{u}$ so
        \begin{align*}
            \intr_\Omega u_\epsilon \Delta \varphi \limit\displaylimits_{\epsilon \to 0} \intr_\Omega \abs{u} \Delta \varphi \yesnumber{u_eps limit 1}
        \end{align*}
        Using $\abs{u} \le u_\epsilon$, when $u(x) \neq 0$,
        \begin{align*}
            \Delta u_\epsilon = \nabla \cdot \dfrac{\abs{u}\nabla\abs{u}}{u_\epsilon} = \dfrac{\abs{\nabla\abs{u}}^2 + \abs{u}\Delta \abs{u}}{u_\epsilon} - \dfrac{\abs{u}^2 \abs{\nabla\abs{u}}^2}{u_\epsilon^3} \ge \dfrac{\abs{u}}{u_\epsilon} \Delta \abs{u} \yesnumber{u_eps 3}
        \end{align*}
        so \cref{eq:u_eps} implies that
        \begin{align*}
            \intr_\Omega u_\epsilon \Delta \varphi \ge \intr_{\Omega\backslash u^{-1}\prth{\sett{0}}} \varphi  \dfrac{\abs{u}}{u_\epsilon} \Delta \abs{u} \yesnumber{u_eps 2}
        \end{align*}
        With dominated convergence using inequality \cref{eq:u_eps 3},
        \begin{align*}
            \intr_{\Omega\backslash u^{-1}\prth{\sett{0}}} \varphi  \dfrac{\abs{u}}{u_\epsilon} \Delta\abs{u} \limit\displaylimits_{\epsilon \to 0} \intr_{\Omega\backslash u^{-1}\prth{\sett{0}}} \varphi \Delta \abs{u} \yesnumber{u_eps limit 2}
        \end{align*}
        With \cref{eq:u_eps 2}, \cref{eq:u_eps limit 1} and \cref{eq:u_eps limit 2} we have
        \begin{align*}
            \intr_\Omega \abs{u} \Delta \varphi \ge \intr_{\Omega\backslash u^{-1}\prth{\sett{0}}} \varphi \Delta \abs{u}
        \end{align*}
        we can conclude that
        \begin{align*}
            -\hbar^2 \intr_\Omega \abs{u} \Delta \varphi 
            \le& -\hbar^2 \intr_{\Omega\backslash u^{-1}\prth{\sett{0}}} \varphi \Delta \abs{u} 
            \le  \intr_{\Omega\backslash u^{-1}\prth{\sett{0}}} \text{Re}\sbra{s(u)\mathscr{L}_{\hbar, b} u} \varphi
            = \intr_\Omega \text{Re}\sbra{s(u)\mathscr{L}_{\hbar, b} u} \varphi
        \end{align*}
    \end{proof}

\subsection{Diamagnetic inequality}

    The diamagnetic inequality in terms of Green functions allows us to restrict ourselves to the non magnetic case.

    \begin{notation}[title=Green functions]
        Resolvents of $-\hbar^2 \Delta$ with periodic boundary conditions and $\mathscr{L}_{\hbar, b}$ are well defined for $\lambda > 0$:
        \begin{align*}
            G_{bA, \lambda} \coloneq (\mathscr{L}_{\hbar, b} + \lambda)^{-1} \quad G_\lambda \coloneq (- \hbar^2 \Delta + \lambda)^{-1}
        \end{align*}
        Their integral kernels define the corresponding Green functions.
    \end{notation}
    
    They have the following properties:
    
    \begin{property}[label=green func]
        Let $x \in \Omega$, then $G_{bA, \lambda}(x, \bullet),  G_\lambda(x, \bullet)\in L^2(\Omega)$ and
        \begin{align*}
            G_\lambda(x, y) = G_\lambda(x - y) = G_\lambda(y - x) = \dfrac{1}{L^2} \sum_{k \in \frac{2\pi \hbar}{L}\mathbb{Z}^2} \dfrac{1}{k^2 + \lambda}e^{i k\cdot (x - y)} \ge 0
        \end{align*}
    \end{property}
    
    \begin{proof}
        The periodic Laplacian is diagonalizable in the plane wave basis indexed by $k \in \dfrac{2\pi}{L}\mathbb{Z^2}$:
        \begin{align*}
            e_k(x) \coloneq \dfrac{1}{L}e^{i k\cdot x}
        \end{align*}
        Indeed
        \begin{align*}
            -\hbar^2 \Delta + \lambda =  \sum_{k \in \frac{2\pi}{L}\mathbb{Z}^2} \prth{\hbar^2 k^2 + \lambda} \ket{e_k}\bra{e_k} 
        \end{align*}
        so
        \begin{align*}
            G_\lambda(x, y) = \dfrac{1}{L^2} \sum_{k \in \frac{2\pi}{L}\mathbb{Z}^2} \dfrac{1}{\hbar^2k^2 + \lambda}e^{i k\cdot \prth{x - y}}
        \end{align*}
        A change of index $k \coloneq -k$ gives $G_\lambda(x, y) \in \mathbb{R}$. Let $f \in L^2(\Omega)$, since $G_\lambda f$ solves
        \begin{equation*}
            (-\hbar^2\Delta + \lambda) u = f, u \in H_p^2(\Omega)
        \end{equation*}
        by the Lax-Milgram theorem, $G_\lambda f$ is the unique minimizer of the following functional
        \begin{align*}
            \mathcal{J}(u) \coloneq \intr_\Omega \prth{\hbar^2\abs{\nabla u}^2 + \lambda \abs{u}^2 - fu}dx
        \end{align*}
        Assuming $f \ge 0$, we see that $\mathcal{J}(u) \ge \mathcal{J}(\abs{u})$ and conclude that $G_\lambda f \ge 0$. This implies that $G_\lambda(x, y) \ge 0$. Finally,
        \begin{align*}
            -\hbar^2 \Delta + \lambda \ge \lambda \text{ and } \mathscr{L}_{\hbar, b} + \lambda \ge \lambda
        \end{align*}
        so
        \begin{align*}
            \norm{G_\lambda}_{\mathcal{L}^\infty} \le \dfrac{1}{\lambda} \text{ and } \norm{G_{bA, \lambda}}_{\mathcal{L}^\infty} \le \dfrac{1}{\lambda}
        \end{align*}
        and $G_{bA, \lambda}(x, \bullet), G_\lambda(\bullet, y)\in L^2(\Omega)$.
    \end{proof}

    Now we prove a diamagnetic inequality:

    \begin{proposition}[title=Diamagnetic inequality for Green functions, label=diam ineq]
        For all $x\in \Omega$ and for almost every $y \in \Omega$,
        \begin{align*}
            \abs{G_{bA, \lambda}(x, y)} \le G_\lambda(x, y)
        \end{align*}
    \end{proposition}
    
    \begin{proof}
        Let $\varphi \in C^\infty(\Omega)$, by definition
        \begin{align*}
            \intr_\Omega G_{bA, \lambda}\prth{x, \bullet} \prth{\mathscr{L}_{\hbar, b} + \lambda} \varphi = G_{bA, \lambda}\prth{\mathscr{L}_{\hbar, b} + \lambda}\varphi = \varphi
        \end{align*}
        so, in the distributional sense
        \begin{align*}
            \prth{\mathscr{L}_{\hbar, b} + \lambda} G_{bA, \lambda}(x, \bullet) = \delta_x \yesnumber{magnetic green EDP}
        \end{align*}
        Let $\rho \in C^\infty(\mathbb{R}^2, \mathbb{R}_+)$ radial with support included in the ball $B(0, L/2)$ such that $\norm{\rho}_{L^1} = 1$. Let $n\in \mathbb{N^*}$, define the localizer $\rho_n \in C^\infty(\mathbb{T})$ defined by 
        \begin{align*}
            \rho_n(x) \coloneq \syst{n^2 \rho(n x) &\text{if } x \in B\prth{0, \frac{L}{2n}} \\ 0 &\text{else}}
        \end{align*}
        Since $\rho_n$ is periodic, the regularisation
        \begin{align*}
            u_x \coloneq G_{bA, \lambda}(x, \bullet) * \rho_n \in C_{p}^\infty(\Omega)
        \end{align*}
        Thus, equation \cref{eq:magnetic green EDP} becomes
        \begin{align*}
            \prth{\mathscr{L}_{\hbar, b} + \lambda}  u_x = \delta_x * \rho_n = \rho_n(x - \bullet) \yesnumber{kato rho n}
        \end{align*}
        We estimate
        \begin{align*}
            \text{Re}\sbra{s(u_x) \prth{\mathscr{L}_{\hbar, b} + \lambda} u_x} =  \text{Re}\sbra{s(u_x)\rho_n(x - \bullet) } \le \rho_n(x -\bullet)
        \end{align*}
        Then we apply Kato's inequality \cref{eq:kato H1} to $u_x$, use $s(u_x) u_x = \abs{u_x}$ and obtain
        \begin{align*}
            \prth{-\hbar^2\Delta + \lambda} \abs{u_x} \le \text{Re}\sbra{s(u_x) \mathscr{L}_{\hbar, b} u_x} + \lambda \abs{u_x} \le \rho_n(x-\bullet) \yesnumber{kato3}
        \end{align*}
        in a weak sense on $C_{p}^\infty(\Omega)^*$.
        
        Similarly as \cref{eq:kato rho n},
        \begin{align*}
            (-\hbar^2\Delta +\lambda) \rho_n * G_\lambda(\bullet, y) = \rho_n(y - \bullet)
        \end{align*}
        Thus testing inequality \cref{eq:kato3} on $\rho_n * G_\lambda(\bullet, y) \in C_{p}^\infty(\Omega, \mathbb{R}_+)$ we get
        \begin{align*}
            \intr_\Omega \abs{u_x(z)} \rho_n(y - z) dz &\le \intr_\Omega \rho_n(x - z)  \rho_n * G_\lambda(\bullet, y) (z) dz
        \end{align*}
        With the changes of variables $t \coloneq t + x - y$, $z \coloneq z + x - y$ and \cref{prty:green func},
        \begin{align*}
            \abs{G_{bA, \lambda}(x, \bullet) * \rho_n}*\rho_n (y)
            \le& \iint\displaylimits_{\Omega^2} \rho_n(x - z)\rho_n(z - t) G_\lambda(t - y) dz dt \\
            =& \iint\displaylimits_{\Omega^2} \rho_n(2x - y - z)\rho_n(z - t) G_\lambda(t - x) dz dt \\
            =& \rho_n * \rho_n * G_\lambda (x, \bullet) (2x - y) \yesnumber{approx diamagn}
        \end{align*}
        If $\varphi_n  \to \varphi$ in $L^2(\Omega)$, by Young's inequality 
        \begin{align*}
            \norm{\rho_n * \varphi_n - \varphi}_{L^2} 
            &\le \norm{\rho_n * (\varphi_n - \varphi)}_{L^2} + \norm{\rho_n * \varphi - \varphi}_{L^2} \\
            &\le \norm{\rho_n}_{L^1} \norm{\varphi_n - \varphi}_{L^2} + \norm{\rho_n * \varphi - \varphi}_{L^2} 
            \limit\displaylimits_{n\to\infty} 0
        \end{align*}
        Let $x \in \Omega$, with \cref{prty:green func}, after extraction of a subsequence, $\abs{\rho_n * G_{bA, \lambda}(x, \bullet)} \to \abs{G_{bA, \lambda}(x, \bullet)}$ in $L^2(\Omega)$ so
        \begin{align*}
            &\abs{G_{bA, \lambda}(x, \bullet) * \rho_n} * \rho_n \limit\displaylimits_{n\to\infty} \abs{G_{bA, \lambda}(x, \bullet)}
        \end{align*}
        in $L^2(\Omega)$ and up to another subsequence almost everywhere. Similarly, almost everywhere
        \begin{align*}
            &\rho_n * \rho_n * G_\lambda(x, \bullet) \limit\displaylimits_{n\to\infty} G_\lambda(x, \bullet)
        \end{align*}
        So passing to the limit in \cref{eq:approx diamagn}, for almost every $y\in\Omega$,
        \begin{align*}
            \abs{G_{bA, \lambda}(x, y)} \le G_\lambda(x, 2x - y) = G_\lambda(y - x) = G_\lambda(x, y)
        \end{align*}
    \end{proof}

\subsection{Lieb-Thirring inequality}

    We add a constant to the Laplacian to ensure that the constant function has a non-zero energy.

    \begin{theorem}[title=Lieb-Thirring Inequality, label=LT]
        Let $\mathcal{V} \in L^2(\Omega, \mathbb{R}_+)$,
        \begin{align*}
            -\Tr{\mathbb{1}_{\prth{\mathscr{L}_{\hbar, b} + 1 - \mathcal{V}} \le 0} \prth{\mathscr{L}_{\hbar, b} + 1 - \mathcal{V}}} \le \dfrac{C_{LT}}{\hbar^2} \intr_\Omega \mathcal{V}(x)^2dx \yesnumber{LT}
        \end{align*}
    \end{theorem}
    
    \begin{proof}
        We denote $N_\lambda$ the number of eigenvalues of $\mathscr{L}_{\hbar, b} + 1$ less than or equal to $\lambda$. From \cite[Section 4.3]{LS08},
        \begin{align*}
            -\Tr{\mathbb{1}_{\prth{\mathscr{L}_{\hbar, b} + 1 - \mathcal{V}} \le 0} \prth{\mathscr{L}_{\hbar, b} + 1 - \mathcal{V}}} &= \intr_{\mathbb{R}_+}N_\lambda d\lambda
        \end{align*}
        Define the Birman-Schwinger operator 
        \begin{align*}
            K_\lambda \coloneq \sqrt{\mathcal{V}} G_{bA, \lambda} \sqrt{\mathcal{V}}
        \end{align*}
        and let $B_\lambda$ be the number of eigenvalues of $K_\lambda$ greater or equal to $1$. We use the diamagnetic inequality to restrict to the non magnetic case. Since $G_{bA, \lambda}$ is positive, we can define its square root. Using the arguments of \cite[Theorem 4.4]{LS08} we can deduce from \cref{prop:diam ineq} the diamagnetic inequality for $\sqrt{G_{bA, \lambda}}$:
        \begin{align*}
            \abs{\sqrt{G_{bA, \lambda}}(x, y)} \le \sqrt{G_\lambda}(x, y)
        \end{align*}
        Hence with \cref{prop:diam ineq},
        \begin{align*}
            \abs{G^\frac{3}{2}_{bA, \lambda}(x, y)} = \abs{\intr_\Omega G_{bA, \lambda}(x, z) \sqrt{G_{bA, \lambda}}(z, y) dz} \le \intr_\Omega G_\lambda(x, z) \sqrt{G_\lambda}(z, y) dz = G^\frac{3}{2}_\lambda(x, y)
        \end{align*}
        So taking $m \coloneq 3/2$ and using an inequality on the traces of powers (see \cite[Theorem 4.5]{LS08}),
        \begin{align*} 
            N_\lambda  &= B_\lambda \le \Tr{K_\lambda^m}\le \Tr{\mathcal{V}^{\frac{m}{2}}K_\lambda^m \mathcal{V}^{\frac{m}{2}}} \le \intr_\Omega \mathcal{V}(x)^m \abs{G_{A, \lambda+1}^m(x, x)}dx \\
            &\le \intr_\Omega {\mathcal{V}(x)}^m G^m_{\lambda+1}(x, x)dx
        \end{align*}
        So we obtain
        \begin{align*}
            -\Tr{\mathbb{1}_{\prth{\mathscr{L}_{\hbar, b} + 1 - \mathcal{V}} \le 0} \prth{\mathscr{L}_{\hbar, b} + 1 - \mathcal{V}}} \le \intr_\Omega\intr_1^\infty {\mathcal{V}(x)}^m G^m_\lambda(x, x)dx d\lambda \yesnumber{BS prin}
        \end{align*}
        The kernel of $G_\lambda^m$ is
        \begin{align*}
            G^m_\lambda(x) &= \dfrac{1}{L^2}\sum_{k \in \frac{2\pi \hbar}{L}\mathbb{Z}^2} \dfrac{1}{\prth{k^2 + \lambda}^m} e^{i \frac{k\cdot x}{\hbar}}
        \end{align*}
        We use the integral bound for the sum
        \begin{align*}
            \sum_{k\in \mathbb{Z}} \dfrac{1}{\prth{k^2 + \lambda}^m} \le \lambda^{-m} + \intr_\mathbb{R} \dfrac{1}{\prth{u^2 + \lambda}^m}du
        \end{align*}
        so
        \begin{align*}
            G^m_\lambda(0) 
            &= \dfrac{1}{L^2} \sum_{k, q\in \mathbb{Z}} \dfrac{1}{\prth{\prth{\frac{2\pi\hbar}{L}}^2(k^2 + q^2) + \lambda}^m} 
            = \dfrac{1}{L^2} \prth{\frac{L}{2\pi\hbar}}^{2m} \sum_{k, q\in \mathbb{Z}} \dfrac{1}{\prth{k^2 + q^2 + \prth{\frac{L}{2\pi\hbar}}^2\lambda}^m} \\
            &\le \dfrac{1}{L^2} \prth{\frac{L}{2\pi\hbar}}^{2m} \sum_{k\in\mathbb{Z}} \prth{\dfrac{1}{\prth{k^2 + \prth{\frac{L}{2\pi\hbar}}^2\lambda}^m} + \intr_\mathbb{R}\dfrac{1}{\prth{k^2 + u^2 + \prth{\frac{L}{2\pi\hbar}}^2\lambda}^m}du} \\
        \end{align*}
        We estimate the integral
        \begin{align*}
            \intr_\mathbb{R}\dfrac{1}{\prth{k^2 + u^2 + \prth{\frac{L}{2\pi\hbar}}^2\lambda}^m}du 
            &= \prth{k^2 + \prth{\frac{L}{2\pi\hbar}}^2\lambda}^{-m} \intr_\mathbb{R}\dfrac{1}{\prth{\frac{u^2}{k^2 + \prth{\frac{L}{2\pi\hbar}}^2\lambda} + 1 }^m}du \\
            &= \dfrac{I(m)}{\prth{k^2 + \prth{\frac{L}{2\pi\hbar}}^2\lambda}^{m-\frac{1}{2}}} 
        \end{align*}
        with
        \begin{align*}
            m > \frac{1}{2} \implies I(m) \coloneq \intr_\mathbb{R} \frac{1}{\prth{1 + u^2}^m}du < \infty
        \end{align*}
        Similarly
        \begin{align*}
            G^m_\lambda(0) 
            \le& \dfrac{1}{L^2} \prth{\frac{L}{2\pi\hbar}}^{2m} \sum_{k\in\mathbb{Z}} \prth{\dfrac{1}{\prth{k^2 + \prth{\frac{L}{2\pi\hbar}}^2\lambda}^m} +  \dfrac{I(m)}{\prth{k^2 + \prth{\frac{L}{2\pi\hbar}}^2\lambda}^{m-\frac{1}{2}}}} \\
            \le& \dfrac{\lambda^{-m}}{L^2}  + \dfrac{I(m)}{2\pi\hbar L}\lambda^{-m + \frac{1}{2}} \\
            &+ \dfrac{1}{L^2} \prth{\frac{L}{2\pi\hbar}}^{2m} \prth{ \intr_\mathbb{R}\dfrac{1}{\prth{u^2 + \prth{\frac{L}{2\pi\hbar}}^2\lambda}^{m}}du + \intr_\mathbb{R}\dfrac{ I(m)}{\prth{u^2 + \prth{\frac{L}{2\pi\hbar}}^2\lambda}^{m - \frac{1}{2}}}du} \\
            \le& \dfrac{\lambda^{-m}}{L^2}  + \dfrac{I(m)}{\pi\hbar L}\lambda^{-m + \frac{1}{2}}  + \dfrac{1}{\prth{2\pi\hbar}^2} I(m)I\prth{m-\frac{1}{2}} \lambda^{-m + 1} \\
            \le& \dfrac{C(m)}{\hbar^2} \lambda^{-m +1}
        \end{align*}
        since $\lambda \ge 1$. We need $m > 1$ for the integrals to converge. We use the same trick as \cite{LS08} changing the potential to
        \begin{align*}
            \mathcal{W}_\lambda(x) \coloneq \max \prth{\mathcal{V} - \dfrac{\lambda}{2}, 0}
        \end{align*}
        Combining this with \cref{eq:BS prin} and the change of variable
        \begin{align*}
            \mu \coloneq \frac{2V(x)}{\lambda}, d\lambda = -\dfrac{2V(x)}{\mu^2} d\mu
        \end{align*}
        we obtain
        \begin{align*}
            -\Tr{\mathbb{1}_{\prth{\mathscr{L}_{\hbar, b} + 1 - \mathcal{V}} \le 0} \prth{\mathscr{L}_{\hbar, b} + 1 - \mathcal{V}}}
            &\le \dfrac{C(m)}{\hbar^2}  \intr_1^\infty\intr_\Omega \lambda^{-m + 1}\max\prth{V(x) - \dfrac{\lambda}{2}, 0}^m d\lambda dx\\
            &\le \dfrac{C(m)}{\hbar^2}\intr_\Omega\prth{\intr_0^{2V(x)}  \lambda \prth{\dfrac{2V(x)}{\lambda} - 1}^m d\lambda }dx \\
            &= \dfrac{C(m)}{\hbar^2} \intr_\Omega\prth{\intr_1^\infty \dfrac{2V(x)}{\mu} \prth{\mu - 1}^m \cdot \dfrac{2V(x)}{\mu^2} d\mu }dx \\
            &= \dfrac{C(m)}{\hbar^2} \intr_\Omega V(x)^2 \prth{\intr_1^\infty  \dfrac{\prth{\mu - 1}^m}{\mu^3} d\mu} dx
        \end{align*}
        The integral in $\mu$ converges if $3 - m > 1$. To conclude we need $1 < m < 2$ hence the choice $m=3/2$ is convenient.
    \end{proof}
    
    This leads to proof of the Fundamental inequality of kinetic energy:
    
    \begin{proof}[th:kin energy ineq]
        With the variational principle and the Lieb-Thirring inequality \cref{eq:LT},
        \begin{align*}
            \Tr{\prth{\mathscr{L}_{\hbar, b} + 1}\gamma} 
            =& \Tr{\prth{\mathscr{L}_{\hbar, b} + 1 - \mathcal{V}}\gamma} + \Tr{\mathcal{V} \gamma}  \\
            \ge& \norm{\gamma}_{\mathcal{L}^\infty} \Tr{\mathbb{1}_{\prth{\mathscr{L}_{\hbar, b} + 1 - \mathcal{V}} \le 0} \prth{\mathscr{L}_{\hbar, b} + 1 - \mathcal{V}}} + \Tr{\mathcal{V} \gamma} \\
            \ge& - \dfrac{C_{LT}\norm{\gamma}_{\mathcal{L}^\infty}}{\hbar^2} \intr_\Omega \mathcal{V}^2 + \intr_\Omega \mathcal{V} \rho_\gamma
        \end{align*}
        Then choose $\mathcal{V} \coloneq C_N \mathbb{1}_{\rho_\gamma \le c}\rho_\gamma$:
        \begin{align*}
            \Tr{\prth{\mathscr{L}_{\hbar, b} + 1}\gamma} \ge C_N\prth{1 - C_N \dfrac{C_{LT}\norm{\gamma}_{\mathcal{L}^\infty}}{\hbar^2}} \intr_{\rho_\gamma \le c} {\rho_\gamma}^2
        \end{align*}
        The constant preceding the integral is maximal when
        \begin{align*}
            C_N \coloneq \dfrac{\hbar^2}{2C_{LT}\norm{\gamma}_{\mathcal{L}^\infty}}
        \end{align*}
        and we get
        \begin{align*}
            \Tr{\prth{\mathscr{L}_{\hbar, b} + 1} \gamma} \ge \dfrac{\hbar^2}{4C_{LT}\norm{\gamma}_{\mathcal{L}^\infty}} \intr_{\rho_\gamma \le c} {\rho_\gamma}^2 \yesnumber{fondkineq1}
        \end{align*}
        Since $\mathscr{L}_{\hbar, b} \ge \hbar b$,
        \begin{align*}
            \Tr{\mathscr{L}_{\hbar, b} \gamma} \ge \hbar b \Tr{\gamma}
        \end{align*}
        so because $\hbar b \to \infty$,
        \begin{align*}
             \Tr{\prth{\mathscr{L}_{\hbar, b} + 1} \gamma} \le \prth{1+\dfrac{1}{\hbar b}}\Tr{\mathscr{L}_{\hbar, b} \gamma} \le C \Tr{\mathscr{L}_{\hbar, b} \gamma}
        \end{align*}
        With this and monotone convergence we take the limit $c\to\infty$ in inequality \cref{eq:fondkineq1} and obtain \cref{eq:kin energy}. Applying this to \cref{eq:red d}, we have
        \begin{align*}
            \dfrac{1}{\hbar b} \Tr{\mathscr{L}_{\hbar, b}\gamma_{\psi_N}^{(1)}} 
            \ge C\dfrac{l_b^2}{\norm{\gamma_{\psi_N}^{(1)}}_{\mathcal{L}^\infty}} \norm{\rho_{\psi_N}^{(1)}}_{L^2}^2 
            \ge C Nl_b^2 \norm{\rho_{\psi_N}^{(1)}}_{L^2}^2
            \ge C \norm{\rho_{\psi_N}^{(1)}}_{L^2}^2
        \end{align*}
        For the second reduced density, by symmetry
        \begin{align*}
            &N\prth{\Tr{\mathscr{L}_{\hbar, b}\gamma_{\psi_N}^{(1)}} - \Tr{w \gamma_{\psi_N}^{(2)}}} 
            = \bk{\psi_N}{\prth{\sum_{i=1}^N \mathscr{L}_{\hbar, b}(x_i) - \dfrac{2}{N-1}\sum_{i<j}w(x_i - w_j)}\psi_N} \\
            =& \bk{\psi_N}{\prth{\sum_{i=1}^N \mathscr{L}_{\hbar, b}(x_i) - \dfrac{N}{N-1}\sum_{j=2}^N w(x_1 - w_j)}\psi_N} \\
            \ge& \bk{\psi_N}{\sum_{j=2}^N \prth{\mathscr{L}_{\hbar, b}(x_i) - \dfrac{N}{N-1}w(x_1 - x_j)}\psi_N} \\
            =& \intr_\Omega \bk{\psi_N(x, \bullet)}{\sum_{j=2}^N \prth{\mathscr{L}_{\hbar, b}(x_i) - \dfrac{N}{N-1}w(x - x_j)} \psi_N(x, \bullet)}dx \\
            =& \intr_\Omega\prth{\bk{\psi_N(x, \bullet)}{\sum_{j=2}^N \mathscr{L}_{\hbar, b}(x_i) \psi_N(x, \bullet)} - N \intr_\Omega w(x - y)\rho_{\psi_N(x, \bullet)}^{(1)}dy}dx
        \end{align*}
        Then using \cref{eq:Kin E ineq} and then Young's inequality,
        \begin{align*}
            &\Tr{\mathscr{L}_{\hbar, b}\gamma_{\psi_N}^{(1)}} - \Tr{w \gamma_{\psi_N}^{(2)}} \\
            \ge& \dfrac{1}{N}\intr_\Omega\prth{C \hbar b (N-1) \norm{\rho_{\psi_N(x, \bullet)}^{(1)}}_{L^2}^2 - N \intr_\Omega w(x - y)\rho_{\psi_N(x, \bullet)}^{(1)}dy}dx \\
            \ge& \intr_\Omega\prth{C \hbar b \norm{\rho_{\psi_N(x, \bullet)}^{(1)}}_{L^2}^2 - \intr_\Omega w(x - y)\rho_{\psi_N(x, \bullet)}^{(1)}dy}dx \\
            \ge& \intr_\Omega\prth{C \hbar b \norm{\rho_{\psi_N(x, \bullet)}^{(1)}}_{L^2}^2 - \dfrac{1}{2}\prth{\dfrac{1}{2C\hbar b}\norm{w}_{L^2}^2 + 2C\hbar b \norm{\rho_{\psi_N(x, \bullet)}^{(1)}}_{L^2}^2}}dx
            \ge - \dfrac{C}{\hbar b} \norm{w}_{L^2}^2
        \end{align*}
        Changing $w$ to $\epsilon w$, dividing by $\epsilon$ and using \cref{eq:E reduced d} gives
        \begin{align*}
            \intr_{\Omega^2} w\rho_{\psi_N}^{(2)} \le \dfrac{1}{\epsilon}\Tr{\mathscr{L}_{\hbar, b}\gamma_{\psi_N}^{(1)}} + \dfrac{C}{\hbar b} \epsilon\norm{w}_{L^2}^2
        \end{align*}
        To optimise in $\epsilon$, we choose $\epsilon \coloneq \dfrac{\hbar b}{\norm{w}_{L^2}}$, we get
        \begin{align*}
            \intr_{\Omega^2} w\rho_{\psi_N}^{(2)} 
            \le \prth{\dfrac{1}{\hbar b}\Tr{\mathscr{L}_{\hbar, b}\gamma_{\psi_N}^{(1)}} + C} \norm{w}_{L^2}
            \le  \dfrac{C}{\hbar b }\Tr{\mathscr{L}_{\hbar, b}\gamma_{\psi_N}^{(1)}} \norm{w}_{L^2}
        \end{align*}
        because $\mathscr{L}_{\hbar, b} \ge \hbar b$. Similarly with Young's inequality and \cref{eq:Kin E ineq},
        \begin{align*}
            \abs{\intr_\Omega V \rho_{\psi_N}^{(1)}} \le \dfrac{C}{\hbar b} \Tr{\mathscr{L}_{\hbar, b}\gamma_{\psi_N}^{(1)}} \norm{V}_{L^2}
        \end{align*}
    \end{proof}
    \section{Semi-classical limit}\label{sec:SC}

In this section we introduce the Husimi functions representing densities in the phase space. The fundamentals properties of these functions can be found in \cref{prty:sym measure}. Then we prove that the $N$-body quantum energy can be approximated by a semi-classical functional depending only on Husimi functions in \cref{prop:SC energy}.

\subsection{Husimi functions}

    \begin{notation}[label=Husimi]
        Let $k\in\mathbb{N}^*, \gamma_k \in \mathcal{L}^1(L^2(\Omega^k))$, recalling \cref{not:proj} and \cref{eq:lists} we define the associated Husimi functions or lower symbol as
        \begin{align*}
            m_{\gamma_k}(X_{1:k}) \coloneq \text{Tr} \sbra{\gamma_k \bigotimes_{i=1}^k\Pi_{X_i}} \text{ with } X_{1:k} \in (\mathbb{N}\times\Omega)^k 
        \end{align*}
        Conversely, if $m_k \in L^1\prth{(\mathbb{N}\times\Omega)^k }$, define the associated density matrix
        \begin{align*}
            \gamma_{m_k} \coloneq (2\pi l_b^2)^k \intr_{(\mathbb{N}\times\Omega)^k} m_k(X_{1:k})\bigotimes_{i=1}^k\Pi_{X_i} d\eta^{\otimes k}(X_{1:k})
        \end{align*}
        We call $m_k$ the upper symbol of $\gamma_{m_k}$. We also associate a density to $m_k$:
        \begin{align*}
            \rho_{m_k} \coloneq \sum_{n_{1:k}\in\mathbb{N}^k} m_k(n_{1:k};\bullet)
        \end{align*}
        we extend the definition \cref{eq:SC energy} to Husimi functions, if $k\ge 2$:
        \begin{align*}
            \mathcal{E}_{sc, \hbar b}\sbra{m_k} \coloneq \intr_{\mathbb{N}\times\Omega} E_n m_k^{(1)}(n, R)d\eta(n, R) + \intr_{\mathbb{N}\times\Omega}  V m_k^{(1)} d\eta + \intr_{(\mathbb{N}\times\Omega)^2}  w m_k^{(2)} d\eta^{\otimes 2} \yesnumber{SC energy 2}
        \end{align*}
        and we also extend \cref{eq:qLL} to $\rho_k \in L^1\prth{\Omega^k}$:
        \begin{align*}
            \mathcal{E}_{qLL}\sbra{\rho_k} = \intr_\Omega V \rho_k^{(1)} + \intr_{\Omega^2} w \rho_k^{(2)} \yesnumber{qLL energy 2}
        \end{align*}
        If one starts from $\psi_N \in L^2_-\prth{\Omega^N}$ we use the notation
        \begin{align*}
            m_{\psi_N} \coloneq m_{\gamma_{\psi_N}}
        \end{align*}
    \end{notation}
    
    For another discussion and further references about lower and upper symbols one can look at \cite[Definition 3.13]{Rougerie20}. The $k$-body Husimi function is the joint probability distribution for $k$ particles in phase space. Similarly as for \cref{eq:consistent not}, we have
    \begin{align*}
        m_{\gamma_N}^{(k)} = m_{\gamma_N^{(k)}} \text{ and } \rho_{m_{\gamma_N}}^{(k)} = \rho_{m_{\gamma_N}^{(k)}}
    \end{align*}

    We have the following properties for the Husimi functions, inherited from density matrices.
    
    \begin{property}[title=Husimi functions, label=sym measure]
        Let $\gamma_N$ be an $N$-body density matrix, then $m_{\gamma_N}^{(k)}$ are symmetric, consistent and satisfy
        \begin{align*}
            &0 \le m_{\gamma_N}^{(k)} \le \dfrac{(N-k)!}{(2\pi l_b^2)^k N!} + \mathcal{O}(l_b) \yesnumber{Husimi k bound} \\
            & \intr_{(\mathbb{N}\times\Omega)^k} m_{\gamma_N}^{(k)} d\eta^{\otimes k} = \norm{m_{\gamma_N}^{(k)}}_{L^1} = 1 \yesnumber{Husimi norm} \\
            &\rho_{m_{\gamma_N}}^{(k)} = (g^2_\lambda)^{\otimes k} * \rho_{\gamma_N}^{(k)} \yesnumber{husimi to density}
        \end{align*}
    \end{property}

    These are the usual properties for Husimi functions slightly modified here due to the approximation in \cref{cor:tr pi}. The proof (see \cref{sec:appendix}) uses the following translation between reduced density matrices and Husimi functions:

    \begin{lemma}[label=equiv Husimi admissible, title=Relations between Husimi functions and reduced densities]
        Let $\gamma_k \in \mathcal{L}^1\prth{L^2\prth{\Omega^k}}$ be a positive operator, then $m_{\gamma_k} \in L^1\prth{(\mathbb{N}\times\Omega)^k }$ and
        \begin{align*}
            &0 \le m_{\gamma_k} \le \dfrac{\norm{\gamma_k}_{\mathcal{L}^\infty}}{(2\pi l_b^2)^k}(1 + \mathcal{O}(l_b))
            & \intr_{(\mathbb{N}\times\Omega)^k} m_{\gamma_k} d\eta^{\otimes k} = \Tr{\gamma_k}
        \end{align*}
        Conversely if $m_k \in L^1\prth{(\mathbb{N}\times\Omega)^k}$ is positive, then $\gamma_{m_k} \in \mathcal{L}^1(L^2(\Omega^k))$ and
        \begin{align*}
            &0\le \gamma_{m_k} \le \prth{2\pi l_b^2}^k\norm{m_k}_{L^\infty}
            & \Tr{\gamma_{m_k}} = \norm{m_k}_{L^1} + \mathcal{O}(l_b)
        \end{align*}
        Moreover if $\gamma_N \in \mathcal{L}^1\prth{L_-^2\prth{\Omega^k}}$ and $1 \le k \le N$, then
        \begin{align*}
            m_{\gamma_N}^{(k)} \le \dfrac{(N-k)!}{\prth{2\pi l_b^2}^k N!}\Tr{\gamma_N}\prth{1 + \mathcal{O}\prth{l_b}}
        \end{align*}
    \end{lemma}

\subsection{Semi-classical energy} \label{subsec:SC energy}

    \begin{proposition}[title=Semi-classical approximation, label=SC energy]
        Let $\psi_N \in L_-^2(\Omega^N), \norm{\psi_N}_{L^2} = 1$, the quantum energy can be approximated with the semi-classical energy \cref{eq:SC energy 2}
        \begin{align*}
            \dfrac{\bk{\psi_N}{\mathcal{H}_N\psi_N}}{N} \eq_{N\to\infty}  \mathcal{E}_{sc, \hbar b}\sbra{m_{\psi_N}} + \mathcal{O}\prth{\dfrac{f(\lambda)}{\hbar b} \Tr{\mathscr{L}_{\hbar, b} \gamma_N^{(1)}}} + \mathcal{O}\prth{(\hbar \lambda)^2} \yesnumber{SC approx}
        \end{align*}
        where        
        \begin{align*}
            f(\lambda) \coloneq \text{max}\prth{\norm{g_\lambda^2 * V - V}_{L^2}, \norm{(g_\lambda^2)^{\otimes 2} * w - w}_{L^2}} \limit\displaylimits_{\lambda \to \infty} 0 \yesnumber{f(lambda)}
        \end{align*}
    \end{proposition}
        
    The kinetic energy
    \begin{align*}
        \dfrac{1}{\hbar b } \Tr{\mathscr{L}_{\hbar, b} \gamma_{\psi_N}^{(1)}}
    \end{align*}
    will be bounded when we will take a sequence of minimizers of the $N$-body quantum energy. Recalling \cref{eq:hbar} and \cref{eq:l_b scaling},
    \begin{align*}
        b l_b = \mathcal{O}\prth{\hbar N l_b} = \mathcal{O}\prth{\hbar N^{\frac{1}{2}}} \gg 1
    \end{align*}
    so with \cref{eq:lambda scaling}
    \begin{align*}
        \prth{\hbar \lambda}^2 \ll \hbar \lambda \ll \hbar b \lambda l_b \ll 1 \yesnumber{lambda}
    \end{align*}
    Moreover, $\lambda \to \infty$ so the error terms in \cref{eq:SC approx} will be small.
    
    \begin{proof}[prop:SC energy]
        With \cref{eq:E reduced d} in mind, we start by computing the kinetic term. Inserting the resolution of identity \cref{eq:nR res id}, we have
        \begin{align*}
            \Tr{\mathscr{L}_{\hbar, b}\gamma_{\psi_N}^{(1)}} = \intr_{\mathbb{N}\times\Omega} \Tr{\mathscr{L}_{\hbar, b} g_\lambda(\bullet-R)\Pi_n g_\lambda(\bullet-R)\gamma_{\psi_N}^{(1)}} d\eta(n, R) 
        \end{align*}
        Now, we use $\mathscr{L}_{\hbar, b} \Pi_n= E_n \Pi_n$ by commuting $\mathscr{L}_{\hbar, b}$ with $g_\lambda(\bullet-R)$ to obtain
        \begin{align*}
            \Tr{\mathscr{L}_{\hbar, b}\gamma_{\psi_N}^{(1)}} 
            =& \Tr{\intr_{\mathbb{N}\times\Omega}E_n \Pi_{n,R} \gamma_{\psi_N}^{(1)} d\eta(n, R)} \\ 
            &+ \Tr{\gamma_{\psi_N}^{(1)}\intr_{\mathbb{N}\times\Omega} \sbra{\mathscr{L}_{\hbar, b}, g_\lambda(\bullet-R)}\Pi_n g_\lambda(\bullet-R) dR} \\
            =&  \intr_{\mathbb{N}\times\Omega} E_n m_{\psi_N}^{(1)}(n, R)d\eta(n, R) + \Tr{\gamma_{\psi_N}^{(1)}\intr_\Omega \sbra{\mathscr{L}_{\hbar, b}, g_\lambda(\bullet-R)} g_\lambda(\bullet-R)dR} \yesnumber{Tr kin}
        \end{align*}
        We compute
        \begin{align*}
            \sbra{\mathscr{P}_{\hbar, b}, g_\lambda(\bullet-R)} = i\hbar\nabla g_\lambda(\bullet-R)
        \end{align*}
        and
        \begin{align*}
            \sbra{\mathscr{L}_{\hbar, b}, g_\lambda(\bullet-R)} 
            &= \sbra{\mathscr{P}_{\hbar, b}, g_\lambda(\bullet-R)}\cdot\mathscr{P}_{\hbar, b} + \mathscr{P}_{\hbar, b} \cdot\sbra{\mathscr{P}_{\hbar, b}, g_\lambda(\bullet-R)} \\
            &= 2i\hbar \nabla g_\lambda(\bullet-R) \cdot \mathscr{P}_{\hbar, b} - \hbar^2 \Delta g_\lambda(\bullet-R) \\
            &= \mathscr{P}_{\hbar, b} \cdot 2i\hbar \nabla g_\lambda(\bullet-R) + \hbar^2 \Delta g_\lambda(\bullet-R) \yesnumber{com H g}
        \end{align*}
        inserting this in \cref{eq:Tr kin}, we find
        \begin{align*}
            \Tr{\mathscr{L}_{\hbar, b}\gamma_{\psi_N}^{(1)}} 
            =&  \intr_{\mathbb{N}\times\Omega} E_n m_{\psi_N}^{(1)}(n, R)d\eta(n, R)  \\
            &+ 2i \hbar \Tr{\gamma_{\psi_N}^{(1)}\mathscr{P}_{\hbar, b}\cdot\intr_\Omega \nabla g_\lambda(\bullet-R)g_\lambda(\bullet-R)dR} \\
            &+ \hbar^2 \Tr{\gamma_{\psi_N}^{(1)}\intr_\Omega \Delta g_\lambda(\bullet-R) g_\lambda(\bullet-R)dR}
        \end{align*}
        But because $g$ has a fixed $L^2$ norm and is periodic
        \begin{align*}
            \nabla \intr_\Omega g_\lambda(\bullet-R)^2dR = 0 = 2\intr_\Omega \nabla g_\lambda(\bullet-R) g_\lambda(\bullet-R)dR
        \end{align*}
        Moreover
        \begin{align*}
            \intr_\Omega \Delta g_\lambda(\bullet-R) g_\lambda(\bullet-R)dR 
            = - \intr_\Omega \prth{\nabla g_\lambda}^2 
            = - \lambda^4 \intr_\Omega \prth{\nabla g(\lambda x)}^2 dx
            = - \lambda^2 \norm{\nabla g}_{L^2}^2 \yesnumber{grad g lambda}
        \end{align*}
        Therefore
        \begin{align*}
            \Tr{\mathscr{L}_{\hbar, b}\gamma_{\psi_N}^{(1)}} = \intr_{\mathbb{N}\times\Omega}E_n m_{\psi_N}^{(1)}(n, R) d\eta(n, R) - (\hbar\lambda)^2 \norm{\nabla g}_{L^2}^2 \yesnumber{kinetic energy term}
        \end{align*}
        If we take a $k$ variable potential $V_k \in L^1\prth{\Omega^k}$,
        \begin{align*}
            \Tr{V_k \gamma_{\psi_N}^{(k)}} = \intr_{\Omega^k}\gamma_{\psi_N}^{(k)}\prth{x_{1:k}; x_{1:k}}V_k(x_{1:k})dx_{1:k} = \intr_{\Omega^k}\rho_{\psi_N}^{(k)}V_k
        \end{align*}
        To express this in terms of Husimi functions we use \cref{eq:husimi to density}:
        \begin{align*}
            \Tr{V_k \gamma_{\psi_N}^{(k)}} &= \intr_{\Omega^k} \rho_{m_{\gamma_N}}^{(k)} V_k  + \intr_{\Omega^k} \prth{\rho_{\psi_N}^{(k)} - (g^2_\lambda)^{\otimes k} * \rho_{\psi_N}^{(k)}}V_k \\
            &=\intr_{\Omega^k} \rho_{m_{\gamma_N}}^{(k)} V_k  + \intr_{\Omega^k} \rho_{\psi_N}^{(k)}\prth{V_k - (g^2_\lambda)^{\otimes k} * V_k}
        \end{align*}
        Thus applying \cref{eq:E reduced d} and using \cref{eq:kinetic energy term},
        \begin{align*}
            \dfrac{\bk{\psi_N}{\mathcal{H}_N\psi_N}}{N} &=  \Tr{(\mathscr{L}_{\hbar, b}+V)\gamma_{\psi_N}^{(1)}} + \Tr{w \gamma_{\psi_N}^{(2)}} \\
            &= \intr_{\mathbb{N}\times\Omega} E_n m_{\psi_N}^{(1)}(n, R)d\eta(n, R) + \intr_\Omega \rho_{\psi_N}^{(1)} V + \intr_{\Omega^2} \rho_{\psi_N}^{(2)} w -  (\hbar \lambda)^2  \norm{\nabla g}_{L^2}^2 \\
            &=  \mathcal{E}_{sc, \hbar b}\sbra{m_{\psi_N}} + \intr_\Omega \rho_{\psi_N}^{(1)}\sbra{V - g^2_\lambda*V} + \intr_{\Omega^2}\rho_{\psi_N}^{(2)} \sbra{w - (g^2_\lambda)^{\otimes 2} * w} - (\hbar \lambda)^2  \norm{\nabla g}_{L^2}^2
        \end{align*}
        Using $V, w \in L^2(\Omega)$ and the fact that $w$ and thus 
        \begin{align*}
           (g^2_\lambda)^{\otimes 2} * w (x, y) = \iint\displaylimits_{\Omega^2} g^2_\lambda(z)g^2_\lambda(t)w(x - y + t - z)dzdt
        \end{align*}
        only depends on $x - y$ we can use the kinetic energy inequalities \cref{eq:Kin E ineq} and \cref{eq:Kin E ineq w} to control the errors terms:
        \begin{align*}
            \abs{\dfrac{\bk{\psi_N}{\mathcal{H}_N\psi_N}}{N} -  \mathcal{E}_{sc, \hbar b}\sbra{m_{\psi_N}}} 
            \le& \abs{\intr_\omega \rho_{\psi_N}^{(1)}\sbra{V - g_\lambda^2 * V}} + \abs{\intr_{\Omega^2}\rho_{\psi_N}^{(2)} \sbra{w - (g^2_\lambda)^{\otimes 2} * w}}  \\
            &+ (\hbar \lambda)^2  \norm{\nabla g}_{L^2}^2\\
            \le& \dfrac{C}{\hbar b } \Tr{\mathscr{L}_{\hbar, b} \gamma_N^{(1)}} f(\lambda) + (\hbar \lambda)^2 \norm{\nabla g}_{L^2}^2
        \end{align*}
        and we have        
        \begin{align*}
            f(\lambda) \limit\displaylimits_{\lambda\to\infty} 0
        \end{align*}
    \end{proof}

    \section{Mean field limit}\label{sec:MF}

    In \cref{sec:SC}, we went from the quantum $N$-body energy to the semi-classical energy \cref{eq:SC energy 2} (\cref{prop:SC energy}). The last step needed to obtain the limit models \cref{eq:SC energy} and \cref{eq:qLL} out of \cref{eq:SC energy 2} and \cref{eq:qLL energy 2} is to remove correlations. Indeed we see that for $m \in L^1\prth{\mathbb{N}\times \Omega}$ and $\rho \in L^1(\Omega)$
    \begin{align*}
        \mathcal{E}_{sc, \hbar b}\sbra{m^{\otimes 2}} = \mathcal{E}_{sc, \hbar b}\sbra{m} \qquad
        \mathcal{E}_{qLL}\sbra{\rho^{\otimes 2}} =  \mathcal{E}_{qLL}\sbra{\rho}
    \end{align*}
    For fermionic states there are always some correlations due to anti-symmetry. Therefore the objective of this section is to prove that all other correlations are negligible, that is to say justifying the mean field approximation. The main tools are Lieb's variational principle (\cref{th:PLieb}) for the energy upper bound in \cref{sec:UP E bound} and the De Finetti \cref{th:DeFinetti} for the lower bound in \cref{sec:LOW E bound}.

\subsection{Energy upper bound}\label{sec:UP E bound}

    In this part we prove the energy upper bound:

    \begin{proposition}[title=Upper energy bound, label=EUB]
        \begin{align*}
            \dfrac{E_N^0}{N} \le \hbar b E^{q,r} + E_{V}^{q,r} + E_{w}^{q,r} + \mathcal{E}_{qLL}\sbra{\rho}
                + \hbar b\mathcal{O}\prth{1 - \dfrac{d(q+r)}{N}} + \mathcal{O}\prth{f(\lambda)} + \mathcal{O}\prth{\hbar b \lambda l_b} 
        \end{align*}
    \end{proposition}
    
    For this result, we use Hartree-Fock theory, meaning we restrict the energy to Slater determinants. Computations are simplified by Wick's \cref{th:wick}. Hartree-Fock theory can be extended to general one body operators (see \cref{not:HF theory}), and using Lieb's variational principle (\cref{th:PLieb}) one can show that the theory still provides an approximate upper bound for the $N$-body quantum energy (\cref{prop:Lieb prop}). Then we conclude by approaching the semi-classical energy with the Hartree-Fock energy (\cref{prop:HF to SC}).
    
    \begin{notation}[title=Hartree Fock theory, label=HF theory]
        Let $s, t, u, v \in L^2(\Omega)$, if one define the exchange operator on $\mathcal{L}^1\prth{L^2\prth{\Omega^2}}$ as
        \begin{align*}
            \text{Ex} \ket{s \otimes t} \bra{u \otimes v} \coloneq \ket{s \otimes t} \bra{v \otimes u} \yesnumber{Ex}
        \end{align*}
        Let $\gamma \in \mathcal{L}^1\prth{L^2(\Omega)}$, define 
        \begin{align*}
            \gamma_2 \coloneq \dfrac{N}{N-1}\prth{1 - \text{Ex}}\gamma^{\otimes 2} \yesnumber{gamma2}
        \end{align*}
        Define the Hartree-Fock energy
        \begin{align*}
            \mathcal{E}_{HF}\sbra{\gamma} \coloneq \Tr{\prth{\mathscr{L}_{\hbar, b} + V}\gamma} +  \Tr{w \gamma_2} \yesnumber{HF energy}
        \end{align*}
    \end{notation}
    
    With Wick's theorem definitions \cref{eq:gamma2} and \cref{eq:HF energy} are actual statements for Slater determinants.

    \begin{theorem}[title=Wick's theorem, label=wick]
        Let $\psi_N = \dfrac{1}{\sqrt{N!}}\bigwedge_{j=1}^N \phi_j \in L^2_-(\Omega^N)$ with $(\phi_j)_j$ an orthonormal family, then
        \begin{align*}
            \gamma_{\psi_N}^{(1)} &= \dfrac{1}{N}\sum_{i = 1}^N \ket{\phi_i}\bra{\phi_i} \\
            \gamma_{\psi_N}^{(2)}
            &= \dfrac{N}{N-1}\prth{1 - \text{Ex}}\prth{\gamma_N^{(1)}}^{\otimes 2}
            = \dfrac{1}{N(N-1)}\sum_{i, j = 1}^N \ket{\phi_i \otimes \phi_j}\bra{\phi_i \otimes \phi_j - \phi_j \otimes \phi_i}
        \end{align*}
    \end{theorem}
    
    Thus for a Slater determinant $\gamma_{\psi_N}$
    \begin{align*}
        \prth{\gamma_{\psi_N}^{(1)}}_2 = \gamma_{\psi_N}^{(2)}
    \end{align*}
    and the Hartree-Fock energy is exactly what we obtain for the quantum $N$-body energy:
    \begin{align*}
        \mathcal{E}_{HF}\sbra{\gamma_{\psi_N}^{(1)}} = \Tr{\prth{\mathscr{L}_{\hbar, b} + V}\gamma_{\psi_N}^{(1)}} +  \Tr{w \gamma_{\psi_N}^{(2)}}
    \end{align*}
    Lieb's theorem \cite{Lieb} extends the usual variational principle for operators of the form \cref{eq:gamma2}.

    \begin{theorem} [title=Lieb's variational principle, label=PLieb]
        Let $\gamma \in \mathcal{L}^1\prth{L^2(\Omega)}$ satisfying
        \begin{align*}
        &\Tr{\gamma} = 1 \qquad 0\le \gamma \le \dfrac{1}{N}    
        \end{align*}
        there exits an $N$-body density matrix $\gamma_N$ and a positive operator $L_2$ such that
        \begin{align*}
            \gamma_N^{(1)} = \gamma \qquad
            \gamma_N^{(2)} = \gamma_2 - L_2
        \end{align*}
    \end{theorem}
    
    We start with Lieb's variational principle to get an energy upper bound in term of the operator $\gamma_2$. An important remark here is that we don't assume that the interaction potential is repulsive to get the upper bound as it is usually done when dealing with Lieb's variational principle. The reason why we were able to relax the assumption $w \ge 0$ is independent of our specific model and comes from the fact that $L_2$ looks like an exchange term in the mean field limit since (see \cref{eq:trL2})
    \begin{align*}
        \Tr{L_2} \le \dfrac{1}{N-1}
    \end{align*}
    Hence the contributions coming from this term in the computation \cref{eq:HF estimate} will be treated as error terms. Lieb's variational principle has also been recently generalised in \cite{Bach}.

    \begin{proposition}[label=Lieb prop]
        Let $\gamma \in \mathcal{L}^1\prth{L^2(\Omega)}$ such that $\Tr{\gamma} = 1$ and $0\le \gamma \le \dfrac{1}{N}$, then
        \begin{align*}
            \dfrac{E_N^0}{N} \le& \mathcal{E}_{HF}\sbra{\gamma} + \dfrac{\Tr{\mathscr{L}_{\hbar, b}\gamma}}{\hbar b}\mathcal{O}\prth{l_b}
        \end{align*}
    \end{proposition}
    
    \begin{proof}
        First we prove a lower bound for the interaction term. Using The Gagliardo-Nirenberg inequality for $\psi \in L^2(\Omega)$,
        \begin{align*}
            \norm{\psi}_{L^4}^2 \le C_{GN} \prth{\sqrt{\norm{\psi}_{L^2}  \norm{\nabla \psi}}_{L^2} + \norm{\psi}_{L^2}}
        \end{align*}
        along with Hölder's, Young's and Kato's \cref{eq:kato H1} inequalities,
        \begin{align*}
            \abs{\bk{\psi}{\mathcal{V}\psi}} 
            &\le \norm{\psi}_{L^4} \norm{\mathcal{V}\psi}_{L^{\frac{4}{3}}} 
            \le \norm{\mathcal{V}}_{L^2} \norm{\psi}_{L^4}^2 
            \le C_{GN} \norm{\mathcal{V}}_{L^2}\prth{\norm{\psi}_{L^2}  \norm{\nabla \abs{\psi}}_{L^2} + \norm{\psi}_{L^2}^2} \\
            &\le C_{GN}\norm{\mathcal{V}}_{L^2}\prth{\dfrac{1}{\hbar}\norm{\psi}_{L^2} \norm{\mathscr{P}_{\hbar, b}\psi}_{L^2} + \norm{\psi}_{L^2}^2} \\
            &\le C_{GN}\norm{\mathcal{V}}_{L^2}\prth{\epsilon  \norm{\mathscr{P}_{\hbar, b}\psi}_{L^2}^2 + \prth{1 + \frac{1}{4\epsilon\hbar^2}}\norm{\psi}_{L^2}^2}
        \end{align*}
        So for $\psi_2 \in L^2(\Omega) \otimes \text{Dom}\prth{\mathscr{L}_{\hbar, b}}$,
        \begin{align*}
            \abs{\bk{\psi_2}{w\psi_2}} &\le \intr_{\Omega^2} \abs{w(x - y)} \abs{\psi_2(x, y)}^2 dx dy 
            \le \norm{w}_{L^2} \intr_\Omega \norm{\psi_2(x, \bullet)}_{L^4}^2 \\
            &\le C_{GN}\norm{w}_{L^2} \intr_\Omega \prth{\epsilon  \norm{\mathscr{P}_{\hbar, b}\psi_2(x, \bullet)}_{L^2}^2 + \prth{1 + \frac{1}{4\epsilon\hbar^2}}\norm{\psi_2(x, \bullet)}_{L^2}^2} dx \yesnumber{GN+kato}\\
            &= C_{GN}\norm{w}_{L^2}\prth{\epsilon \norm{1\otimes \mathscr{P}_{\hbar, b}\psi_2}_{L^2}^2 + \prth{1 + \dfrac{1}{4\epsilon \hbar^2}}\norm{\psi_2}_{L^2}^2}
        \end{align*}
        Thus
        \begin{align*}
            \bk{\psi_2}{\prth{C_{GN}\norm{w}_{L^2}\epsilon \prth{\text{Id}_{L^2\prth{\Omega}} \otimes \mathscr{L}_{\hbar, b}} + w}  \psi_2} 
            =& C_{GN}\norm{w}_{L^2} \epsilon \norm{1\otimes \mathscr{P}_{\hbar, b}\psi_2}_{L^2}^2 + \bk{\psi_2}{w\psi_2} \\
            \ge& - C_{GN}\norm{w}_{L^2} \prth{1 + \dfrac{1}{4\epsilon \hbar^2}}\norm{\psi_2}_{L^2}^2
        \end{align*}
        and
        \begin{align*}
            \epsilon C_{GN}\norm{w}_{L^2} \prth{\text{Id}_{L^2\prth{\Omega}} \otimes \mathscr{L}_{\hbar, b}} + w \ge - C_{GN}\norm{w}_{L^2} \prth{1 + \dfrac{1}{4\epsilon \hbar^2}} \yesnumber{C_N}
        \end{align*}
        Let $\gamma_N$ and $L_2$ be the operators in \cref{th:PLieb}. Now we use \cref{eq:E reduced d}, and \cref{eq:C_N}:
        \begin{align*}
            \dfrac{E_N^0}{N} \le& \dfrac{\Tr{\mathscr{H}_{N}\gamma_N}}{N} 
            = \Tr{\prth{\mathscr{L}_{\hbar, b} +V} \gamma_N^{(1)}} + \Tr{w \gamma_N^{(2)}} \\
            =& \Tr{\prth{\mathscr{L}_{\hbar, b} +V} \gamma} + \Tr{w \prth{\gamma_2 -L_2}} = \mathcal{E}_{HF}\sbra{\gamma} - \Tr{w L_2}\\
            \le& \mathcal{E}_{HF}\sbra{\gamma} + C_{GN}\norm{w}_{L^2} \prth{\prth{1 + \dfrac{1}{4\epsilon \hbar^2}}\Tr{L_2} + \epsilon \Tr{\prth{\text{Id}_{L^2\prth{\Omega}} \otimes \mathscr{L}_{\hbar, b}} L_2}} \yesnumber{HF estimate}
        \end{align*}
        To conclude we need to estimate the error terms. If $A$ is an operator on $L^2(\Omega)$ it follows from \cref{eq:Ex} that
        \begin{align*}
            \Tr{\prth{\text{Id}_{L^2\prth{\Omega}}\otimes A}\text{Ex}\gamma^{\otimes 2}} = \Tr{A \gamma^2} \yesnumber{Tr and Ex}
        \end{align*}
        Indeed, if we decompose $\gamma$ in an orthonormal family:
        \begin{align*}
            \gamma \eqcolon \sum_{i\in\mathbb{N}}\lambda_i \ket{u_i}\bra{u_i}
        \end{align*}
        then
        \begin{align*}
             \Tr{\prth{\text{Id}_{L^2\prth{\Omega}}\otimes A}\text{Ex}\gamma^{\otimes 2}} 
             =& \sum_{i, j\in \mathbb{N}} \lambda_i \lambda_j \Tr{\text{Id}_{L^2\prth{\Omega}}\otimes A  \ket{u_i \otimes u_j} \bra{u_j \otimes u_i}} \\
             =& \sum_{i, j\in \mathbb{N}} \lambda_i \lambda_j \Tr{\prth{\ket{u_i}\bra{u_j}} \otimes \prth{A \ket{u_j}\bra{u_i}}} \\
             =& \sum_{i, j\in \mathbb{N}} \lambda_i \lambda_j \Tr{\ket{u_i}\bra{u_j}}\Tr{A \ket{u_j}\bra{u_i}}
             = \sum_{i\in \mathbb{N}} \lambda_i^2 \Tr{A \ket{u_i}\bra{u_i}} \\
             =& \Tr{A \gamma^2}
        \end{align*}
        Taking $A \coloneq \text{Id}_{L^2\prth{\Omega}}$, we obtain
        \begin{align*}
            \Tr{\text{Ex}\gamma^{\otimes 2}} = \Tr{\gamma^2}
        \end{align*}
        and since $\gamma$ is positive, with \cref{eq:gamma2} we can estimate
        \begin{align*}
            \Tr{L_2} 
            &= \Tr{\gamma_2} - \Tr{\gamma_N^{(2)}} 
            = \dfrac{N}{N-1}\Tr{\gamma^{\otimes 2} - \text{Ex}\gamma^{\otimes 2}} - 1 = \dfrac{N}{N-1} - \dfrac{N}{N-1}\Tr{\gamma^2} - 1 \\
            &\le \dfrac{1}{N-1} \yesnumber{trL2}
        \end{align*}
        If $\epsilon \to 0$, we can control the first error term in \cref{eq:HF estimate} with
        \begin{align*}
            0 \le \prth{1 + \dfrac{1}{4 \epsilon \hbar ^2}} \Tr{L_2} \le \dfrac{C}{N\epsilon \hbar^2} \yesnumber{HF error 1}
        \end{align*}
        For the second error term, using \cref{th:PLieb}, \cref{eq:gamma2} and \cref{eq:Tr and Ex} for $A \coloneq \mathscr{L}_{\hbar, b}$,
        \begin{align*}
            0 \le& \Tr{\prth{\text{Id}_{L^2\prth{\Omega}} \otimes \mathscr{L}_{\hbar, b}} L_2} 
            = \Tr{\prth{\text{Id}_{L^2\prth{\Omega}} \otimes \mathscr{L}_{\hbar, b}} \prth{\gamma_2 - \gamma_N^{(2)}}} \\
            =& \dfrac{N}{N - 1}\Tr{\prth{\text{Id}_{L^2\prth{\Omega}} \otimes \mathscr{L}_{\hbar, b}} \prth{1 - \text{Ex}}\gamma^{\otimes 2}} - \Tr{\prth{\text{Id}_{L^2\prth{\Omega}} \otimes \mathscr{L}_{\hbar, b}} \gamma_N^{(2)}} \\
            =& \dfrac{N}{N-1}\Tr{\mathscr{L}_{\hbar, b} \gamma} - \dfrac{N}{N - 1}\Tr{\prth{\text{Id}_{L^2\prth{\Omega}} \otimes \mathscr{L}_{\hbar, b}}\text{Ex}\gamma^{\otimes 2}} - \Tr{\mathscr{L}_{\hbar, b} \gamma}\\
            =& \dfrac{1}{N-1}\Tr{\mathscr{L}_{\hbar, b} \gamma} - \dfrac{N}{N - 1}\Tr{\mathscr{L}_{\hbar, b} \gamma^2}
            \le  \dfrac{1}{N-1}\Tr{\mathscr{L}_{\hbar, b} \gamma}
        \end{align*}
        When the kinetic energy is minimised $\Tr{\mathscr{L}_{\hbar, b}\gamma}$ is of order $\hbar b$ so we estimate the second error term in \cref{eq:HF estimate} with:
        \begin{align*}
            0 \le \epsilon \Tr{\prth{\text{Id}_{L^2\prth{\Omega}} \otimes \mathscr{L}_{\hbar, b}} L_2} \le C \dfrac{\epsilon\hbar b}{N}\cdot\dfrac{\Tr{\mathscr{L}_{\hbar, b}\gamma}}{\hbar b} \yesnumber{HF error 2}
        \end{align*}
        We optimise in $\epsilon$ so the bounds in \cref{eq:HF error 1} and \cref{eq:HF error 2} are of the same order:
        \begin{align*}
            \dfrac{1}{N\epsilon \hbar^2} = \dfrac{\epsilon\hbar b}{N} \implies \epsilon = \dfrac{1}{\sqrt{\hbar^3 b}} = N^{2\delta - \frac{1}{2}} = o(1) \implies \dfrac{1}{N\epsilon \hbar^2} = \dfrac{\epsilon\hbar b}{N} = \dfrac{1}{l_b N} = \mathcal{O}(l_b) \yesnumber{EUP epsilon choice}
        \end{align*}
        so \cref{eq:HF estimate} becomes
        \begin{align*}
            \dfrac{E_N^0}{N} 
            \le \mathcal{E}_{HF}\sbra{\gamma} + \prth{1 + \dfrac{\Tr{\mathscr{L}_{\hbar, b}\gamma}}{\hbar b}}\mathcal{O}\prth{l_b}
            = \mathcal{E}_{HF}\sbra{\gamma} + \dfrac{\Tr{\mathscr{L}_{\hbar, b}\gamma}}{\hbar b}\mathcal{O}\prth{l_b}
        \end{align*}
    \end{proof}

    Recalling definitions \cref{eq:SC energy} and \cref{eq:f(lambda)}, we now go from the Hartree-Fock energy to the semi-classical energy.
    
    \begin{proposition}[title=Semi-classical approximation of Hartree-Fock energy, label=HF to SC]
        Let $n_0 \in \mathbb{N}$, $m \in L^1(\mathbb{N}\times\Omega)$ such that $\forall n > n_0, m(n, \bullet) = 0$ and
        \begin{align*}
            0\le m \le \dfrac{1}{2\pi l_b^2 N} \yesnumber{pp in prop}
        \end{align*}
        then
        \begin{align*}
            \mathcal{E}_{HF}\sbra{\gamma_m} = \mathcal{E}_{sc, \hbar b}\sbra{m} + \mathcal{O}\prth{f(\lambda)} + \mathcal{O}\prth{\hbar b \lambda l_b}
        \end{align*}
    \end{proposition}
    
    \begin{proof}
    We start by proving that we recover the semi-classical functional from the direct terms. We compute the kinetic term using the commutation relation \cref{eq:com H g} and \cref{cor:tr pi}:
        \begin{align*}
            \Tr{\mathscr{L}_{\hbar, b}\gamma_m} =& 2\pi l_b^2 \intr_{\mathbb{N}\times\Omega}m(X) \Tr{\mathscr{L}_{\hbar, b} \Pi_X} d\eta(X) \\
            =& 2\pi l_b^2 \intr_{\mathbb{N\times\Omega}}m(X)E_n\Tr{\Pi_X}d\eta(X) \\
            &+ 2\pi l_b^2 \intr_{\mathbb{N\times\Omega}}m(n,R)\Tr{\sbra{\mathscr{L}_{\hbar, b}, g_\lambda(\bullet - R)} \Pi_n g_\lambda(\bullet - R)}d\eta(n, R) \\
            =& \intr_{\mathbb{N}\times\Omega}E_n m(X)d\eta(X) + \mathcal{O}(\hbar b l_b) \\
            &+ 2\pi l_b^2 \intr_{\mathbb{N\times\Omega}}m(n,R)\Tr{2i\hbar \nabla g_\lambda(\bullet-R) \mathscr{P}_{\hbar, b} \Pi_n g_\lambda(\bullet - R)}d\eta(n, R) \\
            &- 2\pi l_b^2 \intr_{\mathbb{N\times\Omega}}m(n,R)\Tr{\hbar^2 \Delta g_\lambda(\bullet-R) \Pi_n g_\lambda(\bullet - R)}d\eta(n, R) \yesnumber{DT compt 1}
        \end{align*}
        Using \cref{eq:kernel approx 1}, $\exists \mathcal{E}:\mathbb{N}\times \Omega \to \mathbb{R}$ such that
        \begin{align*}
            &2\pi l_b^2 \Pi_n(x, x) = 1 + l_b  \mathcal{E}(n, x) \\
            &\abs{\mathcal{E}(n , x)} \le C(n)
        \end{align*}
        With \cref{eq:grad g lambda},
        \begin{align*}
            &- 2\pi l_b^2 \intr_{\mathbb{N\times\Omega}}m(n,R)\Tr{\hbar^2 \Delta g_\lambda(\bullet-R) \Pi_n g_\lambda(\bullet - R)}d\eta(n, R) \\
            =& - \hbar^2 \intr_{\mathbb{N}\times\Omega} m(n, R) \prth{\intr_\Omega \Delta g_\lambda(x-R)\prth{1 + l_b \mathcal{E}(n, x)}g_\lambda(x - R) dx}d\eta(n, R)\\
            =& \prth{\hbar \lambda}^2 \norm{\nabla g}_{L^2}^2 \norm{m}_{L^1} - \hbar^2 l_b \intr_{\mathbb{N}\times\Omega} m(n, R) \prth{\intr_\Omega \lambda^3 \Delta g(\lambda x) \mathcal{E}(n, x + R)\lambda g(\lambda x) dx}d\eta(n, R) \\
            =& \prth{\hbar \lambda}^2 \norm{\nabla g}_{L^2}^2 \norm{m}_{L^1} + \prth{\hbar\lambda}^2 \mathcal{O}\prth{l_b} = \mathcal{O}\prth{\prth{\hbar \lambda}^2} \yesnumber{DT compt 2}
        \end{align*}
        And by \cref{eq:kernel approx 2}, $\exists \widetilde{\mathcal{E}}:\mathbb{N}\times \Omega \to \mathbb{R}$ such that
        \begin{align*}
            &\mathscr{P}_{\hbar, b}\Pi_n(x, x) = \dfrac{b}{l_b} C(n) + b \widetilde{\mathcal{E}}(n, x) \\
            &\abs{\mathcal{E}(n , x)} \le \widetilde{C}(n)
        \end{align*}
        so
        \begin{align*}
            &2\pi l_b^2 \intr_{\mathbb{N}\times\Omega}m(n,R)\Tr{2i\hbar \nabla g_\lambda(\bullet-R) \mathscr{P}_{\hbar, b} \Pi_n g_\lambda(\bullet - R)}d\eta(n, R) \\
            =& 4i\pi l_b^2 \hbar \intr_{\mathbb{N}\times\Omega} m(n, R) \prth{\intr_\Omega \nabla g_\lambda(x-R) \prth{C(n) \dfrac{b}{l_b} + b\widetilde{\mathcal{E}}(n, R)}g_\lambda(x - R) dx} d\eta(n, R) \\
            =& \mathcal{O}\prth{\hbar b \lambda l_b} \yesnumber{DT compt 3}
        \end{align*}
        Inserting \cref{eq:DT compt 2} and \cref{eq:DT compt 3} in \cref{eq:DT compt 1}, we obtain
        \begin{align*}
            \Tr{\mathscr{L}_{\hbar, b}\gamma_m} = \intr_{\mathbb{N}\times\Omega}E_n m(X)d\eta(X) +\mathcal{O}\prth{\hbar b \lambda l_b} + \mathcal{O}\prth{\prth{\hbar \lambda}^2} \yesnumber{DT compt 4}
        \end{align*}
        Let $k \in \mathbb{N^*}$ and $W_k \in L^2\prth{\Omega^k}$, with the Fubini theorem
        \begin{align*}
            \Tr{W_k \gamma_m^{\otimes k}} 
            =& (2\pi l_b^2)^k \intr_{\prth{\mathbb{N}\times\Omega}^k} m^{\otimes k}(X_{1:k})\Tr{W_k \bigotimes_{i=1}^k\Pi_{X_i}}d\eta^{\otimes k}\prth{X_{1:k}} \\
            =& \prth{2\pi l_b^2}^k \intr_{\prth{\mathbb{N}\times\Omega}^k} m^{\otimes k}(X_{1:k}) \intr_{\Omega^k}W_k(x_{1:k})\prth{\bigotimes_{i=1}^k \Pi_{X_i}}(x_{1:k}, x_{1:k}) dx_{1:k} d\eta^{\otimes k}\prth{X_{1:k}} \\
            =& \intr_{\Omega^k}  W_k(x_{1:k})  \prth{\prod_{i=1}^k 2\pi l_b^2 \intr_{\prth{\mathbb{N}\times\Omega}}  m(X) \Pi_X(x_i, x_i)d\eta\prth{X}} dx_{1:k} \\
            =& \intr_{\Omega^k}  W_k(x_{1:k}) \prth{\prod_{i=1}^k \intr_{\prth{\mathbb{N}\times\Omega}}  m(n, R) g_\lambda^2(x_i- R)\prth{1 + l_b \mathcal{E}(n, x_i)}d\eta\prth{n, R}}dx_{1:k} \\
            =& \intr_{\Omega^k} W_k \prth{\rho_m^{\otimes k}* (g_\lambda^2)^{\otimes k}} dx \\
            &+ l_b \intr_{\Omega^k}  W_k(x_{1:k}) \prth{\prod_{i=1}^k \intr_\Omega g_\lambda^2(x_i- R) \sum_{n=0}^{n_0} m(n, R)  \mathcal{E}(n, x_i)dR}dx_{1:k}
        \end{align*}
        $m$ has finitely many filled Landau level so with the Pauli principle \cref{eq:pp in prop}, $\rho_m \in L^\infty(\Omega)$ and
        \begin{align*}
            \Tr{W_k \gamma_m^{\otimes k}} = \intr_{\Omega^k} W_k\rho_m^{\otimes k} + \mathcal{O}\prth{\norm{W_k - W_k * (g_\lambda^2)^{\otimes k}}_{L^1}} + \mathcal{O}\prth{l_b} \yesnumber{pot HFSC}
        \end{align*}
        Now we need to control the exchange term. It follows from \cref{eq:Ex} that
        \begin{align*}
            \text{Ex}\gamma_m^{\otimes 2}(x, y; z, t) = \gamma_m(x, t)\gamma_m(y, z)
        \end{align*}
        so with \cref{eq:GN+kato} for $\gamma_m \in L^2(\Omega)\otimes \text{Dom}\prth{\mathscr{L}_{\hbar, b}}$ as an integral kernel,
        \begin{align*}
            \abs{\Tr{w\text{Ex}\gamma_m^{\otimes 2}}}
            =& \abs{\intr_{\Omega^2} w(x - y) \abs{\gamma_m(x, y)}^2 dxdy } \\
            \le& C_{GN}\norm{w}_{L^2} \intr_\Omega \prth{\epsilon  \norm{\mathscr{P}_{\hbar, b}\gamma_m(x, \bullet)}_{L^2}^2 + \prth{1 + \frac{1}{4\epsilon\hbar^2}}\norm{\gamma_m(x, \bullet)}_{L^2}^2} dx \yesnumber{HFSC 1}
        \end{align*}
        With an integration by part,
        \begin{align*}
            \intr_\Omega \norm{\mathscr{P}_{\hbar, b}\gamma_m(x, \bullet)}_{L^2}^2 dx
            =& \intr_{\Omega^2} \mathscr{P}_{\hbar, b}\gamma_m(x, \bullet)(y) \cdot \overline{\mathscr{P}_{\hbar, b}\gamma_m(x, \bullet)(y)} dxdy \\
            =& \intr_{\Omega^2} \mathscr{L}_{\hbar, b} \gamma_m(x, \bullet)(y) \overline{\gamma_m(x, y)} dx dy
            = \intr_{\Omega^2} \overline{\gamma_m}(x, y) \mathscr{L}_{\hbar, b} \overline{\gamma_m}(\bullet, x)(y)dxdy \\
            =& \intr_{\Omega^2}  \overline{\gamma_m}(x, y) \prth{\mathscr{L}_{\hbar, b}  \overline{\gamma_m}}(y, x) dxdy
            = \Tr{\overline{\gamma_m}\mathscr{L}_{\hbar, b}\overline{\gamma_m}}
        \end{align*}
        Inserting this in \cref{eq:HFSC 1}, using the cyclicity of the trace we get
        \begin{align*}
            \abs{\Tr{w\text{Ex}\gamma_m^{\otimes 2}}} = \abs{\Tr{w\text{Ex}\overline{\gamma_m}^{\otimes 2}}}
            \le& C_{GN}\norm{w}_{L^2} \prth{\epsilon\Tr{\mathscr{L}_{\hbar, b}\gamma_m^2} +  \prth{1 + \frac{1}{4\epsilon\hbar^2}} \Tr{\gamma_m^2}} \\
            \le& \dfrac{C_{GN}\norm{w}_{L^2}}{N}\prth{\epsilon\Tr{\mathscr{L}_{\hbar, b}\gamma_m} + \prth{1 + \frac{1}{4\epsilon\hbar^2}}\Tr{\gamma_m}}
        \end{align*}
        With \cref{eq:DT compt 4}, $\Tr{\mathscr{L}_{\hbar, b}\gamma_m} = \mathcal{O}\prth{\hbar b}$ and using \cref{lem:equiv Husimi admissible}, $\Tr{\gamma_m} = \norm{m}_{L^1} + \mathcal{O}(l_b)$ so the choice of $\epsilon$ is the same as in \cref{eq:EUP epsilon choice} thus
        \begin{align*}
            \Tr{w\text{Ex}\gamma_m^{\otimes 2}} = \mathcal{O}(l_b)
        \end{align*}
        We conclude with \cref{eq:HF energy} and \cref{eq:gamma2} then \cref{eq:DT compt 4} and \cref{eq:pot HFSC} applied to $V$ and $w$
        \begin{align*}
            \mathcal{E}_{HF}\sbra{\gamma_m} =& \Tr{\mathscr{L}_{\hbar, b}\gamma_m} + \Tr{V\gamma_m} + \dfrac{N}{N-1}\Tr{w\gamma_m^{\otimes 2}} +  \dfrac{N}{N-1}\Tr{w \text{Ex}\gamma_m^{\otimes 2}} \\
            =& \intr_{\mathbb{N}\times\Omega}E_n m(X)d\eta(X) + \intr_\Omega V\rho_m  + \dfrac{N}{N-1}\intr_{\Omega^2} w\rho_m^{\otimes 2} + \mathcal{O}\prth{\norm{V - V * (g_\lambda^2)}_{L^1}} \\
            &+ \dfrac{N}{N-1}\mathcal{O}\prth{\norm{w - w * (g_\lambda^2)^{\otimes 2}}_{L^1}} + \mathcal{O}\prth{l_b}
            + \mathcal{O}\prth{\hbar b \lambda l_b} +  \mathcal{O}\prth{\prth{\hbar \lambda}^2}
        \end{align*}
        Recalling \cref{eq:f(lambda)}, the semi-classical energy expression \cref{eq:SC energy}, \cref{eq:lambda} and $\hbar b \lambda \gg 1$,
        \begin{align*}
            \mathcal{E}_{HF}\sbra{\gamma_m} = \mathcal{E}_{sc, \hbar}\sbra{m} + f(\lambda) + \mathcal{O}\prth{\hbar b \lambda l_b}
        \end{align*}
    \end{proof}
    
    With the notation of Equation \cref{eq:saturated low LL}, we would like to define a one body operator with saturated low Landau levels:
    \begin{align*}
        \gamma_\rho \coloneq \dfrac{L^2(q+r)}{N} \intr_{\Omega \times \mathbb{N}} m_\rho(X)\Pi_X d\eta(X)
    \end{align*}
    We need to prove that the direct term gives the limit model for qLL and to control the exchange terms. But we cannot apply directly Lieb's principle because with \cref{lem:equiv Husimi admissible} we have an error on the trace
    \begin{align*}
        \Tr{\gamma_\rho} = 1 +o(1) \text{ and } 0\le \gamma_\rho \le \frac{1}{N}
    \end{align*}
    To cure this we modify $m_\rho$ slightly in the following technical Lemma:
    
    \begin{lemma}[title=Corrected Husimi function, label=corrected Husimi]
        Let $n_0 \in \mathbb{N}$, $m \in L^1(\mathbb{N}\times\Omega)$ such that $\forall n > n_0, m(n, \bullet) = 0, \norm{m}_{L^1} = 1 + o(1)$ and
        \begin{align*}
            0\le m \le \dfrac{1}{2\pi l_b^2 N}
        \end{align*}
        then there exist $\widetilde{m} \in L^1(\mathbb{N}\times\Omega), n_1 \in \mathbb{N}$ such that $\forall n > n_1, \widetilde{m}(n, \bullet) = 0$,
        \begin{align*}
            \Tr{\gamma_{\widetilde{m}}} = 1,
            \quad 0\le \gamma_{\widetilde{m}} \le \dfrac{1}{N}
        \end{align*}
        and
        \begin{align*}
            \mathcal{E}_{sc, \hbar b}\sbra{\widetilde{m}} = \mathcal{E}_{sc, \hbar b}\sbra{m} + \mathcal{O}\prth{\hbar b l_b} + \mathcal{O}\prth{\hbar b\prth{1 - \norm{m}_{L^1}}} \yesnumber{SC approx m tilde}
        \end{align*}
    \end{lemma}

    The proof of this Lemma (see \cref{sec:appendix}) consists in moving some small amount of mass from one Landau level to another. Putting all together we obtain the upper bound.
    
    \begin{proof}[prop:EUB]
        Recalling \cref{eq:saturated low LL}, let $\rho \in \mathcal{D}_{qLL}$ and define
        \begin{align*}
            m_{\rho, N} \coloneq \dfrac{d(q+r)}{N} m_\rho \yesnumber{mrhotilde}
        \end{align*}
        then
        \begin{align*}
            &0 \le m_{\rho, N} \le \frac{d}{L^2 N} = \dfrac{1}{2\pi l_b^2 N} \\
            &\intr_{\mathbb{N}\times \Omega} m_{\rho, N} d\eta = \dfrac{d(q+r)}{N} = 1 + o(1)
        \end{align*}
        We consider $\widetilde{m}_{\rho, N}$ the corrected Husimi function in \cref{lem:corrected Husimi} associated with $ m_{\rho, N}$, it satisfies
        \begin{align*}
            \mathcal{E}_{sc, \hbar b}\sbra{\widetilde{m}_{\rho, N}} = \mathcal{E}_{sc, \hbar b}\sbra{m_{\rho, N}} + \mathcal{O}\prth{\hbar b l_b} + \hbar b\mathcal{O}\prth{1 - \dfrac{d(q+r)}{N}} \yesnumber{ELB 1}
        \end{align*}
        and $\Tr{\gamma_{m_{\rho, N}}} = 1, 0\le \gamma_{m_{\rho, N}} \le \dfrac{1}{N}$. Moreover by \cref{eq:DT compt 4},
        \begin{align*}
            \Tr{\mathscr{L}_{\hbar, b}\gamma_{m_{\rho, N}}} = \mathcal{O}\prth{\hbar b}
        \end{align*}
        Thus, we can apply Propositions \cref{prop:Lieb prop}, \cref{prop:HF to SC} and \cref{eq:ELB 1}:
        \begin{align*}
            &\dfrac{E_N^0}{N}  \\
            \le& \mathcal{E}_{HF}\sbra{\gamma_{m_{\rho, N}}} + \mathcal{O}\prth{l_b}
            = \mathcal{E}_{sc, \hbar b}\sbra{\widetilde{m}_{\rho, N}} + \mathcal{O}\prth{f(\lambda)} + \mathcal{O}\prth{\hbar b \lambda l_b}\\
            =&  \mathcal{E}_{sc, \hbar b}\sbra{m_{\rho, N}} + \hbar b\mathcal{O}\prth{1 - \dfrac{d(q+r)}{N}} + \mathcal{O}\prth{f(\lambda)} + \mathcal{O}\prth{\hbar b \lambda l_b} \\
            =& \hbar b E^{q,r} + E_{V}^{q,r} + E_{w}^{q,r} + \mathcal{E}_{qLL}\sbra{\dfrac{d(q+r)}{N}\rho} + \hbar b\mathcal{O}\prth{1 - \dfrac{d(q+r)}{N}} + \mathcal{O}\prth{f(\lambda)} + \mathcal{O}\prth{\hbar b \lambda l_b} \\
            =& \hbar b E^{q,r} + E_{V}^{q,r} + E_{w}^{q,r} + \mathcal{E}_{qLL}\sbra{\rho}
                + \hbar b\mathcal{O}\prth{1 - \dfrac{d(q+r)}{N}} + \mathcal{O}\prth{f(\lambda)} + \mathcal{O}\prth{\hbar b \lambda l_b}
        \end{align*}
        For the last equality we use the estimate
        \begin{align*}
            &\abs{\mathcal{E}_{qLL}\sbra{\dfrac{d(q+r)}{N}\rho} - \mathcal{E}_{qLL}\sbra{\rho}} \\
                \le& \abs{1 - \dfrac{d(q+r)}{N}}\norm{V}_{L^2}\norm{\rho}_{L^2} + \prth{1 - \prth{\dfrac{d(q+r)}{N}}^2}\norm{w}_{L^2}\norm{\rho}_{L^2}^2
        \end{align*}
        and
        \begin{align*}
            \abs{\prth{1 - \prth{\dfrac{d(q+r)}{N}}^2}} \le C\abs{1 - \dfrac{d(q+r)}{N}}
        \end{align*}
    \end{proof}

\subsection{Energy lower bound}\label{sec:LOW E bound}
    
    In this part we prove the Energy lower bound:
    
    \begin{proposition}[title=Lower bound, label=ELB]
    Let $(\psi_N)_N$ be a sequence of minimizers of \cref{eq:E_N^0},
        \begin{align*}
            \mathcal{E}_{sc, \hbar b}\sbra{m_{\psi_N}} \ge \hbar b E^{q,r} + E_{V}^{q,r} + E_{w}^{q,r} + \mathcal{E}_{qLL}^0 + o(1)
        \end{align*}
    \end{proposition}
    
    Husimi functions are symmetric and consistent measures. The De Finetti \cref{th:DeFinetti} states that such measures are reduced to trivial measure of this kind, namely tensorized products of one body measures and their convex combinations. This result plays an important role in the control of correlations for the lower bound. We start by extracting some limit Husimi functions and give their fundamental properties. Similar arguments can be found in \cite[Section 2]{FLS}. With \cref{not:Husimi},

    \begin{proposition}[label=limit densities & H]
        Let $(\psi_N)_N$ be a sequence of minimizers of \cref{eq:E_N^0}, up to a subsequence
        \begin{enumerate}[wide, labelindent=0pt, label=\alph*)]
            \item 
                there exists limit Husimi functions $M^{(k)} \in L^\infty \prth{(\mathbb{N}\times\Omega)^k}$ such that
                \begin{align*}
                    &m_{\psi_N}^{(k)} \wslim\displaylimits_{N\to\infty} M^{(k)} \text{ in the weak star topology on } L^\infty\prth{\prth{\mathbb{N}\times\Omega}^k} \yesnumber{conv Husimi} \\
                    &0 \le M^{(k)} \le \dfrac{1}{\prth{L^2(q+r)}^k} \yesnumber{limit H bound}
                \end{align*}
            \item
            $M^{(1)}(q, \bullet) \in \mathcal{D}_{qLL}$ and
                    \begin{align*}
                        M^{(1)}(n, \bullet) = \mathbb{1}_{n < q}\dfrac{1}{L^2(q+r)} + \mathbb{1}_{n = q} M^{(1)}(q, \bullet) `\yesnumber{M1}
                    \end{align*}
            \item
                $M^{(k)}$ are the reduced densities of a symmetric measure $M$ on $\prth{\mathbb{N}\times\Omega}^\mathbb{N}$ and $\norm{M^{(k)}}_{L^1} = 1$
            \item
                in the sense of Radon measures 
                \begin{align*}
                    \rho_{m_{\psi_N}}^{(k)} \wslim\displaylimits_{N\to\infty} \rho_{M^{(k)}} \yesnumber{Radon m}
                \end{align*}
            \item
                we have convergence of the potential terms:
                \begin{align*}
                    \mathcal{E}_{qLL}\sbra{\rho_{m_{\psi_N}}} \limit\displaylimits_{N\to\infty} \mathcal{E}_{qLL}\sbra{\rho_{M}} \yesnumber{to rhoM}
                \end{align*}
        \end{enumerate}
    \end{proposition}
    
    \begin{proof}
        \begin{enumerate}[wide, labelindent=0pt, label=\alph*)]
            \item 
                From inequality \cref{eq:Husimi k bound} the Husimi functions are uniformly bounded, with a diagonal extraction we obtain \cref{eq:conv Husimi} and the bound \cref{eq:Husimi k bound} with \cref{eq:l_b scaling} induce \cref{eq:limit H bound} in the limit.
            \item
                Now since we took a minimizer of the energy, with the upper bound \cref{prop:EUB} and the Kinetic energy inequalities \cref{eq:Kin E ineq} and \cref{eq:Kin E ineq w},
                \begin{align*}
                    \dfrac{E_N^0}{N} 
                    =& \Tr{\mathscr{L}_{\hbar, b} \gamma_{\psi_N}^{(1)}} + \intr_\Omega V\rho_{\psi_N}^{(1)} + \intr_{\Omega^2} w\rho_{\psi_N}^{(2)}
                    = \Tr{\mathscr{L}_{\hbar, b} \gamma_{\psi_N}^{(1)}}\prth{1 + \mathcal{O}\prth{\dfrac{1}{\hbar b}}}\\
                    \le& \mathcal{E}_{sc, \hbar b}\sbra{m_\rho} + \hbar b\mathcal{O}\prth{1 - \dfrac{d(q+r)}{N}} + \mathcal{O}\prth{f(\lambda)} + \mathcal{O}\prth{\hbar b \lambda l_b}
                \end{align*}
                so by \cref{eq:sat LLL} we know that
                \begin{align*}
                    \Tr{\mathscr{L}_{\hbar, b} \gamma_{\psi_N}^{(1)}} = \mathcal{O}\prth{\hbar b} \yesnumber{controlled kin E}
                \end{align*}
                Since the contribution of the potential are bounded, the only thing we have to look at are the kinetic terms. Let $m_\rho$ be the Husimi function with saturated low Landau levels defined here \cref{eq:saturated low LL}. We denote
                \begin{align*}
                    c_{N,n} \coloneq \intr_\Omega\prth{m_{\psi_N}^{(1)}(n, .) - m_\rho(n, .)} 
                \end{align*}
                By definition of $m_\rho$ and \cref{lem:equiv Husimi admissible} we have
                \begin{align*}
                    &\sum_{n\in\mathbb{N}} c_{N,n} = \intr_{\mathbb{N}\times\Omega} m_{\psi_N}^{(1)} - \intr_{\mathbb{N}\times\Omega} m_\rho = 1 - 1 = 0 \\
                    &n < q \implies c_{N,n} \le \dfrac{L^2}{2\pi l_b^2 N} + \mathcal{O}\prth{l_b} - \dfrac{1}{q+r} = \mathcal{O}\prth{l_b} + \mathcal{O}\prth{1 - \dfrac{d(q+r)}{N}}\\
                    & n > q \implies c_{N,n} = \norm{ m_{\psi_N}^{(1)}(n, \bullet)}_{L^1} \ge 0
                \end{align*}
                Since $(E_n)_n$ is increasing
                \begin{align*}
                    \sum_{n\in\mathbb{N}} E_n c_{N,n} 
                    \ge \sum_{n=0}^q E_n c_{N,n}  + E_q \sum_{n>q} c_{N,n} 
                    =& -\sum_{n=0}^{q-1} (E_q - E_n) c_{N, n} \\
                    \ge& \mathcal{O}\prth{\hbar b l_b} + \hbar b \mathcal{O}\prth{1 - \dfrac{d(q+r)}{N}} \yesnumber{cn}
                \end{align*}
                Now we compute
                \begin{align*}
                    \mathcal{E}_{sc, \hbar b}\sbra{m_{\psi_N}} -  \mathcal{E}_{sc, \hbar b}\sbra{m_\rho} 
                    =& \sum_{n\in\mathbb{N}} E_n c_{N,n} + \intr_{\mathbb{N}\times\Omega}  V \prth{m_{\psi_N}^{(1)} - m_\rho}d\eta \\
                    &+  \intr_{\prth{\mathbb{N}\times\Omega}^2}  w \prth{m_{\psi_N}^{(2)} - m_\rho^{\otimes 2}}d\eta^{\otimes 2} \yesnumber{expression for lowerbound}
                \end{align*}
                From the semi-classical approximation (\cref{prop:SC energy}), \cref{eq:controlled kin E} and the upper bound (\cref{prop:EUB}),
                \begin{align*}
                    \dfrac{E_N^0}{N} 
                    =& \dfrac{\bk{\psi_N}{\mathcal{H}_N\psi_N}}{N}
                    = \mathcal{E}_{sc, \hbar b}\sbra{m_{\psi_N}} + \mathcal{O}\prth{f(\lambda)} + \mathcal{O}\prth{(\hbar \lambda)^2} \\
                    \le& \mathcal{E}_{sc, \hbar b}\sbra{m_\rho} + \hbar b\mathcal{O}\prth{1 - \dfrac{d(q+r)}{N}} + \mathcal{O}\prth{f(\lambda)} + \mathcal{O}\prth{\hbar b \lambda l_b}
                \end{align*}
                so with \cref{eq:lambda},
                \begin{align*}
                    \mathcal{E}_{sc, \hbar b}\sbra{m_{\psi_N}} -  \mathcal{E}_{sc, \hbar b}\sbra{m_\rho} \le \hbar b\mathcal{O}\prth{1 - \dfrac{d(q+r)}{N}} + \mathcal{O}\prth{f(\lambda)} + \mathcal{O}\prth{\hbar b \lambda l_b} \yesnumber{ineq Encn}
                \end{align*}
                All the potential terms in \cref{eq:expression for lowerbound} are of order $1$, therefore the sum in \cref{eq:cn} is bounded and we have
                \begin{align*}
                    \mathcal{O}\prth{\hbar b l_b} + \hbar b \mathcal{O}\prth{1 - \dfrac{d(q+r)}{N}} \le - \sum_{n=0}^{q-1} (E_q - E_n) c_{N,n} \le \sum_{n \in\mathbb{N}}E_nc_{N,n}\le C 
                \end{align*}
                So
                \begin{align*}
                    \sum_{n=0}^{q-1} \dfrac{E_n - E_q}{\hbar b} c_{N,n} = \mathcal{O}\prth{\dfrac{1}{\hbar b}} \yesnumber{order cNn}
                \end{align*}
                With a similar inequality as \cref{eq:cn} but with $E_{q+1}$ instead of $E_q$ we deduce
                \begin{align*}
                    C &\ge \sum_{n \in\mathbb{N}}E_nc_{N,n} \ge \sum_{n=0}^q E_n c_{N,n}  + E_{q+1} \sum_{n>q} c_{N,n} = \sum_{n=0}^{q} (E_n - E_{q+1}) c_{N,n}\\
                    &\ge \sum_{n=0}^q E_n c_{N,n}  + E_{q} \sum_{n>q} c_{N,n} \ge \mathcal{O}\prth{\hbar b l_b} + \hbar b \mathcal{O}\prth{1 - \dfrac{d(q+r)}{N}} \yesnumber{positivity kinetic diff}
                \end{align*}
                and therefore \cref{eq:order cNn} implies
                \begin{align*}
                    c_{N,q} = \mathcal{O}\prth{\dfrac{1}{\hbar b}}
                \end{align*}
                Then
                \begin{align*}
                    \sum_{n>q} \dfrac{E_n}{\hbar b} c_{N,n} 
                    = \sum_{n\in\mathbb{N}} \dfrac{E_n}{\hbar b} c_{N,n} -  \sum_{n=0}^q \dfrac{E_n}{\hbar b} c_{N,n} = \mathcal{O}\prth{\dfrac{1}{\hbar b}} \ge \sum_{n>q} \intr_{\Omega} m_{\psi_N}^{(1)}(n, \bullet)
                \end{align*}
                and
                \begin{align*}
                    c_{N,q} = \intr_\Omega m_{N}^{(1)}(q, R)dR -  \intr_\Omega \rho(R)dR =  \norm{m_{N}^{(1)}(q, \bullet)}_{L^1} - \dfrac{r}{q+r} =  \mathcal{O}\prth{\dfrac{1}{\hbar b}} \yesnumber{Mq norm}
                \end{align*}
                From the consistency of $m_{\psi_N}^{(k)}$ in \cref{prty:sym measure},
                \begin{align*}
                    \norm{m_{\psi_N}^{(1)}(n_1, \bullet)}_{L^1} 
                    =& \intr_\Omega \prth{\intr_{(\mathbb{N}\times \Omega)^{k-1}} m_{\psi_N}^{(k)}(n_1, x_1; X_{2:k}) d\eta^{\otimes(k-1)}(X_{2:k})}dx_1 \\
                    =& \sum_{n_{2:k}\in \mathbb{N}^{k-1}} \norm{m_{\psi_N}^{(k)}(n_{1:k}, \bullet)}_{L^1} \yesnumber{m1 to mk}
                \end{align*}
                Since
                \begin{align*}
                    \mathbb{N}^k \backslash \intint{0:q}^k = \bigsqcup_{i=1}^k \mathbb{N}^{i-1} \times \prth{\mathbb{N}\backslash\intint{0:q}} \times \mathbb{N}^{k-i}
                \end{align*}
                by the symmetry of $m_{\psi_N}^{(k)}$, \cref{eq:m1 to mk} and \cref{eq:order cNn},
                \begin{align*}
                    \sum_{n_{1:k} \in \mathbb{N}^k \backslash \intint{0:q}^k} \norm{m_{\psi_N}^{(k)}(n_{1:k}, \bullet)}_{L^1} = k \sum_{n_1 > q} \norm{m_{\psi_N}^{(1)}(n_1, \bullet)}_{L^1} = \mathcal{O}\prth{\dfrac{1}{\hbar b}} \yesnumber{over q mk}
                \end{align*}
                $\Omega$ is bounded, thus testing \cref{eq:conv Husimi} against $\mathbb{1}_{\sett{n_{1:k}}\times\Omega} \in L^1\prth{\prth{\mathbb{N}\times\Omega}^k}$,
                \begin{align*}
                    \norm{m_{\psi_N}^{(k)}(n_{1:k}; \bullet)}_{L^1} \limit\displaylimits_{N\to\infty} \norm{M^{(k)}(n_{1:k}; \bullet)}_{L^1}
                \end{align*}
                So \cref{eq:Mq norm} gives
                \begin{align*}
                    \norm{M^{(1)}(q, \bullet)}_{L^1} = \dfrac{r}{q+r}
                \end{align*}
                and with \cref{eq:over q mk}, if $n_{1:k} \in \mathbb{N}^k \backslash \intint{0:q}^k$, then $M^{(k)}(n_{1:k}, \bullet) = 0$ and we see that the norm \cref{eq:Husimi norm} passes to the limit:
                \begin{align*}
                    \norm{M^{(k)}}_{L^1} 
                    =& \sum_{n_{1:k}\in\intint{0:q}^k} \norm{M^{(k)}(n_{1:k}, \bullet)}_{L^1} 
                    = \lim_{N\to\infty} \sum_{n_{1:k}\in\intint{0:q}^k} \norm{m_{\psi_N}^{(k)}(n_{1:k}, \bullet)}_{L^1} \\
                    =& \lim_{N\to\infty} \prth{\sum_{n_{1:k}\in\intint{0:q}^k} \norm{m_{\psi_N}^{(k)}(n_{1:k}, \bullet)}_{L^1} + \sum_{n_{1:k} \in \mathbb{N}^k \backslash \intint{0:q}^k} \norm{m_{\psi_N}^{(k)}(n_{1:k}, \bullet)}_{L^1}} \\
                    =& \lim_{N\to\infty}  \norm{m_{\psi_N}^{(k)}}_{L^1}
                    = 1
                \end{align*}
                If $n < 0$, by \cref{eq:order cNn},
                \begin{align*}
                    \norm{m_{\psi_N}^{(1)}(n ,\bullet) - \dfrac{1}{L^2(q+r)}}_{L^1} 
                    \le& \norm{m_{\psi_N}^{(1)}(n ,\bullet) - \dfrac{1}{2\pi l_b^2 N}}_{L^1} +  \mathcal{O}\prth{1 - \dfrac{d(q+r)}{N}} \\
                    =& \intr_\Omega\prth{\dfrac{1}{2\pi l_b^2 N} - m_{\psi_N}^{(1)}(n ,\bullet)} + \mathcal{O}\prth{1 - \dfrac{d(q+r)}{N}}\\ 
                    =& \intr_\Omega\prth{\dfrac{1}{L^2(q+r)} - m_{\psi_N}^{(1)}(n ,\bullet)} + \mathcal{O}\prth{1 - \dfrac{d(q+r)}{N}} \\
                    =& -C_{N, n} + \mathcal{O}\prth{1 - \dfrac{d(q+r)}{N}}
                    = \mathcal{O}\prth{\dfrac{1}{\hbar b}} +  \mathcal{O}\prth{1 - \dfrac{d(q+r)}{N}} 
                \end{align*}
                so $M^{(1)}(n, \bullet) = \dfrac{1}{L^2(q+r)}$.
            \item
                Testing \cref{eq:coherent Husimi} against $\varphi_q \in C_c^0\prth{\prth{\mathbb{N}\times\Omega}^q}$, we have
                \begin{align*}
                    \intr_{\prth{\mathbb{N}\times\Omega}^q} \varphi_q m_{\psi_N}^{(q)}d\eta^{\otimes q} 
                    = \intr_{\prth{\mathbb{N}\times \Omega}^k}  \varphi_q(X_{1:q}) m_{\psi_N}^{(k)}(X_{1:k}) d\eta^{\otimes k}(X_{1:k}) \yesnumber{limit consistency}
                \end{align*}
                Since $\varphi_q \in L^1\prth{\prth{\mathbb{N}\times\Omega}^k}$, with \cref{eq:conv Husimi},
                \begin{align*}
                    \intr_{\prth{\mathbb{N}\times\Omega}^q} \varphi_q m_{\psi_N}^{(q)}d\eta^{\otimes q} \limit\displaylimits_{N\to\infty} \intr_{\prth{\mathbb{N}\times\Omega}^q} \varphi_q M^{(q)}d\eta^{\otimes q} \yesnumber{cons lim 1}
                \end{align*}
                In order to pass to the limit in the right term of \cref{eq:limit consistency}, for the low Landau levels we use \cref{eq:conv Husimi} on \begin{align*}
                    \mathbb{1}_{\prth{\intint{0:q}\times\mathbb{N}}^k} \prth{\varphi_q \otimes \text{Id}_{\prth{\mathbb{N}\times \Omega}^{k-q}}} \in  L^1\prth{\prth{\mathbb{N}\times\Omega}^k}
                \end{align*}
                and for the high Landau levels we use \cref{eq:over q mk} and $\varphi_q \in L^\infty\prth{\prth{\mathbb{N}\times\Omega}^k}$:
                \begin{align*}
                    \intr_{\prth{\mathbb{N}\times \Omega}^k}  \varphi_q(X_{1:q}) m_{\psi_N}^{(k)}(X_{1:k}) d\eta^{\otimes k}(X_{1:k})
                    =& \intr_{\Omega^k} \mathbb{1}_{\prth{\intint{0:q}\times\mathbb{N}}^k} \prth{\varphi_q \otimes \text{Id}_{\prth{\mathbb{N}\times \Omega}^{k-q}}}m_{\psi_N}^{(k)} d\eta^{\otimes k} \\
                    &+ \sum_{n_{1:k}\in \mathbb{N^k}\backslash\intint{0:q}^k} \intr_{\Omega^k} \varphi_q(n_{1:q}, x_{1:q})m_{\psi_N}^{(k)}(n_{1:k}, x_{1:k}) dx_{1:k} \\
                    \limit\displaylimits_{N\to\infty}& \intr_{\Omega^k} \mathbb{1}_{\prth{\intint{0:q}\times\mathbb{N}}^k} \prth{\varphi_q \otimes \text{Id}_{\prth{\mathbb{N}\times \Omega}^{k-q}}}M^{(k)} d\eta^{\otimes k} \\
                    =& \intr_{\prth{\mathbb{N}\times \Omega}^k} \varphi_q(X_{1:q}) M^{(k)}(X_{1:k}) d\eta^{\otimes k}(X_{1:k}) \yesnumber{cons lim 2}
                \end{align*}
                Thus passing to the limit in \cref{eq:limit consistency} and inserting \cref{eq:cons lim 1} and\cref{eq:cons lim 2} we obtain 
                \begin{align*}
                    \forall \varphi_q \in C_c^0\prth{\prth{\mathbb{N}\times\Omega}^q}, \intr_{\prth{\mathbb{N}\times\Omega}^q} \varphi_q M^{(q)}d\eta^{\otimes q} = \intr_{\prth{\mathbb{N}\times \Omega}^k} \varphi_q(X_{1:q}) M^{(k)}(X_{1:k}) d\eta^{\otimes k}(X_{1:k})
                \end{align*}
                and this proves that the limit Husimi functions are also consistent. Testing against $\varphi_q$, we also obtain that the symmetry of Husimi functions passes to the limit. Then we can conclude with the Kolmogorov extension theorem that there exists $M$ a symmetric measure on $\prth{\mathbb{N}\times\Omega}^\mathbb{N}$ whose marginals are $(M^{(k)})_k$.
            \item
                Let $\varphi_k \in C^0(\Omega^k)$, $\varphi_k$ is bounded and
                \begin{align*}
                    \mathbb{1}_{\intint{0:q}^k} \otimes \varphi_k \in L^1\prth{\prth{\mathbb{N}\times\Omega}^k}
                \end{align*}
                so using \cref{eq:conv Husimi} and \cref{eq:over q mk}
                \begin{align*}
                    \intr_{\Omega^k} \varphi_k \rho_{m_{\psi_N}}^{(k)}
                    =& \intr_{\prth{\mathbb{N}\times \Omega}^k} \prth{\mathbb{1}_{\intint{0:q}^k} \otimes \varphi_k} m_{\psi_N}^{(k)} d\eta^{\otimes k}
                    + \sum_{n_{1:k}\in\mathbb{N}^k\backslash\intint{0:q}^k} \intr_{\Omega^k} \varphi_k m_{\psi_N}^{(k)} \\
                    \limit\displaylimits_{N\to\infty}& \intr_{\prth{\mathbb{N}\times \Omega}^k} \prth{\mathbb{1}_{\intint{0:q}^k} \otimes \varphi_k} M^{(k)} d\eta^{\otimes k} 
                    = \intr_{\Omega^k} \varphi_k \rho_M^{(k)}
                \end{align*}
            \item
                Let $V_k \in L^2(\Omega^k)$, and $(V_{k, n})_n \subset C^\infty(\Omega^k)$ a sequence regularised with a convolution to a regular function so that
                \begin{align*}
                    \norm{V_k - V_{k, n}}_{L^2} \limit\displaylimits_{n\to\infty} 0
                \end{align*}
                we have
                \begin{align*}
                    \intr_{\Omega^k} V_k\prth{\rho_{m_{\psi_N}}^{(k)} - \rho_M^{(k)}} 
                    = \intr_{\Omega^k} V_{k, n} \prth{\rho_{m_{\psi_N}}^{(k)} - \rho_M^{(k)}} 
                        + \intr_{\Omega^k} \rho_{m_{\psi_N}}^{(k)} \prth{V_k - V_{k, n}}
                        + \intr_{\Omega^k} \rho_M^{(k)}\prth{V_{k, n} - V_k}
                \end{align*}
                For a fixed $n$, since $V_{k, n} \in C^0(\Omega^k)$ by \cref{eq:Radon m} the first term goes to $0$ when $N\to\infty$. For the second term we use \cref{eq:Kin E ineq} if $V_1 = V$, \cref{eq:Kin E ineq w} if $V_2 = w$ and \cref{eq:husimi to density}
                \begin{align*}
                    \abs{\intr_{\Omega^k} \rho_{m_{\psi_N}}^{(k)} \prth{V_k - V_{k, n}}} 
                    =& \abs{\intr_{\Omega^2} \prth{\prth{g_\lambda^k}^{\otimes k}*\rho_N^{(k)}} \prth{V_k - V_{k, n}}} 
                    \le C \norm{\prth{V_k - V_{k, n}} * \prth{g_\lambda^2}^{\otimes k}}_{L^2}  \\
                    \le& C\norm{\prth{V_k - V_{k, n}}}_{L^2}
                \end{align*}
                For the third term we use Hölder's inequality since $\rho_M^{(k)} \in L^\infty\prth{\Omega^k}$ so we have
                \begin{align*}
                    \lim_{N\to\infty}\abs{\intr_{\Omega^k} V_k\prth{\rho_{m_{\psi_N}}^{(k)} - \rho_M^{(k)}}}
                    \le C \norm{V_k - V_{n, k}}_{L^2} \limit\displaylimits_{n\to\infty} 0
                \end{align*}
        \end{enumerate}
    \end{proof}

    Now we want to apply the De Finetti theorem to $M$:

    \begin{theorem}[title=De Finetti or Hewitt-Savage, label=DeFinetti]
        Let $X$ be a metric space, $\mu \in \mathcal{P}_s(X^\mathbb{N})$ be a symmetric probability measure with marginals $\prth{\mu^{(n)}}_{n \ge 1}$.
        
        $\exists P_{\mu} \in \mathcal{P}(\mathcal{P}(X)) $ such that:
        \begin{equation}
            \forall n\in \mathbb{N}^*, \mu^{(n)} = \intr_{\mathcal{P}(\Omega)}  \rho^{\otimes n}dP_{\mu}(\rho)
        \end{equation}
    \end{theorem}

    For a proof of this via the the Diaconis-Freedman theorem see \cite[Section 2.1.]{Rougerie} and for some further context one can look at \cite[Section 2.2.]{Rougerie20}.
    
    Recalling the definition of the semi-classical domain \cref{eq:sc domain}, we obtain:
    
    \begin{proposition}[title=Low Landau level filling of the limit factorised densities, label=DF]
        There exists $\mathcal{P}_M \in \mathcal{P}\prth{\mathcal{D}_{sc}}$ such that
        \begin{align*}
            \forall k\in \mathbb{N}^*, M^{(k)} = \intr_{\mathcal{D}_{sc}} m^{\otimes k} d\mathcal{P}_M(m) \yesnumber{PM De Finetti}
        \end{align*}
        Let $\mu$ be the push-forward measure of $\mathcal{P}_M$ by the application
        \begin{align*}
            \matrx{
                L^1\prth{\mathbb{N}\times\Omega} &\to &L^1(\Omega) \\
                m &\mapsto &m(q, \bullet)
            }
        \end{align*}
        then $\mu \in \mathcal{P}\prth{\mathcal{D}_{qLL}}$ and
        \begin{align*}
            &\rho_M^{(k)} = \intr_{\mathcal{D}_{qLL}} \prth{\dfrac{q}{L^2(q+r)} + \rho}^{\otimes k}  d\mu (\rho) \yesnumber{limit rhoM} \\
            &\mathcal{E}_{qLL}\sbra{\rho_M} =  \intr_{\mathcal{D}_{qLL}} \mathcal{E}_{qLL}\sbra{\dfrac{q}{L^2(q+r)} + \rho}  d\mu (\rho) = E_V^{q, r} +  E_w^{q, r} + \intr_{\mathcal{D}_{qLL}} \mathcal{E}_{qLL}\sbra{\rho} d\mu (\rho) \yesnumber{EqLL rhoM}
        \end{align*}
    \end{proposition}
    
    \begin{proof}
        From \cref{th:DeFinetti} applied to $M$ in \cref{prop:limit densities & H}, $\exists \mathcal{P}_M \in \mathcal{P}\prth{\mathcal{P}\prth{\mathbb{N}\times\Omega}}$ such that
        \begin{align*}
            \forall k\in \mathbb{N}^*, M^{(k)} = \intr_{\mathcal{P}\prth{\mathbb{N}\times\Omega}} m^{\otimes k} d\mathcal{P}_M(m) \yesnumber{pre Defin}
        \end{align*}
        Let $\varphi \in C_c^0(\mathbb{N}\times \Omega, \mathbb{R}_+), \varphi \neq 0, \epsilon > 0, k \in \mathbb{N}^*$, and
        \begin{align*}
            A_{\epsilon}(\varphi) \coloneq \sett{m \in \mathcal{P}\prth{\mathbb{N}\times \Omega}| \intr_{\mathbb{N}\times\Omega} \varphi dm \ge \dfrac{1+\epsilon}{L^2(q+r)}\intr_{\mathbb{N}\times\Omega} \varphi}
        \end{align*}
        If $m \in A_\epsilon(\varphi)$, then
        \begin{align*}
            1 \le \dfrac{L^2(q+r)}{(1+\epsilon)\norm{\varphi}_{L^1(\eta)}} \intr_{\mathbb{N}\times\Omega} \varphi dm\le \prth{\dfrac{L^2(q+r)}{(1+\epsilon)\norm{\varphi}_{L^1(\eta)}}\intr_{\mathbb{N}\times\Omega} \varphi dm}^k
        \end{align*}
        so with \cref{eq:limit H bound},
        \begin{align*}
            \mathcal{P}_M \prth{A_{\epsilon}(\varphi)} 
            =& \intr_{\mathcal{P}\prth{\mathbb{N}\times\Omega}} \mathbb{1}_{A_{\epsilon}(\varphi)} d\mathcal{P}_M
            \le \intr_{\mathcal{P}\prth{\mathbb{N}\times\Omega}} \prth{\dfrac{L^2(q+r)}{(1+\epsilon)\norm{\varphi}_{L^1(\eta)}} \intr_{\mathbb{N}\times\Omega} \varphi dm }^k d\mathcal{P}_M(m)\\
            =& \prth{\dfrac{L^2(q+r)}{(1+\epsilon)\norm{\varphi}_{L^1(\eta)}}}^k \intr_{\mathcal{P}\prth{\mathbb{N}\times\Omega}} \prth{ \intr_{\prth{\mathbb{N}\times\Omega}^k} \varphi^{\otimes k} dm^{\otimes k} } d\mathcal{P}_M(m) \\
            =& \prth{\dfrac{L^2(q+r)}{(1+\epsilon)\norm{\varphi}_{L^1(\eta)}}}^k \intr_{\prth{\mathbb{N}\times\Omega}^k} \varphi^{\otimes k} dM^{(k)}
            \le \prth{\dfrac{1}{1 + \epsilon}}^{k} \limit\displaylimits_{k\to\infty} 0
        \end{align*}
        we proved that $\mathcal{P}_M\prth{A_\epsilon(\varphi)} = 0$ and therefore
        \begin{align*}
            \mathcal{P}_M\prth{\bigcap_{\substack{\varphi \in C_c^0\prth{\mathbb{N}\times \Omega, \mathbb{R}_+}\\ \epsilon > 0}}\mathcal{P}\prth{\mathbb{N}\times \Omega} \backslash A_\epsilon(\varphi)} 
            = 1 - \mathcal{P}_M\prth{\bigcup_{\substack{\varphi \in C_c^0\prth{\mathbb{N}\times \Omega, \mathbb{R}_+}\\ \epsilon > 0}}A_\epsilon(\varphi)} 
            = 1
        \end{align*}
        therefore for $\mathcal{P}_m$ almost every $m \in \mathcal{P}\prth{\mathbb{N}\times \Omega}$,
        \begin{align*}
            \forall \varphi \in C_c^0\prth{\mathbb{N}\times \Omega, \mathbb{R}_+}, \epsilon > 0, \intr_{\mathbb{N}\times\Omega} \varphi dm < \dfrac{1+\epsilon}{L^2(q+r)}\intr_{\mathbb{N}\times\Omega} \varphi \yesnumber{bound PM ae}
        \end{align*}
        So for $\mathcal{P}_m$ almost every $m \in \mathcal{P}\prth{\mathbb{N}\times \Omega}$, $m$ is the density of a probability measure thus a positive function such that $\norm{m}_{L^1} =1$ and by \cref{eq:bound PM ae}, $m \in L^\infty(\mathbb{N}\times\Omega)$ and
        \begin{align*}
            m \le \dfrac{1}{L^2(q+r)} \yesnumber{bound m ae}
        \end{align*}
        We have shown $\mathcal{P}_M \in \mathcal{P}\prth{\mathcal{D}_{sc}}$, therefore \cref{eq:pre Defin} implies \cref{eq:PM De Finetti}.
        
        Moreover if $n<q$ by \cref{eq:M1},
        \begin{align*}
            \intr_\Omega \dfrac{1}{L^2(q+r)}dx = \intr_{\mathbb{N}\times\Omega} \mathbb{1}_{\{n\}\times \Omega} dM^{(1)} = \intr_{\mathcal{P}\prth{\mathbb{N}\times\Omega}} \prth{\intr_\Omega m(n,x)dx}d\mathcal{P}_M(m)
        \end{align*}
        so
        \begin{align*}
            \intr_{\mathcal{P}\prth{\mathbb{N}\times\Omega}} \prth{\intr_\Omega \prth{\dfrac{1}{L^2(q+r)} - m(n,x)} dx }d\mathcal{P}_M(m) = 0
        \end{align*}
        By \cref{eq:bound m ae} the integrand is positive thus null $\mathcal{P}_M$ almost everywhere, we conclude that for $\mathcal{P}_M$ almost every $m$
        \begin{align*}
            n < q \implies m(n, \bullet) = \dfrac{1}{L^2(q+r)} \yesnumber{m n<q}
        \end{align*}
        If $ n>q$ by \cref{eq:M1},
        \begin{align*}
            0 = \intr_{\mathbb{N}\times\Omega} \mathbb{1}_{\{n\}\times \Omega} dM^{(1)} = \intr_{\mathcal{P}\prth{\mathbb{N}\times\Omega}} \prth{\intr_\Omega m(n,x)dx}d\mathcal{P}_M(m)
        \end{align*}
        Once again by \cref{eq:bound m ae} the right integrand is positive and thus null so for $\mathcal{P}_M$ almost every $m$
        \begin{align*}
            n>q \implies m(n, \bullet) = 0 \yesnumber{m n>q}
        \end{align*}
        Finally if $n = q$, since $m \in \mathcal{P}\prth{\mathbb{N}\times \Omega}$ we conclude using \cref{eq:m n>q} and \cref{eq:m n<q}: for $\mathcal{P}_M$ almost everywhere $m$ 
        \begin{align*}
            \intr_\Omega m(q, \bullet) = \intr_{\mathbb{N}\times\Omega} m - \sum_{n<q} \intr_\Omega m(n, \bullet)  - \sum_{n>q} \intr_\Omega m(n, \bullet) = 1 - \dfrac{q}{q+r} = \dfrac{r}{q+r} \yesnumber{m n=q}
        \end{align*}
        Gathering \cref{eq:bound m ae}, \cref{eq:m n<q}, \cref{eq:m n>q} and \cref{eq:m n=q}, we now know that for $\mathcal{P}_M$ almost every $m$  we have $m(q, \bullet) \in \mathcal{D}_{qLL}$. This means that $\mu \in \mathcal{P}\prth{\mathcal{D}_{qLL}}$.
        
        Finally we compute
        \begin{align*}
            \rho_M^{(k)} 
            =& \sum_{n_{1:k}}M^{(k)}(n_{1:k};\bullet) 
            = \intr_{\mathcal{D}_{sc}} \sum_{n_{1:k}} m^{\otimes k}(n_{1:k}; \bullet)d\mathcal{P}_M(m) 
            =  \intr_{\mathcal{D}_{sc}} \prth{\sum_{n\in\mathbb{N}} m(n; \bullet) }^{\otimes k}d\mathcal{P}_M(m) \\
            =& \intr_{\mathcal{D}_{sc}} \prth{\dfrac{q}{L^2(q+r)} + m(q; \bullet) }^{\otimes k}d\mathcal{P}_M(m) 
            = \intr_{\mathcal{D}_{qLL}} \prth{\dfrac{q}{L^2(q+r)} + \rho}^{\otimes k}  d\mu (\rho) \\
            =& \intr_{\mathcal{D}_{qLL}} \prth{\dfrac{q}{L^2(q+r)} + \rho}^{\otimes k}  d\mu (\rho)
        \end{align*}
        and
        \begin{align*}
            \mathcal{E}_{qLL}\sbra{\rho_M} 
            =& \intr_{\mathcal{D}_{qLL}} \mathcal{E}_{qLL}\sbra{\dfrac{q}{L^2(q+r)} + \rho}  d\mu (\rho) 
            =& \intr_{\mathcal{D}_{qLL}} \prth{ E_V^{q, r} +  E_w^{q, r} + \mathcal{E}_{qLL}\sbra{\rho}} d\mu (\rho) \\
            =&  E_V^{q, r} +  E_w^{q, r} + \intr_{\mathcal{D}_{qLL}} \mathcal{E}_{qLL}\sbra{\rho} d\mu (\rho)
        \end{align*}
    \end{proof}

    Now we are ready for the proof of the lower bound.
    
    \begin{proof}[prop:ELB]
        Let $\rho \in \mathcal{D}_{qLL}$, starting from \cref{eq:expression for lowerbound}, using inequality \cref{eq:positivity kinetic diff} and \cref{eq:sat LLL} we have
        \begin{align*}
            \mathcal{E}_{sc, \hbar b}\sbra{m_{\psi_N}} 
            &\ge\mathcal{E}_{sc, \hbar b}\sbra{m_\rho} + \mathcal{E}_{qLL}\sbra{\rho_{m_{\psi_N}}} - \mathcal{E}_{qLL}\sbra{\rho_{m_\rho}} + \mathcal{O}\prth{\hbar b l_b} + \hbar b \mathcal{O}\prth{1 - \dfrac{d(q+r)}{N}} \\
            =& \hbar b E_{q, r} + \mathcal{E}_{qLL}\sbra{\rho_{m_{\psi_N}}} + \mathcal{O}\prth{\hbar b l_b} + \hbar b \mathcal{O}\prth{1 - \dfrac{d(q+r)}{N}} 
        \end{align*}
         We conclude with \cref{eq:to rhoM} and \cref{eq:EqLL rhoM} and that
        \begin{align*}
            \mathcal{E}_{sc, \hbar b}\sbra{m_{\psi_N}} 
            \ge&  \hbar b E_{q, r} + \mathcal{E}_{qLL}\sbra{\rho_{m_{\psi_N}}} 
                + \mathcal{O}\prth{\hbar b l_b} + \hbar b \mathcal{O}\prth{1 - \dfrac{d(q+r)}{N}} \\
            =& \hbar b E_{q, r} + \mathcal{E}_{qLL}\sbra{\rho_M} + o(1)
            = \hbar b E^{q,r} + E_{V}^{q,r} + E_{w}^{q,r} + \intr_{\mathcal{D}_{qLL}}\mathcal{E}_{qLL}\sbra{\rho}d\mu(\rho) + o(1) \yesnumber{supp nu}\\
            \ge& \hbar b E^{q,r} + E_{V}^{q,r} + E_{w}^{q,r} + \mathcal{E}_{qLL}^0 + o(1)
        \end{align*}
    \end{proof}
    
\subsection{Conclusion}

\begin{proof}[th:main]
    Let $(\psi_N)_N$ be a sequence of minimizers of \cref{eq:E_N^0}, by \cref{eq:SC approx}
    \begin{align*}
        \dfrac{E(N)}{N} = \dfrac{\bk{\psi_N}{\mathscr{H}_N\psi_N}}{N} = \mathcal{E}_{sc, \hbar b}\sbra{m_{\psi_N}} + o(1)
    \end{align*}
    Since the lower bound is true up to a subsequence for which the have \cref{prop:limit densities & H}, for every adherence value of $E(N)/N$ we conclude by gathering \cref{prop:EUB} and \cref{prop:ELB}.
\end{proof}

\begin{proof}[th:density conv]
    With \cref{eq:limit rhoM} and \cref{eq:Radon m} we get
    \begin{align*}
        \rho_{m_{\psi_N}}^{(k)} &\wslim\displaylimits_{N\to\infty} \intr_{\mathcal{D}_{qLL}} \prth{\dfrac{q}{L^2(q+r)} + \rho}^{\otimes k}  d\mu (\rho)
    \end{align*}
    Let $\varphi \in C^\infty\prth{\Omega^k} $ with \cref{eq:husimi to density},
    \begin{align*}
        \intr_{\Omega^k} \varphi \prth{\rho_{m_{\psi_N}}^{(k)} - \rho_{\psi_N}^{(k)}} 
        = \intr_{\Omega^k} \varphi
        \prth{(g^2_\lambda)^{\otimes k} * \rho_{\psi_N}^{(k)} - \rho_{\psi_N}^{(k)}}
        = \intr_{\Omega^k} \rho_{\psi_N}^{(k)}
        \prth{(g^2_\lambda)^{\otimes k} * \varphi - \varphi} \limit\displaylimits_{N\to\infty} 0 \yesnumber{co limits}
    \end{align*}
    by Hölder's inequality since
    \begin{align*}
        \norm{\rho_{\psi_N}^{(k)}}_{L^1} = 1    
    \end{align*}
    and $\varphi$ is Lipschitz. Up to a subsequence $\rho_{\psi_N}^{(k)}$ converges $\forall k \in \mathbb{N}^*$ in the sense of Radon measures. But with \cref{eq:co limits} this limit coincides with the one of $\rho_{m_{\psi_N}}^{(k)}$ so we obtain \cref{eq:th conv rho}.
    
    Moreover by \cref{eq:supp nu} and \cref{prop:EUB}
    \begin{align*}
        \mathcal{E}_{qLL}^0 \ge \intr_{\mathcal{D}_{qLL}}\mathcal{E}_{qLL}\sbra{\rho}d\mu(\rho) + o(1)
    \end{align*}
    thus $\mu$ only gives mass to minimizers of $\mathcal{E}_{qLL}$.
\end{proof}
    \setcounter{equation}{0}
    \renewcommand\yesnumber[1]{\refstepcounter{equation}\tag{A\arabic{equation}}\label{eq:#1}}
    \section{Appendix: technical proofs} \label{sec:appendix}

    \begin{proof}[prop:poisson wave func]
        We start from \cref{eq:psi_nl} expressed in terms of $h_n$:
        \begin{align*}
            \psi_{nl}(z)= \dfrac{c_n}{\sqrt{L l_b}} e^{2i\pi l\frac{x}{L}} \sum_{k\in\mathbb{Z}} h_n \left(\dfrac{1}{l_b}\left[y+kL+l\dfrac{L}{d}\right]\right) e^{2i\pi k d \frac{x}{L}}
        \end{align*}
        Define
        \begin{align*}
            g(u) \coloneq h_n \left(\dfrac{1}{l_b}\left[y+uL+l\dfrac{L}{d}\right]\right)e^{2i\pi d u \frac{x}{L}}
        \end{align*}
        so we have
        \begin{align*}
            \psi_{nl}(z)= \dfrac{c_n}{\sqrt{L l_b}} e^{2i\pi l\frac{x}{L}} \sum_{k\in\mathbb{Z}} g(k)  \yesnumber{phinl g}
        \end{align*}
        in order to apply the Poisson summation formula to $g$. To do so, we compute $\hat{g}$ with a change of variable and equations \cref{eq:flux quant} and \cref{eq:HO wave func transf}:
        \begin{align*}
            \hat{g}(\nu) &= \dfrac{1}{\sqrt{2\pi}}\intr_\mathbb{R} h_n \left(\dfrac{1}{l_b}\left[y+uL+l\dfrac{L}{d}\right]\right)e^{-iu\left(\nu - 2\pi d \frac{x}{L}\right)}du \\
            &= \dfrac{l_b}{L\sqrt{2\pi}}e^{i\prth{\frac{y}{L}+\frac{l}{d}}\left(\nu - 2\pi d \frac{x}{L}\right)}\intr_\mathbb{R} h_n(u)e^{-i\frac{u l_b}{L}\left(\nu - 2\pi d \frac{x}{L}\right)}du \\
            &= \dfrac{l_b}{L\sqrt{2\pi}}e^{i\prth{\frac{y}{L}+\frac{l}{d}}\left(\nu - 2\pi d \frac{x}{L}\right)} \widehat{h_n}\prth{\frac{l_b}{L}\nu - \frac{x}{l_b}}
            =\dfrac{(-i)^n l_b}{L}e^{i\prth{\frac{y}{L}+\frac{l}{d}}\left(\nu - 2\pi d \frac{x}{L}\right)} h_n\prth{\frac{l_b}{L}\nu - \frac{x}{l_b}}
        \end{align*}
        so by using \cref{eq:flux quant} again:
        \begin{align*}
            \hat{g}(2\pi k) &=\dfrac{(-i)^n l_b}{L} e^{i\prth{\frac{y}{L}+\frac{l}{d}}\left(2\pi k - 2\pi d \frac{x}{L}\right)} h_n\left(2\pi k\dfrac{l_b}{L} - \dfrac{x}{l_b} \right) \\
            &=\dfrac{(-i)^n l_b}{L} e^{-i\frac{xy}{l_b^2} - 2i\pi l \frac{x}{L}} e^{2i\pi k \left(\frac{y}{L}+\frac{l}{d}\right)} h_n\prth{\dfrac{1}{l_b}\sbra{k\dfrac{L}{d} - x}}
        \end{align*}
        To conclude the computation we insert this after applying the Poisson summation formula \cref{eq:poisson formula} to \cref{eq:phinl g}:
        \begin{align*}
            \psi_{nl}(z)&=  \dfrac{c_n}{\sqrt{L l_b}} e^{2i\pi l\frac{x}{L}} \sqrt{2\pi}\sum_{k\in \mathbb{Z}}\hat{g}\left(2\pi k\right) \\
            &= \dfrac{c_n}{\sqrt{L l_b}}\cdot\dfrac{\sqrt{2\pi}(-i)^n l_b}{L}e^{-i\frac{xy}{l_b^2} }  \sum_{k\in \mathbb{Z}} H_n \prth{\dfrac{1}{l_b}\sbra{k\dfrac{L}{d} - x}} e^{2i\pi k \left(\frac{y}{L}+\frac{l}{d}\right) -\frac{1}{2l_b^2}\left(k\frac{L}{d} - x \right)^2} \\
            &= \widetilde{c}_n\dfrac{\sqrt{l_b}}{L^{\frac{3}{2}}}e^{-i\frac{xy}{l_b^2}}  \sum_{k\in \mathbb{Z}} H_n \left(\dfrac{1}{l_b}\left[x + k\dfrac{L}{d}\right]\right) e^{-2i\pi k \left(\frac{y}{L}+\frac{l}{d}\right) -\frac{1}{2l_b^2}\left(x + k\frac{L}{d}\right)^2}
        \end{align*}
        by changing the sum index $k$ to $-k$, using the parity of Hermite polynomials and the relation
        \begin{align*}
            c_n \sqrt{2\pi}(-i)^n = \widetilde{c}_n
        \end{align*}
    \end{proof}

    \begin{lemma}[label=series conv]
        Let $m \in \mathbb{N}, c > 0$, the following series are uniformly bounded in $\alpha, a, b$:
        \begin{align*}
            &\forall  \alpha \in \mathbb{R}_+, a, b\in[-1,1], \alpha\sum_{q \in \mathbb{Z}^*}\abs{a + b + \alpha q}^m e^{-c\prth{a + \alpha q}^2} \le C(c, m) \yesnumber{convlem q} \\
            &\forall \alpha \in [0, 1], a \in \mathbb{R},b \in [-1, 1], \alpha\sum_{k \in \mathbb{Z}}\abs{a + b + \alpha k}^m e^{-c\prth{a + \alpha k}^2} \le C(c, m) \yesnumber{convlem k}
        \end{align*}
        Moreover, if $P_n, Q_n$ are complex polynomials of degree $n$, the function
        \begin{align*}
            \Xi(z) \coloneq \sum_{k, q\in\mathbb{Z}} P_n\prth{\dfrac{1}{l_b}\sbra{x+k\dfrac{L}{d}}} Q_n\prth{\dfrac{1}{l_b}\sbra{x+qL+k\dfrac{L}{d}}}e^{2i\pi qd \frac{y}{L} -\frac{1}{2l_b^2}\prth{x+k\frac{L}{d}}^2 -\frac{1}{2l_b^2}\prth{x+ qL + k\frac{L}{d}}^2} 
        \end{align*}
        is of order $\dfrac{1}{l_b}$ and can be uniformly approximated as
        \begin{align*}
            \norm{\Xi(z) - \dfrac{L}{2\pi l_b}\intr_\mathbb{R}P_n(u)Q_n(u)e^{-u^2}du}_{L^\infty} \le C(n) \yesnumber{kernel approx}
        \end{align*}
    \end{lemma}

    \begin{proof}
        Let $\alpha \in \mathbb{R}_+, a, b\in[-1,1]$. If $q \ge 2$ then $q \le 2(q-1)$ so
        \begin{align*}
            \forall u \in [q-1, q), \abs{a+b+\alpha q}^m e^{-c\prth{a + \alpha q}^2} \le (2 +2\alpha u)^m e^{-c(a+\alpha u)^2}
        \end{align*}
        and
        \begin{align*}
            \alpha\sum_{q \ge 2}\abs{a + b + \alpha q}^m e^{-c\prth{a + \alpha q}^2} 
            \le \intr_1^\infty (2+2\alpha u)^m e^{-c(a+\alpha u)^2} \le \intr_\mathbb{R} (2+2u)^m e^{-c(a+u)^2}du \le C(c, m)
        \end{align*}
        the term for $q=1$ is
        \begin{align*}
            \alpha \abs{a+b+\alpha}^m e^{-c(a+\alpha)^2} \le C
        \end{align*}
        for the negative $q$, we see that
        \begin{align*}
            \alpha \sum_{q \le -1}\abs{a + b + \alpha q}^m e^{-c\prth{a + \alpha q}^2} &= \alpha \sum_{q \ge 1}\abs{-a - b + \alpha q}^m e^{-c\prth{-a + \alpha q}^2} \le C(c, m)
        \end{align*}
        because $-a, -b \in [-1, 1]$.
        
        For \cref{eq:convlem k}, let $\alpha \in [0, 1], a \in \mathbb{R},b \in [-1, 1]$. We see that the series is $\alpha$-periodic in $a$ so we can assume $0\le a \le 1$ and use \cref{eq:convlem q} and for $k=0$:
        \begin{align*}
            \alpha \abs{a+b}^me^{-ca^2} \le 2^m
        \end{align*}
        Now we use this result to prove the approximation of $\Xi$. Due to the Gaussian factor, all terms for which $q \neq 0$ have a fair chance to vanish when $l_b \to 0$. Thus, we focus first on the term indexed by $q=0$. To simplify notation we introduce
        \begin{align*}
            &u_{b, x} \coloneq \dfrac{1}{l_b}\prth{x+ u \dfrac{L}{d}} = \dfrac{x}{l_b} + 2\pi u \dfrac{l_b}{L} \yesnumber{ubx} \\
            &\xi(u) \coloneq P_n\prth{u_{b, x}}Q_n\prth{u_{b, x}}e^{-{u_{b, x}}^2} \\
            &\Xi_{|q\neq 0}(z) \coloneq \Xi(z) - \sum_{k\in \mathbb{Z}} \xi(k)
        \end{align*}
        so
        \begin{align*}
            \sum_{k\in \mathbb{Z}} \xi(k)  = \sum_{k\in\mathbb{Z}} P_n\prth{\dfrac{1}{l_b}\sbra{x+k\dfrac{L}{d}}} Q_n\prth{\dfrac{1}{l_b}\sbra{x+k\dfrac{L}{d}}}e^{-\frac{1}{l_b^2}\prth{x+k\frac{L}{d}}^2} 
        \end{align*}
        is the term for $q=0$ and $\Xi_{|q\neq 0}(z)$ contains the other terms.
        
        Note that $\Xi$ is $L/d$-periodic in $x$ so we can choose $x \in\sbra{0,L/d}$ and
        \begin{align*}
            \dfrac{x}{l_b} \le 2\pi \dfrac{l_b}{L}\limit\displaylimits_{N\to\infty} 0 \yesnumber{x bound}
        \end{align*}
        For $q=0$, if we replace the sum in $k$ by the associated integral we obtain:
        \begin{align*}
            \intr_\mathbb{R} \xi(u)du = \dfrac{L}{2\pi l_b}\intr_\mathbb{R}P_n(u)Q_n(u)e^{-u^2}du
        \end{align*}
        which is the approximation in \cref{eq:kernel approx}. For the convergence of the Riemann sum, we compute the derivative of the integrand.
        There exists $R_n$ a polynomial of degree $2n+1$ such that
        \begin{align*}
            \xi'(u) &= 2\pi \dfrac{l_b}{L} R_n\prth{ u_{b, x}}e^{- {u_{b, x}}^2}
        \end{align*}
        Now, use the mean value theorem:
        \begin{align*}
            \abs{\sum_{k\in \mathbb{Z}} \xi(k) - \intr_\mathbb{R}\xi(u)du} \le \sum_{k\in \mathbb{Z}} \intr_k^{k+1} \abs{\xi(k) - \xi(u)}du
            \le 2\pi\dfrac{l_b}{L} \sum_{k\in\mathbb{Z}}  \supm\displaylimits_{k\le u\le k+1} \abs{R_{n}\prth{ u_{b, x}}}e^{- {u_{b, x}}^2} \yesnumber{mean value th}
        \end{align*}
        To control this we only need to control monomials. If $k \le u \le k+1$,
        \begin{align*}
            &\abs{ u_{b, x}}^m e^{-{u_{b, x}}^2} \\
            \le&\abs{ k_{b, x}}^m e^{-{k_{b, x}}^2} + \abs{ \prth{k+1}_{b, x}}^m e^{-{\prth{k+1}_{b, x}}^2} + \abs{ k_{b, x}}^m e^{-{\prth{k+1}_{b, x}}^2} + \abs{ \prth{k+1}_{b, x}}^m e^{-{k_{b, x}}^2} \\
            =& \abs{ k_{b, x}}^m e^{-{k_{b, x}}^2} + \abs{ \prth{k+1}_{b, x}}^m e^{-{\prth{k+1}_{b, x}}^2} + \abs{ \prth{k+1}_{b, x - \frac{L}{d}}}^m e^{-{\prth{k+1}_{b, x}}^2} + \abs{ k_{b, x + \frac{L}{d}}}^m e^{-{k_{b, x}}^2}
        \end{align*}
        Thus after some change of indices,
        \begin{align*}
            2\pi\dfrac{l_b}{L} \sum_{k\in\mathbb{Z}}  \supm\displaylimits_{k\le u\le k+1} &\abs{ u_{b, x}}^m e^{-{u_{b, x}}^2} \le 2\pi \dfrac{l_b}{L} \sum_{k\in\mathbb{Z}} \prth{2\abs{ k_{b, x}}^m + \abs{k_{b, x - \frac{L}{d}}}^m + \abs{k_{b, x + \frac{L}{d}}}^m } e^{-{k_{b, x}}^2}
        \end{align*}
        Using  \cref{eq:convlem k} with $\alpha = 2\pi \dfrac{l_b}{L} \to 0, a =\dfrac{x}{l_b}, b \in \sett{0, 2\pi \dfrac{l_b}{L}, -2\pi \dfrac{l_b}{L}}, c=1$:
        \begin{align*}
            \abs{\sum_{k\in \mathbb{Z}} \xi(k) - \intr_\mathbb{R}\xi(u)du} \le C(n)
        \end{align*}
        We next control $\Xi_{|q\neq 0}$. Let $\epsilon > 0$, with Young's inequality:
        \begin{align*}
            -\prth{\dfrac{x}{l_b} + 2\pi k \dfrac{l_b}{L}}\cdot q\dfrac{L}{l_b} \le \dfrac{\epsilon}{2}\prth{\dfrac{x}{l_b} + 2\pi k \dfrac{l_b}{L}}^2 + \dfrac{1}{2\epsilon}\prth{q\dfrac{L}{l_b}}^2
        \end{align*}
        so
        \begin{align*}
            e^{-\frac{1}{2}\prth{\frac{x}{l_b} + 2\pi k \frac{l_b}{L}}^2 -\frac{1}{2} \prth{\frac{x}{l_b} + 2\pi k \frac{l_b}{L} + q\frac{L}{l_b}}^2} \le e^{-\prth{1-\frac{\epsilon}{2}}\prth{\frac{x}{l_b} + 2\pi k \frac{l_b}{L}}^2 - \prth{\frac{1}{2} - \frac{1}{2\epsilon}}\prth{q\frac{L}{l_b}}^2}
        \end{align*}
        We take $\epsilon = 3/2$. As in \cref{eq:mean value th}, we need to deal with monomial terms of the form
        \begin{align*}
            \sum_{q\neq0, k} \abs{\frac{x}{l_b} + 2\pi k \frac{l_b}{L}}^m \abs{q\frac{L}{l_b}}^{\widetilde{m}} e^{-\frac{1}{4}\prth{\frac{x}{l_b} + 2\pi k \frac{l_b}{L}}^2 -\frac{1}{6}\prth{q\frac{L}{l_b}}^2}
        \end{align*}
        by using
        \begin{itemize}
            \item \cref{eq:convlem q} for the sum in $q$ with $\alpha = \dfrac{L}{l_b}, a=0, b=0, c=\dfrac{1}{6}$
            \item \cref{eq:convlem k} for the sum in $k$ with $\alpha = 2\pi \dfrac{l_b}{L} \to 0, a=\dfrac{x}{l_b} \to 0, b=0, c=\dfrac{1}{4}$
        \end{itemize}
        We conclude that
        \begin{align*}
            \abs{\Xi_{|q\neq 0}} \le C(n) \dfrac{L}{l_b}\cdot\dfrac{l_b}{L} = C(n)
        \end{align*}
    \end{proof}
    
    \begin{proof}[prop:proj conv]
        We start from \cref{eq:Pi_n}:
        \begin{align*}
            \Pi_n(z, z) &= \dfrac{1}{\norm{h_n}_{L^2}^2 Ll_b} \sum_{k, q \in \mathbb{Z}} h_n\prth{\dfrac{1}{l_b}\sbra{x+k\dfrac{L}{d}}} h_n\prth{\dfrac{1}{l_b}\sbra{x+qL+k\dfrac{L}{d}}}e^{2i\pi qd \frac{y}{L}}
        \end{align*}
        We apply \cref{lem:series conv} and thus compute 
        \begin{align*}
                \dfrac{1}{\norm{h_n}_{L^2}^2 Ll_b} \intr_\mathbb{R} h_n\prth{\dfrac{x}{l_b} + 2\pi u\dfrac{l_b}{L}}^2 du = \dfrac{1}{Ll_b}\cdot \dfrac{L}{2\pi l_b} = \dfrac{1}{2\pi l_b^2}
        \end{align*}
        and obtain \cref{eq:kernel approx 1}. Starting again from \cref{eq:Pi_n} and using notation \cref{eq:ubx}, we compute in Landau gauge
        \begin{align*}
            &(\mathscr{P}_{\hbar, b}\Pi_n)(x, y)\\
            =&\prth{\matrx{i\hbar \partial_{x_1} - bx_2 \\ i\hbar\partial_{x_2}}} \dfrac{1}{\norm{h_n}_{L^2}^2 Ll_b} e^{i\frac{y_1y_2 - x_1x_2}{l_b^2}}\sum_{k, q \in \mathbb{Z}} h_n \prth{k_{b, x_1}} h_n \prth{k_{b, y_1 + qL}} \cdot e^{2i\pi k\frac{y_2 - x_2}{L} + 2i\pi d q \frac{y_2}{L}} \\
            =& \dfrac{1}{\norm{h_n}_{L^2}^2 Ll_b} e^{i\frac{y_1y_2 - x_1x_2}{l_b^2}}\sum_{k, q \in \mathbb{Z}} \dfrac{\hbar}{l_b} \prth{\matrx{i h_n'\prth{k_{b, x_1}}\\  k_{b, x_1}h_n\prth{k_{b, x_1}}}}  h_n \prth{k_{b, y_1 + qL}} \cdot e^{2i\pi k\frac{y_2 - x_2}{L} + 2i\pi d q \frac{y_2}{L}}
        \end{align*}
        So 
        \begin{align*}
            (\mathscr{P}_{\hbar, b}\Pi_n)(z, z) =  \dfrac{b}{\norm{h_n}_{L^2}^2 L} \sum_{k, q \in \mathbb{Z}} \prth{\matrx{i h_n'\prth{k_{b, x}}\\  k_{b, x}h_n\prth{k_{b, x}}}}  h_n \prth{k_{b, x + qL}} e^{2i\pi d q \frac{y}{L}}
        \end{align*}
        and with \cref{lem:series conv},
        \begin{align*}
            \norm{(\mathscr{P}_{\hbar, b}\Pi_n)(z, z) - \dfrac{L}{2\pi l_b}\cdot \dfrac{b}{\norm{h_n}_{L^2}^2 L}\intr_\mathbb{R}\prth{\matrx{i h_n'(u)\\  u h_n(u)}}  h_n(u)e^{-u^2}du}_{L^\infty} \le C(n) b
        \end{align*}
    \end{proof}

    \begin{proof}[lem:equiv Husimi admissible]
        $m_{\gamma_k}$ is positive because $\forall X \in \mathbb{N}\times \Omega$, $\Pi_X$ and $\gamma_k$ are positive. With \cref{cor:tr pi},
        \begin{align*}
            m_{\gamma_k}(X_{1:k}) \le& \norm{\gamma_k}_{\mathcal{L}^\infty} \prod_{i=1}^k\Tr{\Pi_{X_i}}
            = \norm{\gamma_k}_{\mathcal{L}^\infty} \prth{\dfrac{1}{2\pi l_b^2} + \mathcal{O}\prth{\dfrac{1}{l_b}}}^k
            = \dfrac{\norm{\gamma_k}_{\mathcal{L}^\infty}}{(2\pi l_b^2)^k}(1 + \mathcal{O}(l_b))
        \end{align*}
        Then, with the resolution of identity \cref{eq:nR res id} we have
        \begin{align*}
            \intr_{(\mathbb{N}\times\Omega)^k} m_{\gamma_k} d\eta^{\otimes k} = \Tr{\gamma_k}
        \end{align*}
        Since $\forall X \in \mathbb{N}\times \Omega$, $\Pi_X$ and $m_k$ are positive $\gamma_{m_k}$ is also positive. \cref{eq:nR res id} also implies
        \begin{align*}
            \gamma_{m_k} \le \prth{2\pi l_b^2}^k\norm{m_k}_{L^\infty} \intr_{(\mathbb{N}\times\Omega)^k} \bigotimes_{i=1}^k\Pi_{X_i} d\eta^{\otimes k}(X_{1:k}) = \prth{2\pi l_b^2}^k\norm{m_k}_{L^\infty}
        \end{align*}
        Finally, using \cref{cor:tr pi},
        \begin{align*}
            \Tr{\gamma_{m_k}} &= (2\pi l_b^2)^k \intr_{(\mathbb{N}\times\Omega)^k} m_k(X_{1:k})\Tr{\bigotimes_{i=1}^k\Pi_{X_i}} d\eta^{\otimes k}(X_{1:k}) = \intr_{(\mathbb{N}\times\Omega)^k} m_k d\eta^{\otimes k} + \mathcal{O}(l_b) \\
            &= \norm{m_k}_{L^1} + \mathcal{O}(l_b)
        \end{align*}
        $\Pi_X \in \mathcal{L}^1(L^2(\Omega))$ is positive, thus it can be diagonalized:
        \begin{align*}
            \Pi_X = \sum_{i\in\mathbb{N}} \lambda_{i, X} \ket{\psi_{i, X}}\bra{\psi_{i ,X}} \text{ with } \lambda_{i, X} \ge 0 \text{ and } \sum_{i\in\mathbb{N}}\lambda_{i, X} = \Tr{\Pi_X}
        \end{align*}
        We have
        \begin{align*}
            m_{\gamma_N^{(k)}}(X_{1:k}) 
            =& \sum_{i_{1:k} \in \mathbb{N}^k} \prth{\prod_{j=1}^k \lambda_{{i_j}, X_j}} \Tr{\gamma_N^{(k)} \ket{\bigotimes_{j=1}^k \psi_{{i_j}, X_j}}\bra{\bigotimes_{j=1}^k \psi_{{i_j}, X_j}}} \\
            =& \sum_{i_{1:k} \in \mathbb{N}^k} \prth{\prod_{j=1}^k \lambda_{{i_j}, X_j}} \Tr{\gamma_N \ket{\bigotimes_{j=1}^k \psi_{{i_j}, X_j}}\bra{\bigotimes_{j=1}^k \psi_{{i_j}, X_j}}\otimes \text{Id}_{L^2\prth{\Omega^{N-k}}}} \yesnumber{better mk}
        \end{align*}
        Let $\psi_{1:N} \in L^2(\Omega)$ be an orthonormal family, we claim that
        \begin{align*}
            \ket{\bigotimes_{i=1}^k \psi_i}\bra{\bigotimes_{i=1}^k \psi_i}\otimes \text{Id}_{L^2\prth{\Omega^{N-k}}} \le \dfrac{(N-k)!}{N!} \text{ on } \mathcal{L}^1\prth{L_-^2\prth{\Omega^N}} \yesnumber{better gamma k}
        \end{align*}
        Indeed, if we consider the Slater determinant
        \begin{align*}
            \chi_N \coloneq \dfrac{1}{\sqrt{N!}}\sum_{\sigma \in S_N} \epsilon(\sigma) \bigotimes_{j=1}^N \psi_{\sigma(j)}
        \end{align*}
        then
        \begin{align*}
            &\bk{\chi_N}{\prth{\ket{\bigotimes_{i=1}^k \psi_i}\bra{\bigotimes_{i=1}^k \psi_i}\otimes \text{Id}_{L^2\prth{\Omega^{N-k}}}} \chi_N} \\
            =& \dfrac{1}{N!} \sum_{\sigma, \tau \in S_N} \epsilon(\sigma \tau) \bk{\bigotimes_{i=1}^k \psi_i}{\bigotimes_{i=1}^k \psi_{\tau(i)}} \bk{\bigotimes_{i=1}^N \psi_{\sigma(i)}}{\prth{\bigotimes_{i=1}^k \psi_i} \otimes \bigotimes_{i=k+1}^N \psi_{\tau(i)}} \\
            =& \dfrac{1}{N!} \sum_{\sigma, \tau \in S_N} \epsilon(\sigma \tau) \prth{\prod_{i=1}^k \delta_{\sigma(i), i} \delta_{\tau(i), i}} \prod_{i=k+1}^N \delta_{\sigma(i), \tau(i)}
            = \dfrac{1}{N!} \sum_{\sigma \in S_N} \prod_{i=1}^k \delta_{\sigma(i), i} 
            = \dfrac{(N-k)!}{N!}
        \end{align*}
        If the Slater determinant does not contain the $\psi_{1:k}$ then the result of this computation is $0$, thus we obtain \cref{eq:better gamma k}. Then with \cref{eq:better mk} and \cref{cor:tr pi},
        \begin{align*}
            m_{\gamma_N^{(k)}}(X_{1:k}) 
            &\le \dfrac{(N-k)!}{N!} \sum_{i_{1:k} \in \mathbb{N}^k} \prth{\prod_{j=1}^k \lambda_{{i_j}, X_j}} \Tr{\gamma_N}
            = \dfrac{(N-k)!}{N!} \Tr{\gamma_N} \prod_{j=1}^k \Tr{\Pi_{X_j}} \\
            &= \dfrac{(N-k)!}{\prth{2\pi l_b^2}^k N!}\Tr{\gamma_N}\prth{1 + \mathcal{O}\prth{l_b}}
        \end{align*}
    \end{proof}

    \begin{proof}[prty:sym measure]
    Let $k > q \ge 1$ and $X_{1:q} \in (\mathbb{N}\times\Omega)^{q}$. Recalling the results of \cref{ssec:red densities}, we prove that the $N$-body Husimi functions are consistent marginals using \cref{eq:nR res id}:
        \begin{align*}
            \intr_{\prth{\mathbb{N}\times\Omega}^{k-q}} m_{\gamma_N}^{(k)}(X_{1:k}) d\eta(X_{q+1:k})
            &= \Tr{\intr_{\prth{\mathbb{N}\times\Omega}^{k-q}} \gamma_N^{(k)}\bigotimes_{i=1}^k\Pi_{X_i}d\eta(X_{q+1:k})} \\
            &= \Tr{\gamma_N^{(k)} \bigotimes_{i=1}^q \Pi_{X_i} \otimes \text{Id}_{\mathbb{N}\times\Omega}^{\otimes(k-q)}}
            = \Tr{\gamma_N^{(q)} \bigotimes_{i=1}^q \Pi_{X_i}} \\
            &= m_{\gamma_N}^{(q)}(X_{1:q}) \yesnumber{coherent Husimi}
        \end{align*}
        The symmetry of the Husimi measures follows from the symmetry of the reduced density matrices and \cref{eq:Husimi k bound} and \cref{eq:Husimi norm} follow from
        \begin{align*}
            \Tr{\gamma^{(k)}_N} =1
        \end{align*}
        and \cref{lem:equiv Husimi admissible}.
        
        For the last point we perform a straightforward computation:
        \begin{align*}
            \sum_{n_{1:k}\ge 0} m_{\gamma_N}^{(k)}\prth{n_{1:k};R_{1:k}}  &= \text{Tr} \sbra{\gamma_N^{(k)} \sum_{n_{1:k}\ge 0}\bigotimes_{i=1}^k\Pi_{n_i, R_i}}
            = \intr_{\Omega^k} \gamma_N^{(k)}(x_{1:k}, x_{1:k}) \prod_{i=1}^k g_\lambda(x_i - R_i)^2 dx_{1:k} \\
            &= (g^2_\lambda)^{\otimes k} * \rho_{\gamma_N}^{(k)} \prth{R_{1:k}}
        \end{align*}
    \end{proof}

    \begin{proof}[lem:corrected Husimi]
        First, by \cref{cor:tr pi}, $\exists \mathcal{E}:\mathbb{N}\times \Omega \to \mathbb{R}$ such that
        \begin{align*}
            &2\pi l_b^2 \Tr{\Pi_X} = 1 + l_b  \mathcal{E}(X) \\
            &\abs{\mathcal{E}(n , R)} \le C(n)
        \end{align*}
        So
        \begin{align*}
            \Tr{\gamma_m} = \intr_{\mathbb{N}\times\Omega} m(X)\prth{1 + l_b \mathcal{E}(X)}d\eta(X)
        \end{align*}
        If $\Tr{\gamma_m} = 1$ then $m$ has the desired properties. If $\Tr{\gamma_m} < 1$ we add some mass to $m$ where it is possible without breaking the Pauli principle. Let $n_1 \in \mathbb{N}$ and
        \begin{align*}
            0 \le \tau \le \dfrac{1}{2\pi l_b^2 N}
        \end{align*}
        we define
        \begin{align*}
            \widetilde{m}(\tau, n_1) \coloneq m + \text{min}\prth{\tau, \dfrac{1}{2\pi l_b^2 N} - m}\mathbb{1}_{n \le n_1}
        \end{align*}
        By construction
        \begin{align*}
            0 \le m \le \widetilde{m}(\tau, n_1) \le \dfrac{1}{2\pi l_b^2 N} \text{ and } \tau \mathbb{1}_{n \le n_1} \le \widetilde{m}(\tau, n_1) \yesnumber{m tild bounded 1}
        \end{align*}
        We choose $n_1 > n_0$ and remark that
        \begin{align*}
            \Tr{\gamma_{\widetilde{m}}}\prth{\dfrac{1}{2\pi l_B^2 N}, n_1} = \dfrac{1}{2\pi l_B^2 N} \intr_{\mathbb{N}\times\Omega} \mathbb{1}_{n \le n_1} \prth{1 + l_b \mathcal{E}(X)}d\eta(X)
            \ge \dfrac{L^2}{2\pi l_b^2 N} n_1 - l_b C(n_1)
        \end{align*}
        Since $\exists n_1 \in \mathbb{N}$ such that 
        \begin{align*}
            \dfrac{L^2}{2\pi l_b^2 N} n_1 > 1
        \end{align*}
        for large enough $N$, 
        \begin{align*}
            \Tr{\gamma_{\widetilde{m}}}\prth{\dfrac{1}{2\pi l_B^2 N}, n_1} > 1
        \end{align*}
        and
        \begin{align*}
            \Tr{\gamma_{\widetilde{m}}}(0, n_1) = \Tr{\gamma_m} < 1
        \end{align*}
        and $\Tr{\gamma_{\widetilde{m}}}$ is Lipschitz in $\tau$, so by the intermediate value theorem $\exists \tau \ge 0$ such that
        \begin{align*}
            \text{with } \widetilde{m} \coloneq \widetilde{m}(\tau, n_1), \Tr{\gamma_{\widetilde{m}}} &= \intr_{\mathbb{N}\times\Omega} \widetilde{m}(X)\prth{1 + l_b \mathcal{E}(X)}d\eta(X) = 1
        \end{align*}
        Thus we can estimate
        \begin{align*}
            \sum_{n \le n_1} \intr_\Omega \text{min}\prth{\tau, \dfrac{1}{2\pi l_b^2 N} - m(n, x)}dx 
            &= \intr_{\mathbb{N}\times \Omega}\prth{\widetilde{m} - m}d\eta 
            = 1 - l_b \intr_{\mathbb{N}\times\Omega} \widetilde{m}\mathcal{E}d\eta - \intr_{\mathbb{N}\times \Omega}m d\eta \\
            &= \mathcal{O}(l_b) + \mathcal{O}\prth{1 - \norm{m}_{L^1}}
        \end{align*}
        so
        \begin{align*}
            \tau = \dfrac{1}{L^2} \intr_\Omega \text{min}\prth{\tau, \dfrac{1}{2\pi l_b^2 N} - m(n_1, x)}dx \le \mathcal{O}(l_b) + \mathcal{O}\prth{1 - \norm{m}_{L^1}}\yesnumber{tau 1}
        \end{align*}
        Now if $\Tr{\gamma_m} > 1$ we remove some mass to $m$:
        \begin{align*}
            \widetilde{m}(\tau) \coloneq \text{max}\prth{0, m - \tau} = m - \text{min}(m, \tau)
        \end{align*}
        by construction
        \begin{align*}
            0\le \widetilde{m} \le m \le \dfrac{1}{2\pi l_b^2 N} \yesnumber{m tild bounded 2}
        \end{align*}
        We see that
        \begin{align*}
            \Tr{\gamma_{\widetilde{m}}}(0) = \Tr{\gamma_m} > 1 \text{ and } \Tr{\gamma_{\widetilde{m}}}\prth{\dfrac{1}{2\pi l_b^2 N}} = 0
        \end{align*}
        so $\exists \tau \ge 0$ such that 
        \begin{align*}
            \text{with }\widetilde{m} \coloneq \widetilde{m}(\tau), 
            \Tr{\gamma_{\widetilde{m}}} = \intr_{\mathbb{N}\times\Omega}\widetilde{m}(X)\prth{1+l_b\mathcal{E}(X)}d\eta(X) = 1
        \end{align*}
        and like before,
        \begin{align*}
            \intr_{\mathbb{N}\times\Omega}\text{min}(m, \tau)d\eta 
            &= \intr_{\mathbb{N}\times\Omega}\prth{m - \widetilde{m}}d\eta 
            = \norm{m}_{L^1} - 1 + l_b \intr_{\mathbb{N}\times\Omega} \widetilde{m}\mathcal{E}d\eta 
            = \mathcal{O}(l_b) + \mathcal{O}\prth{1 - \norm{m}_{L^1}} \\
            &= \intr_{m < \tau}md\eta + \intr_{\tau \le m}\tau d\eta = \norm{m}_{L^1} + \intr_{\tau \le m}(\tau - m) d\eta
        \end{align*}
        So
        \begin{align*}
           \norm{m}_{L^1} + \mathcal{O}(l_b) +  \mathcal{O}\prth{1 - \norm{m}_{L^1}} =  \intr_{\tau \le m}(m - \tau) d\eta \le \dfrac{1}{\pi l_b^2 N}\abs{\mathbb{1}_{\tau \le m}}
        \end{align*}
        and
        \begin{align*}
            \tau \le& \dfrac{1}{\abs{\mathbb{1}_{\tau \le m}}} \intr_{\mathbb{N}\times\Omega}\text{min}(m, \tau)d\eta 
            = \dfrac{1}{\abs{\mathbb{1}_{\tau \le m}}}\prth{\mathcal{O}(l_b) +  \mathcal{O}\prth{1 - \norm{m}_{L^1}}} \\
            \le& \dfrac{1}{\pi l_b^2 N} \cdot \dfrac{\mathcal{O}(l_b) +  \mathcal{O}\prth{1 - \norm{m}_{L^1}}}{\norm{m}_{L^1} + \mathcal{O}(l_b) +  \mathcal{O}\prth{1 - \norm{m}_{L^1}}} = \mathcal{O}(l_b) +  \mathcal{O}\prth{1 - \norm{m}_{L^1}} \yesnumber{tau 2}
        \end{align*}
        With inequalities \cref{eq:tau 1} and \cref{eq:tau 2} we know that
        \begin{align*}
            \norm{m - \widetilde{m}}_{L^\infty} = \mathcal{O}(l_b) + \mathcal{O}\prth{1 - \norm{m}_{L^1}} \yesnumber{m correction}
        \end{align*}
        Finally we prove the estimate on semi-classical energies \cref{eq:SC approx m tilde}:
        \begin{align*}
            \abs{\mathcal{E}_{sc, \hbar b}\sbra{\widetilde{m}} - \mathcal{E}_{sc, \hbar b}\sbra{m}} 
            \le& \sum_{n=0}^{n_1} E_n \intr_{\Omega} \abs{\widetilde{m}(n, \bullet) - m(n, \bullet)} + \sum_{n=0}^{n_1} \intr_{\Omega} \abs{V} \abs{\widetilde{m}(n, \bullet) - m(n, \bullet)} \\
            &+ \sum_{n, \widetilde{n}=0}^{n_1} \intr_{\Omega^2} \abs{w(x-y)} \abs{\widetilde{m}(n, x)\widetilde{m}(\widetilde{n}, y) - m(n, x)m(\widetilde{n}, y)}dxdy \\
            \le& L^2 \sum_{n=0}^{n_1} E_n \norm{m - \widetilde{m}}_{L^\infty} + \prth{n_1 + 1} \norm{V}_{L^1} \norm{m - \widetilde{m}}_{L^\infty} \\
            &+ L^2 \norm{w}_{L^1} \sum_{n, \widetilde{n}=0}^{n_1}  \norm{\widetilde{m}(n, \bullet)\widetilde{m}(\widetilde{n}, \bullet) - m(n, \bullet)m(\widetilde{n}, \bullet)}_{L^\infty}
        \end{align*}
        Moreover
        \begin{align*}
            \norm{\widetilde{m}(n, \bullet)\widetilde{m}(\widetilde{n}, \bullet) - m(n, \bullet)m(\widetilde{n}, \bullet)}_{L^\infty} 
            \le& \norm{\widetilde{m}(n, \bullet)}_{L^\infty} \norm{\widetilde{m}(\widetilde{n}, \bullet) - m(\widetilde{n}, \bullet)}_{L^\infty} \\
            &+ \norm{m(\widetilde{n}, \bullet)}_{L^\infty}\norm{\widetilde{m}(n, \bullet) - m(n, \bullet)}_{L^\infty} \\
            \le& \norm{\widetilde{m}}_{L^\infty} \norm{\widetilde{m} - m}_{L^\infty}  + \norm{m}_{L^\infty}\norm{\widetilde{m} - m}_{L^\infty} \\
        \end{align*}
        so with \cref{eq:m tild bounded 1} and \cref{eq:m tild bounded 2}
        \begin{align*}
            \abs{\mathcal{E}_{sc, \hbar b}\sbra{\widetilde{m}} - \mathcal{E}_{sc, \hbar b}\sbra{m}}  
            \le& \prth{L^2 \sum_{n=0}^{n_1} E_n  + \prth{n_1 + 1} \norm{V}_{L^1} + \dfrac{L^2}{\pi l_b^2 N} \norm{w}_{L^1}\prth{n_1 + 1}^2}  \\
            &\cdot \norm{m - \widetilde{m}}_{L^\infty}
        \end{align*}
        We conclude with \cref{eq:m correction}.
    \end{proof}

    \section*{Acknowledgements}
    \addcontentsline{toc}{section}{Acknowledgements}
    
        I would like to thank my PHD advisor Nicolas Rougerie for discussions, ideas and help that made this work possible. Funding from the European Research Council (ERC) for the project \href{https://cordis.europa.eu/project/id/758620}{Correlated frontiers of many-body quantum mathematics and condensed matter physics (CORFRONMAT No 758620)} is gratefully acknowledged.

    %Bibliography
    \addcontentsline{toc}{section}{Bibliography}
    \printbibliography

\end{document}